\documentclass[final,3p,times]{elsarticle}

\usepackage{amssymb}
\usepackage{amsmath}

\usepackage{amsthm}
\theoremstyle{plain}
\newtheorem{theorem}{Theorem}[section]
\newtheorem{proposition}{Proposition}[section]

\theoremstyle{definition}

\theoremstyle{remark}
\newtheorem{rem}[theorem]{Remark}  

\newcommand\restr[2]{\ensuremath{\left.#1\right|_{#2}}}

\usepackage{algorithm}
\usepackage{algorithmic}
\usepackage{booktabs}
\usepackage{array}
\usepackage{siunitx} 

\journal{Computer Methods in Applied Mechanics and Engineering}

\begin{document}

\begin{frontmatter}



\title{Finite basis physics-informed neural networks with hard constraints for viscous fluid flow in highly perforated domains}


\author[label1]{Jeeeun Lee}
\author[label2]{Denis Korolev}
\author[label3]{Miro Duhovic}
\author[label1]{Seong Su Kim}

\affiliation[label1]{organization={Department of Mechanical Engineering, Korea Advanced Institute of Science and Technology},
             addressline={291 Daehak-ro, Yuseong-gu},
             city={Daejeon},
             postcode={34141},
            country={Republic of Korea}}

\affiliation[label2]{organization={Weierstrass-Institute for Applied Analysis and Stochastics},
            addressline={Anton-Wilhelm-Amo-Str. 39}, 
            city={Berlin},
            postcode={10117}, 
            country={Germany}}

\affiliation[label3]{organization={Leibniz-Institut für Verbundwerkstoffe GmbH},
            addressline={Erwin-Schrödinger-Str. 58}, 
            city={Kaiserslautern},
            postcode={67663}, 
            country={Germany}}


\begin{abstract}
In this work, viscous fluid flow governed by the Stokes equations in highly perforated domains is studied using physics-informed neural networks (PINNs). Perforated microstructures induce complex boundary conditions and fine-scale flow features that are difficult for standard neural networks to resolve. Conventional PINNs, even when combined with advanced training techniques, can suffer from a loss of accuracy and efficiency as the number of perforations increases. One important source of this difficulty is the soft enforcement of boundary conditions through penalty terms, which can lead to stiffness, gradient conflicts, and poor resolution of near-boundary flow structures. Hard constraints provide an alternative by encoding boundary conditions exactly into the network ansatz, but may introduce undesirable non-local effects due to the global nature of the approximation. To address these challenges, finite basis PINNs (FBPINNs), which are based on domain decomposition and localisation principles, are used together with hard boundary constraints that efficiently encode perforation-related boundary conditions. This approach helps mitigate spectral bias, improves overall accuracy, and exhibits convergence that is only weakly affected by the number of perforations, thereby providing an efficient and highly parallelisable neural network framework. The proposed approach is further supported with theoretical arguments, specifically focusing on the localisation and approximation properties of FBPINNs.
\end{abstract}




\end{frontmatter}



\section{Introduction}\label{sec1}

The physics-informed approach to machine learning (ML) is a new paradigm in scientific computing, with rapid growth in techniques forming a new field of scientific machine learning (SciML). Driven by interdisciplinary contributions, it holds great promise for applications to complex technical systems in engineering and mechanics, aiming to provide new modelling approaches and fast simulation surrogates that are less reliant on data, which can often be scarce. The central concept, based on a physics-informed neural network (PINN, pluralised as PINNs), consists of incorporating a given mathematical model in the form of a partial differential equation (PDE) into the learning objective (also called the cost functional or loss), either to complement an existing learning objective driven by data or to serve as the primary source of information about the underlying system of study. The former constitutes a hybrid approach (often used in inverse problems), while the latter serves as a neural network–based PDE solver (also called a forward problem solver). The ability of PINNs to incorporate prior model knowledge, various measurements \cite{zhao2025physics, HANNA2024108019}, surrogate approximations \cite{moon2025physics, korolev2026hybrid}, and constitutive relations \cite{korolev2026hybrid, LEE2025108857, haghighat2023constitutive} into a single framework by modifying the loss, along with the potential for hybridisation with conventional numerical solvers in multi-scale \cite{korolev2026hybrid, hintermuller2026hybrid} and multi-physics \cite{shukla2025neurosem} settings, makes them a very promising tool for assisting with complicated modelling workflows. 

Fluid flow in highly perforated domains (i.e., involving multiple obstacles to flow) arises in many engineering applications across a range of Reynolds numbers (Re) and diameters. Some examples include the flow of resin within fibre tows in composite material manufacturing (Re $\ll$ 1, diameter $\sim$~10~{\textmu}m), flow of coolant past tube arrays in heat exchangers (Re = 100–600, diameter $\sim$~1~mm) \cite{asif2024heatexchanger}, flow in microchemical reactors (Re = 1–50, diameter~$\sim$~100~{\textmu}m) \cite{krishnamurthy2007microreactor, horgue2013microreactor} and flow of underfill material through the solder bumps in semiconductor packaging (Re < 1, diameter $\sim$~100~{\textmu}m) \cite{wu2024semiconductor}. Due to the large number of material boundaries, the mesh resolution must be sufficiently fine to capture the intricate flow characteristics, and the underlying meshless nature of PINNs makes them an attractive choice for modelling complex geometries. There is growing interest in applying PINNs to fluid-related equations; see, e.g., \cite{cai2021physics, cai2021flow, zhu2024physics, botarelli2025using, wang2025simulating, hanna2022residual, toscano2025pinns}. However, many previous studies are limited to benchmark problems with analytic solutions or to flow past single obstacles, which have been used to assess the accuracy of newly developed methods.

The focus of this work is the consideration of a prototypical model of single-phase viscous fluid flow for low Reynolds numbers, namely the Stokes equation, defined on a domain with multiple perforations. This model is a standard benchmark problem in numerical homogenisation of composite materials, see e.g., \cite{bodaghi2016statistics, griebel2010homogenization}. Overall, we believe that it constitutes an interesting benchmark problem for PINNs as well, owing not only to application needs but also to the multi-objective PINN optimisation stemming from the governing fluid mechanics equations, the rich structure of these equations, and the multi-scale features of the problem arising in complex geometries. In addition, hybrid PINN approaches can be applied for understanding the theoretical properties of composite reinforcement materials at the micro- and macro-scales, where existing experimental \cite{yong2025permbenchmark, annamalai2025extension}, numerical \cite{syerko2023imagebased,jo2024permeability}, and other ML \cite{caglar2022deeplearning, jean2026imagebased, schmidt2025numerical} methods can be used as a source of data and prior knowledge for training PINNs with hybrid objectives. However, to the best of our knowledge, the PINN approximation for this problem has only been considered in \cite{korolev2026hybrid}. Even with the use of a state-of-the-art variant of PINNs based on sophisticated modifications of multilayer perceptrons, as proposed in \cite{wang2023expert}, accurate approximations in \cite{korolev2026hybrid} are obtained only for problems with a small number of perforations, while the method remains quite slow in terms of both the iteration count and the wall-clock time required for optimisation, and is hardly scalable to more complex settings. 

Despite their continued growing interest, PINNs are widely acknowledged to be difficult to train and often fail to produce accurate predictions, especially for problems which contain high-frequency or multi-scale features and when scaling to large problem domains \cite{zhu2019surrogate}. The training of PINNs is a multi-objective optimisation problem, involving the minimisation of losses related to the PDE residuals, boundary conditions and any data losses. Recent research has identified unbalanced gradients in the multi-objective optimisation and \textit{spectral bias} among the key factors contributing to these limitations \cite{rohrhofer2023data, wang2021understanding}. We recall that spectral bias refers to the tendency of neural networks to learn low-frequency features preferentially \cite{rahaman2019spectral, wang2021eigenvector}. Since fluid flow in highly perforated domains contains substantial high-frequency components, this bias makes the problem inherently challenging for neural networks to learn and scale to more demanding applications.

One of the major contributors to the computational challenge identified in \cite{korolev2026hybrid} for flow in highly perforated domains is the use of soft boundary constraints, i.e., a penalty approach is employed to enforce the no-slip boundary conditions on the perforations. This approach adds significant stiffness to the discrete gradient flow underlying PINN optimisation, degrades the approximation of derivatives, and consequently hinders accurate simulation of fluid flow in the microstructure created by the perforations. Encoding boundary conditions directly into neural network architectures as "hard" constraints provides an alternative to soft constraints, promising improvements in both approximation quality by enforcing boundary conditions exactly and computational performance by reducing the number of loss terms. However, practical difficulties arise in designing such neural network classes. Despite some progress in this direction \cite{sukumar2022exact, liu2022unified, straub2025hard, berg2018unified, hintermuller2026constrained} and an expanding range of applications \cite{wang2023exact, lu2021physics, xie2024physics, sun2020surrogate}, the explicit incorporation of boundary conditions into neural networks depends strongly on the specific PDE, including the order of the differential operator, the domain shape and boundary smoothness, multi-scale features, and related factors. Therefore, tailored boundary-imposition strategies in PINNs are required, as improper handling may compromise solution properties or even render a well-posed problem ill-posed. 


Recently, an increasing number of studies have used domain decomposition (DD) to reduce a complex problem into many smaller, simpler problems in order to mitigate spectral bias, while also allowing parallelisation across multiple GPUs \cite{klawonn2024domain}. Instead of using one large network over the entire domain, the domain is split into multiple subdomains, with a neural network placed in each subdomain. Various domain decomposition methods have been used, with overlapping or non-overlapping subdomains. In methods with non-overlapping domains, such as extended physics-informed neural networks (XPINNs) \cite{jagtap2020xpinn}, loss terms are introduced to ensure continuity at the interfaces between adjacent subdomains, meaning that continuity is only softly enforced, and these additional loss terms introduce difficulties on both the approximation and optimisation aspects. On the other hand, in DD methods with overlapping subdomains, the solution is a weighted sum of the subdomain network outputs, and continuity is strictly enforced through window functions or gate networks without introducing any additional interface loss terms. Finite basis physics-informed neural networks (FBPINNs) utilise window functions to combine the solution from the subdomain networks, inspired by classical finite element (FE) methods, where the solution to a differential equation is expressed as a sum of a finite set of basis functions \cite{moseley2023finite}. Augmented PINNs (APINNs) are similar to FBPINNs in using overlapping subdomains, but the domain decomposition itself is learned through trainable weighting of the subdomain outputs through a gate network \cite{hu2023augmented}. Although they provide flexibility in the domain decomposition, each subnetwork sees the entire domain as input, which can be computationally expensive with a large number of subdomains, hindering their ability to model large complex domains.

In this work, to efficiently capture the high-frequency flow within the microstructure, we adopt FBPINNs \cite{moseley2023finite} to our problem setting. FBPINNs are especially advantageous for multi-scale problems because the inputs to each subdomain network are normalised, scaling a high-frequency problem into many lower frequency problems, which directly helps reduce the influence of spectral bias. Furthermore, we use tailored hard constraints for the no-slip boundary condition to encode the geometry of the no-slip boundary directly into the neural network model. This effectively shapes the curvature of the viscous terms in the Stokes residual and eliminates the associated boundary loss terms. 
Backed by extensive numerical experiments and supporting theory, we show that hard-constrained FBPINNs are highly effective at mitigating spectral bias. Combined with tailored training techniques, including problem scaling, adaptive collocation point refinement, and loss balancing, we obtain an architecture that scales to more complex settings with multiple perforations and exhibits convergence that is less sensitive to the number of perforations.

The structure of the paper is as follows. In Section 2, we define the PINN framework for solving the Stokes problem in a highly perforated domain. Additionally, we discuss the importance of hard constraints for both the approximation and optimisation of the Stokes problem. We describe the FBPINN structure and the application of hard boundary constraints in Section 3. In this section, we also establish universal approximation properties of FBPINNs and present a Fourier analysis of the FBPINN ansatz, revealing two key mechanisms in its structure: frequency rescaling via localisation and normalisation, and convolution-induced spectral leakage induced by localisation of window functions. In Section 4, numerical experiments are carried out for varying numbers of perforations, highlighting the advantages of the hard-constrained FBPINN approach, even for large numbers of perforations. Finally, Section 5 concludes the paper with a summary of our findings.

\section{Physics-informed neural networks in perforated domains}

To better highlight the implications of our practical approach to enforcing hard constraints in PINNs on some of their important theoretical underpinnings, it will be convenient to introduce the standard notation related to the Hilbert spaces $L^{2}(\Omega)^{d}$ and $H^{k}(\Omega)^{d}$ of vector fields on $\Omega \subset \mathbb{R}^{d}$ (see, e.g. \cite{hunter2001applied}). First, for a scalar-valued function $v: \Omega \rightarrow \mathbb{R}$, we define its corresponding $L^{2}$ norm:  
\begin{align*}
\lVert v \rVert_{L^{2}(\Omega)} = \left(\int_{\Omega} |v(x)|^{2} \, dx \right)^{1/2},
\end{align*}
where $|\cdot| : \mathbb{R} \rightarrow \mathbb{R}_{\geq0}$ denotes the modulus function. We further use $\lVert \, \cdot \,  \rVert : \mathbb{R}^{d} \rightarrow \mathbb{R}_{\geq0}$ to denote the Euclidean norm. For a generic vector field $\boldsymbol{v} = (v_1,\dots,v_d)$ (for which we reserve bold letters), with scalar-valued component functions $v_i$ ($1 \le i \le d$) (for which we now reserve normal letters), we define
\begin{align*}
\|\boldsymbol{v}\|_{L^{2}(\Omega)^{d}}
:= \left( \sum_{i=1}^{d} \|v_i\|_{L^{2}(\Omega)}^{2} \right)^{1/2},
\qquad
\|\boldsymbol{v}\|_{H^{k}(\Omega)^{d}}
:= \left( \sum_{j=0}^{k} \int_{\Omega} \nabla^{j} \boldsymbol{v}(x) : \nabla^{j} \boldsymbol{v}(x)\, dx \right)^{1/2},
\end{align*}
where "$:$" denotes the Frobenius inner product between two tensors of the same order. We say that $\boldsymbol{v} \in L^{2}(\Omega)^{d}$ or $\boldsymbol{v} \in H^{k}(\Omega)^{d}$ (with $k\ge 1$), respectively, if the norm $\|\boldsymbol{v}\|_{L^{2}(\Omega)^{d}}$ or $\|\boldsymbol{v}\|_{H^{k}(\Omega)^{d}}$ is finite. The derivatives here are understood in the weak sense in general.

\subsection{Stokes equation in a perforated domain}\label{Sec 2.1}

Let $\Omega = D  \setminus \cup_{k=1}^{K} \,\mathcal{P}_{k}$ be the set obtained from an open, bounded, and connected domain $D \subset \mathbb{R}^d$ $(d=2,3)$ by removing a (relatively large) number $K$ of circular perforations $\mathcal{P}_{k} := \{\, x : \|x - c_{k}\| \le R \,\}$, where $c_{k}$ denotes the centre of each circle and $R$ denotes the common radius of all perforations. Hereafter, we also assume that $\Omega$ is connected. We consider the Stokes equation in such a highly perforated domain $\Omega$, which is given by the system
\begin{equation}\label{Stokes equation}
\begin{aligned}
- \mu_{\mathcal{D}} \Delta \boldsymbol{u} + \nabla p &= \boldsymbol{f} \quad \text{in } \Omega, \\
\operatorname{div} \boldsymbol{u} &= 0 \quad \text{in } \Omega, \\
\boldsymbol{u} &= 0 \quad \text{on } \partial \Omega^{\mathrm{p}}, \\
\boldsymbol{u} &= \boldsymbol{g} \quad \text{on } \partial \Omega^{\mathrm{w}},
\end{aligned}
\end{equation}
where the velocity field $\boldsymbol{u}: \Omega \rightarrow \mathbb{R}^{d}$ and the pressure function $p: \Omega \rightarrow \mathbb{R}$ are sought. The parameter $\mu_{\mathcal{D}}  \in \mathbb{R}_{+}$ denotes the dynamic viscosity of the fluid. The boundary $\partial \Omega^{\mathrm{p}}=\cup_{k=1}^{K} \, \partial \mathcal{P}_{k}$ corresponds to the collection of smooth surfaces $\partial \mathcal{P}_{k}$ (perforation boundaries) of the circular perforations within $\Omega$, while $\partial \Omega^{\mathrm{w}} := \partial \Omega \setminus \partial \Omega^{\mathrm{p}}$ denotes the outer walls of the domain. The vector field $\boldsymbol{f} \in L^{2}(\Omega)^{d}$ represents the body force acting on the fluid, and $\boldsymbol{g}$ is the prescribed velocity profile on $\partial \Omega^{\mathrm{w}}$ satisfying the compatibility condition
\begin{align}\label{compatibility condition}
\int_{\partial \Omega^{\mathrm{w}}}  \boldsymbol{g} \cdot \boldsymbol{\eta}^{\mathrm{w}} \  ds = \int_{\Omega}\operatorname{div} \boldsymbol{u} \ dx = 0,  
\end{align}
where $\boldsymbol{\eta}^{\mathrm{w}}$ is the outer normal vector to $\partial \Omega^{\mathrm{w}}$. We assume that $\boldsymbol{u} \in H^2(\Omega)^d$ and $ p \in H^1(\Omega) \cap L^{2}_{0}(\Omega)$; that is, $ \lVert p  \rVert_{H^{1}(\Omega)} < \infty$ and $\int_\Omega p(x) \, dx = 0$. Note that the latter mean normalisation results in a unique pressure function for \eqref{Stokes equation}. We note that the above regularity requirements for the velocity field and the pressure are necessary for an appropriate  PINN loss minimisation -- approximation error equivalence \cite{zeinhofer2025unified, mishra2023estimates, shin2023error}, yielding convergence in suitable Sobolev norms for PINNs; see also the next section for details. For example, these regularity assumptions are satisfied if the condition \eqref{compatibility condition} holds, $D$ is a convex polygon, and $\partial \Omega^{\mathrm{p}}$ is of class $C^{2}$; cf. \cite{wolf2022homogenization}. The latter may be achieved by assuming that the circular perforations neither intersect each other nor the boundary $\partial \Omega^{\mathrm{w}}$. If the perforations intersect, the intersection forms an acute corner that may further complicate the convergence of PINNs in the neighbouring region.

Fig.~\ref{fig:oscillations} shows the velocity magnitude $\lVert \boldsymbol{u}(x)\rVert$ obtained from finite element solutions (we defer the details of our Taylor--Hood finite element approximation to a later section) for an example 2D fluid flow problem \eqref{Stokes equation} with a viscosity of 0.1~Pa$\cdot$s and a prescribed velocity of $\boldsymbol{g} = (100,0)$~{\textmu}m/s through rectangular arrangements of 16 and 64 fibres with a diameter of 6~{\textmu}m. The plots of the velocity magnitude at $x = 50$~{\textmu}m clearly illustrate the multi-scale nature of the problem and the increase in high-frequency features with an increase in the number of fibres, making the solution of the Stokes equation in a highly perforated domain a suitable test case for studying the advantages of FBPINNs in mitigating spectral bias.

\begin{figure}
\centering
\includegraphics[width=0.7\linewidth]{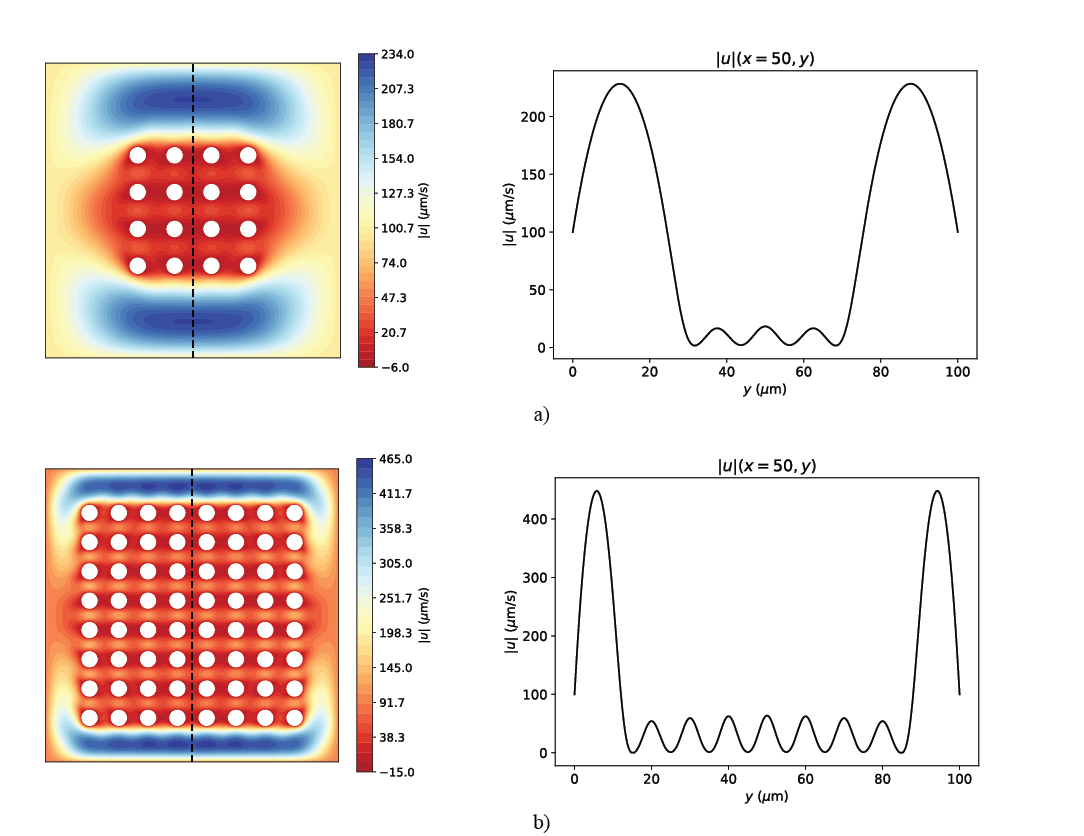}
\caption{\label{fig:oscillations}Plots of the velocity magnitude contours and the velocity magnitude at $x=50$~{\textmu}m for flow through a) 16 and b) 64 fibres.}
\end{figure}


\subsection{Soft and hard constraints in PINNs: Approximation and optimisation aspects}
Physics-informed neural networks approximate the velocity field $\boldsymbol{u}$ and the pressure $p$ using neural networks $\boldsymbol{u}_{\theta}$ and $p_{\psi}$ from suitably chosen neural network classes $\mathcal{NN}_{\Theta}^{\boldsymbol{g}} \subset H^{2}(\Omega)^{d}$ and $\mathcal{NN}_{\Psi}\subset H^{1}(\Omega) \cap L^{2}_{0}(\Omega)$, parameterised by $\theta \in \Theta \subseteq \mathbb{R}^{n_u}$ and $\psi \in \Psi \subseteq \mathbb{R}^{n_p}$, respectively, which yields the total number of neural network parameters $n = n_u +n_p$. For any $\boldsymbol{v}_{\theta} \in \mathcal{NN}_{\Theta}^{\boldsymbol{g}}$, we have $\restr{\boldsymbol{v}_{\theta}}{\partial \Omega^{\mathrm{w}}} = \boldsymbol{g}$, and hence the boundary condition on $\partial \Omega^\mathrm{w}$ for $\boldsymbol{u}_{\theta}$ is satisfied exactly by design. The neural networks are then inserted into the governing equation \eqref{Stokes equation}, yielding the residuals for the momentum equation, the divergence-free condition, and the no-slip boundary condition on the perforations. These residuals are minimised via combined  $L^2$-penalty objectives, resulting in the following optimisation problem
\begin{equation}\label{PINN objective}
\min_{u_\theta, \, p_\psi} \; E(\boldsymbol{u}_\theta, p_\psi) := 
 E_{r}(\boldsymbol{u}_\theta, p_\psi) + 
\lambda_{\text{div}} \, E_{\text{div}}(\boldsymbol{u}_\theta) + 
\lambda_{b} \, E_{b}(\boldsymbol{u}_\theta),
\end{equation}
where $\lambda_{\text{div}}, \lambda_b > 0$ are penalty parameters, and the individual penalty terms are given by:
\begin{equation*}
\begin{aligned}
E_{r}(\boldsymbol{u}_\theta, p_\psi) 
= \int_{\Omega} \left \lVert \boldsymbol{r}_{\theta, \psi}(x) \right\rVert^2 \, dx, \quad 
E_{\text{div}}(\boldsymbol{u}_\theta) 
= \int_{\Omega} \left| \operatorname{div} \boldsymbol{u}_{\theta}(x) \right|^2 \, dx, \quad 
E_{b}(\boldsymbol{u}_\theta) = \int_{\partial \Omega^p} \left \lVert \boldsymbol{u}_{\theta}(x) \right \rVert^2 \, dx,
\end{aligned}
\end{equation*}
where by $\boldsymbol{r}_{\theta, \psi} : \Omega \rightarrow \mathbb{R}^{d}$ we denote the momentum residual 
\begin{align}\label{momentum residual}
\boldsymbol{r}_{\theta, \psi}(x)=-\mu_{\mathcal{D}}  \,\Delta \boldsymbol{u}_{\theta}(x) + \nabla p_{\psi}(x) - \boldsymbol{f}(x).
\end{align}
The integrals above are approximated by Monte Carlo quadrature using random uniformly sampled collocation points $\{x^{r}_{i}\}_{i=1}^{M_{r}} \subset \Omega$ and $\{x^{b}_{i}\}_{i=1}^{M_{b}} \subset \partial\Omega^{\mathrm{p}}$. Reducing the dependence on neural networks to the generating parameters, we obtain
\begin{equation}\label{discrete losses}
\begin{aligned}
\mathcal{L}_{r}(\theta, \psi) 
= \frac{|\Omega|}{M_{r}}\sum_{i=1}^{M_{r}} \left\lVert \boldsymbol{r}_{\theta, \psi}(x_{i}^{r}) \right\rVert^2, \ \ 
\mathcal{L}_{\text{div}}(\theta) 
= \frac{|\Omega|}{M_{r}}\sum_{i=1}^{M_{r}} \left| \operatorname{div} \boldsymbol{u}_{\theta}(x_{i}^{r}) \right|^2, \ \ 
\mathcal{L}_{b}(\theta) = \frac{|\partial \Omega^{\mathrm{p}}|}{M_{b}}\sum_{i=1}^{M_{b}} \left\lVert \boldsymbol{u}_{\theta}(x_{i}^{b}) \right\rVert^2.
\end{aligned}
\end{equation}
The spatial derivatives appearing in these penalties are computed using automatic differentiation, as popularised in \cite{raissi2019physics}. Combining the losses yields the finite-dimensional multi-objective PINN optimisation problem
\begin{align}\label{PINN soft BC}
\min_{\theta, \, \psi} \;  \mathcal{L}(\theta, \psi): =  \mathcal{L}_{r}(\theta, \psi) + \lambda_{\mathrm{div}} \, \mathcal{L}_{\text{div}}(\theta)  +   \lambda_{b} \, \mathcal{L}_{b}(\theta),
\end{align}
which is the formulation introduced in \cite{raissi2019physics, lagaris1998artificial} with many possible extensions, see, e.g., \cite{toscano2025pinns} and references therein.

When training a physics-informed neural network, it is reasonable to ask what quality of approximation is expected, given a particular optimisation objective and the class of neural networks for approximation. These questions fall within approximation theory and numerical analysis, and several works have studied them in depth (see, e.g., \cite{zeinhofer2025unified, mishra2023estimates, shin2023error, de2024error, de2024numerical}), with a well-developed theory available for linear PDEs \cite{zeinhofer2025unified, de2024numerical}. A systematic analysis of the influence of soft and hard constraints on approximation is provided in \cite{zeinhofer2025unified}. From this viewpoint, we first consider the soft-constrained approach \eqref{PINN soft BC}, where the no-slip boundary condition is enforced via the penalty term. For $\boldsymbol{u}_{\theta} \in \mathcal{NN}_{\Theta}^{\boldsymbol{g}}$ and $p_{\psi} \in \mathcal{NN}_{\Psi}$, the following estimate for the approximation error is then expected \cite[Theorem 8]{zeinhofer2025unified}:
\begin{align}\label{H1/2 rate}
\lVert \boldsymbol{u} - \boldsymbol{u}_{\theta} \rVert_{H^{1/2}(\Omega)^{d}}^{2} + \lVert p - p_{\psi} \rVert_{[H^{1/2}(\Omega) \cap L^{2}_{0}(\Omega)]^{\ast}}^{2}  \lesssim \, \mathcal{L}(\theta, \psi) + \mathrm{quad}(\theta, \psi),
\end{align}
where the quadrature error $\mathrm{quad}(\theta, \psi) = E(\boldsymbol{u}_{\theta}, p_{\psi}) - \mathcal{L}(\theta, \psi)$ can be made arbitrarily small with high probability by using sufficiently many Monte Carlo collocation points; cf. \cite[Theorem 1]{zeinhofer2025unified}, also \cite{mishra2023estimates}. Therefore, if the loss \eqref{PINN soft BC} is small, the approximation error of PINNs is small as well. However, the (fractional) Sobolev norm for the velocity field approximation is rather weak, such that even the gradient $\nabla \boldsymbol{u}$ may not be well approximated in the $L^{2}$ norm. Similarly, the pressure approximation is controlled in the dual norm only, which is a very weak norm. As shown in \cite{muller2022notes}, the estimates of the Sobolev exponents in bounds of the above type can, in fact, be sharp. For the hard-constrained approach, consider a neural network class $\mathcal{NN}_{\Theta}^{\dagger} \subset H^{2}(\Omega)^{d}$ satisfying
\begin{align}\label{hard constrained class}
\mathcal{NN}_{\Theta}^{\dagger} = \{\boldsymbol{v}_{\theta} :  \restr{\boldsymbol{v}_{\theta}}{\partial \Omega^{\mathrm{w}}} = \boldsymbol{g} \  \mathrm{and} \   \restr{\boldsymbol{v}_{\theta}}{\partial \Omega^{\mathrm{p}}} = 0 \}. 
\end{align}
Then the following estimate \cite[Remark 13]{zeinhofer2025unified} for $\boldsymbol{u}_{\theta} \in \mathcal{NN}_{\Theta}^{\dagger}$ and $p_{\psi} \in \mathcal{NN}_{\Psi}$ is available
\begin{align}\label{H1 rate}
\lVert \boldsymbol{u} - \boldsymbol{u}_{\theta} \rVert_{H^{1}(\Omega)^d}^{2} + \lVert p - p_{\psi} \rVert_{L^{2}(\Omega)}^{2}  \lesssim \, \mathcal{L}(\theta, \psi) + \mathrm{quad}(\theta, \psi),
\end{align}
yielding more plausible norms for approximation, similar to those in FEM, e.g., in Taylor–Hood finite elements. The norms in \eqref{H1 rate} can be further strengthened to the $H^{2}$ norm for the velocity field and in the $H^{1}$ norm for the pressure, provided that the neural network class satisfies $\mathcal{NN}_{\Theta}^{\dagger} \subset H^{2}(\Omega)^{d}\cap H(\mathrm{div};\Omega)$, where $ H(\mathrm{div};\Omega) = \{ \boldsymbol{v} \in H^{1}(\Omega)^d : \mathrm{div} \  \boldsymbol{v} = 0\}$. For $\Omega \subset \mathbb{R}^{2}$, exact enforcement of the divergence-free condition in a neural network architecture can be achieved in some cases by introducing a scalar-valued stream function $\varphi$ such that $\boldsymbol{u}= \mathrm{curl} \,\varphi := ( \partial \varphi /  \partial x_{2}, - \,  \partial \varphi / \partial x_{1})$, and employing a neural network $\varphi_{\theta}$ to approximate $\varphi$, defining $\boldsymbol{u}_{\theta}=\mathrm{curl} \ \varphi_{\theta}$. However, $\varphi$ is only defined for simply connected domains, in which every closed curve can be continuously shrunk to a point. This topological property does not hold in the presence of perforations \cite{johnson2009numerical}, rendering the stream function formulation inapplicable for our case. Therefore, we enforce the divergence-free condition via the penalty, and only \eqref{H1 rate} can be expected. 
\begin{rem}
The ability to make the losses in \eqref{H1/2 rate} and \eqref{H1 rate} small can, in principle, be ensured by universal approximation results \cite{shin2023error, de2024numerical}. For this argument to hold, one requires $H^2$ regularity of the velocity field and $H^1$ regularity of the pressure. In particular, \cite{guhring2021approximation} (see also \cite[Theorem B7]{de2024error}) shows that any $\boldsymbol{w} \in H^3(\Omega)^d$ can be approximated arbitrarily well in the $H^{2}(\Omega)^d$ norm by a sufficiently large two-layer neural network $\boldsymbol{v}_{\theta}$ with $\tanh$ activation. 
Since $H^{3}(\Omega)^d$ is dense in $H^{2}(\Omega)^d$, for any $\varepsilon > 0$ and any $\boldsymbol{v} \in H^{2}(\Omega)^d$ there exists $\boldsymbol{w} \in H^{3}(\Omega)^d$ and a neural network $\boldsymbol{v}_{\theta}$ such that
\begin{align*}
\|\boldsymbol{v} - \boldsymbol{v}_{\theta}\|_{H^{2}(\Omega)^d}
&\le 
\|\boldsymbol{v} - \boldsymbol{w}\|_{H^{2}(\Omega)^d}
+ 
\|\boldsymbol{w} - \boldsymbol{v}_{\theta}\|_{H^{2}(\Omega)^d}
< \varepsilon.
\end{align*}
An analogous argument applies to the pressure variable, yielding approximation in the $H^{1}(\Omega)$ norm. 
\end{rem}
Despite universal approximation guarantees, the parameters $\theta$ and $\psi$ must in practice be obtained by solving the non-convex, highly nonlinear optimisation problem \eqref{PINN soft BC}, where the soft imposition of boundary conditions further complicates matters. Indeed, PINNs are typically trained by variants of (stochastic) gradient descent, whose generic form in our setting is:
\begin{equation}\label{gradient flow}
\begin{aligned}
\theta^{k+1} &= \theta^k - l_{\textrm{r}} \, \nabla_{\theta} \mathcal{L}(\theta^k, \psi^k) 
= \theta^k - l_{\textrm{r}} \left[ 
 \nabla_{\theta} \mathcal{L}_r(\theta^k, \psi^k) 
+ \lambda_{\rm div} \, \nabla_{\theta} \mathcal{L}_{\rm div}(\theta^k) 
+ \lambda_b \, \nabla_{\theta}\mathcal{L}_b(\theta^k)
\right], \\
\psi^{k+1} &
= \psi^k - l_{\textrm{r}}  \, \nabla_{\psi} \mathcal{L}_r(\theta^k, \psi^k),
\end{aligned}
\end{equation}
where $l_\textrm{r} >0$ is the learning rate. The discrete gradient flow above is known to be stiff, which is one of the major factors contributing to the failure of PINNs \cite{wang2021understanding}. Here, stiffness refers to large differences in the gradient scales of individual loss terms, causing gradient-based training to prioritise objectives with dominant gradients. In our examples, the ratio of $\lVert \nabla_{\theta} \mathcal{L}_{r}(\theta^{k}, \psi^{k})\rVert$ to $\lVert \nabla_{\theta} \mathcal{L}_{b}(\theta^{k})\rVert$ ranges between $10^5$ and $10^6$, which hinders the learning of the no-slip boundary condition and prevents the complex microstructure (high-frequency features) from even emerging in the underlying PDE solution. A common remedy is to balance the gradients in \eqref{gradient flow} by tuning the weights $\lambda=(\lambda_{\mathrm{div}},\lambda_b)$; see \cite{wang2023expert, wang2021understanding, bischof2025multi}. GradNorm is often used as a first choice and performs well on many benchmarks \cite{wang2023expert}. It consists of scaling the gradients based on their relative magnitudes:
\begin{equation}\label{standard loss balancing}
\begin{aligned}
\hat{\lambda}_{\text{div}} & = \frac{\left[\lVert \nabla_\theta \mathcal{L}_r(\theta, \psi) \rVert^2 + \lVert \nabla_\psi \mathcal{L}_r(\theta, \psi) \Vert^2\right ]^{1/2} + \lVert \nabla_\theta \mathcal{L}_{\rm div}(\theta) \rVert + \lVert \nabla_\theta \mathcal{L}_b(\theta) \rVert }{\lVert \nabla_\theta \mathcal{L}_{\text{div}}(\theta) \rVert}, \\
\hat{\lambda}_{b} & =  \frac{\left[\lVert \nabla_\theta \mathcal{L}_r(\theta, \psi) \rVert^2 + \lVert \nabla_\psi \mathcal{L}_r(\theta, \psi) \Vert^2\right ]^{1/2} + \lVert \nabla_\theta \mathcal{L}_{\rm div}(\theta) \rVert + \lVert \nabla_\theta \mathcal{L}_b(\theta) \rVert }{\lVert \nabla_\theta \mathcal{L}_{b}(\theta) \rVert},
\end{aligned}
\end{equation}
and updating the existing weights every $100$–$1000$ iterations using a moving average $\lambda_{\mathrm{new}} = \alpha \, \lambda_{\mathrm{old}} + (1- \alpha) \, \hat{\lambda}$,
where $\hat{\lambda} = (\hat{\lambda}_{\text{div}}, \,  \hat{\lambda}_{b})$, and the parameter $\alpha=0.9$ is recommended. 

While the loss balancing approach \eqref{standard loss balancing} helps stabilise the discrete gradient flow, it comes at a cost of also scaling the Hessian of $\mathcal{L}$. This can introduce complicated curvature in the loss landscape that is difficult to navigate using first-order optimisation algorithms, which is a well-known issue in penalty approaches \cite{wright1999numerical}. Precisely, recall that the Hessian of $\mathcal{L}$ with respect to $(\theta,\psi)$ is given by the symmetric matrix 
\begin{align}\label{Hessian}
\nabla^2 \mathcal{L}(\theta,\psi)
=\begin{pmatrix}
 \nabla^2_{\theta\theta} \mathcal{L}_r(\theta,\psi)
+ \lambda_{\mathrm{div}} \nabla^2_{\theta\theta} \mathcal{L}_{\mathrm{div}}(\theta)
+ \lambda_b \nabla^2_{\theta\theta} \mathcal{L}_b(\theta)
&
\nabla^2_{\theta\psi} \mathcal{L}_r(\theta,\psi)
\\[1ex]
 \nabla^2_{\psi\theta} \mathcal{L}_r(\theta,\psi)
&
 \nabla^2_{\psi\psi} \mathcal{L}_r(\theta,\psi)
\end{pmatrix}.
\end{align}
The following estimate for the loss values between two consecutive optimisation steps then holds \cite{wang2021understanding}: 
\begin{align}\label{loss decrease}
\mathcal{L}(\theta^{k+1}, \psi^{k+1}) - \mathcal{L}(\theta^{k}, \psi^{k}) = l_\textrm{r} \, \lVert  \nabla \mathcal{L}(\theta^{k},\psi^{k}) \Vert^{2} \, \bigg( -1  + \frac{l_\textrm{r}}{2} \sum_{i=1}^{n} \nu_{i} \, y_{i}^{2}\bigg) 
\end{align}
for some $\boldsymbol{y}=(y_{1},\dots, y_{n})$ with $\lVert \boldsymbol{y}\rVert = 1$, and $\nu_{1} \leq \nu_{2} \leq \dots \nu_{n}$ are the eigenvalues of $\nabla^2 \mathcal{L}(\tilde{\theta},\tilde{\psi})$, where $\tilde{\theta} = t \theta^{k} + (1-t) \theta^{k+1}$ and $\tilde{\psi} = t \psi^{k} + (1-t) \psi^{k+1}$ for some $t\in (0,1)$. It is well-known that \eqref{Hessian} is highly ill-conditioned \cite{wang2021understanding, rathore2024challenges}, especially for complex PDEs, as measured by the ratio of its largest to smallest eigenvalues. As a result, many of these eigenvalues are large, and scaling by $\lambda$ can further increase them by orders of magnitude, potentially rendering the right-hand side of \eqref{loss decrease} positive.  Therefore, even with $-\nabla \mathcal{L}(\theta^{k}, \psi^{k})$ defining the descent direction for $\mathcal{L}(\theta^{k}, \psi^{k})$, one may either require a small learning rate to obtain a decrease in \eqref{loss decrease} (too small for efficient optimisation), or face instabilities that prevent the optimisation from reaching good accuracy in a reasonable amount of time \cite{wang2021understanding}. We also refer to \cite{DBLP:conf/iclr/RyckBMB24} for a more refined analysis of the ill-conditioning in PINNs, where, in particular, the conditioning of \eqref{Hessian} is related to the conditioning of the underlying differential operator (specifically its square, due to the least-squares formulation) in the loss and the employed neural network architecture (its neural tangent kernel); within this framework, it is demonstrated that hard constraints act as a preconditioner for PINNs. 

Last but not least, we note that it is impossible to identify a unique PDE solution without satisfying boundary conditions. In PINNs, enforcing these conditions via soft penalties may introduce conflicting gradients between the boundary loss terms and the PDE residual terms, contributing to the presence of multiple local minima in the optimisation landscape \cite{liu2024config, yu2020gradient}. For example, although the no-slip loss in our Stokes setting directly penalises only velocity values on $\partial\Omega^{\mathrm p}$, satisfying this constraint together with the PDE residual requires an accurate representation of the velocity derivatives near the perforation boundaries. A large boundary penalty may then bias the network toward simpler, overly damped approximations that reduce $\mathcal{L}_{b}$ but deteriorate the residual loss $\mathcal{L}_{r}$, leading to conflicting optimisation directions. Similarly, an update that improves incompressibility can simultaneously worsen the momentum balance, again leading to gradient conflicts and unstable training. Several approaches, including second-order optimisation \cite{wang2025gradient} and gradient surgery \cite{liu2024config, yu2020gradient, xu2023transfer}, have been proposed as potential means of minimising these conflicts.

Despite the partial success of modern gradient-based optimisers such as Adam \cite{kingma2014adam} in navigating complex loss landscapes, we believe that spectral bias, gradient stiffness, and Hessian ill-conditioning remain major obstacles to scaling \eqref{Stokes equation} to complex settings with multiple perforations using PINNs. Enforcing the no-slip condition through suitable hard-constrained neural network classes, thereby eliminating the boundary penalty term, simplifies the optimisation problem: it reduces stiffness, improves Hessian conditioning, alleviates boundary-induced gradient conflicts, and promotes convergence in stronger norms.


\section{Hard-constrained FBPINNs}

In this section, we consider the practical design of neural network classes with hard constraints, as given in \eqref{hard constrained class}, using FBPINNs, and discuss several nuances and tricks related to FBPINN training. Before proceeding, however, we study some important theoretical underpinnings of FBPINNs as a domain decomposition framework, specifically their universal approximation properties and their capability to mitigate spectral bias. For the latter, we employ Fourier analysis.

\subsection{FBPINNs as a domain decomposition framework}\label{section:DD}

We begin by defining the unconstrained finite-basis neural network classes. Let $\Omega \subset D$ denote the perforated domain, where $D$ is an enclosing rectangular domain. We consider an overlapping cubic partition $\{D_i\}_{i=1}^{N}$ of $D$. The amount of overlap between neighbouring subdomains is controlled by the overlap ratio $\delta \geq 1$, where $\delta = 1$ corresponds to the non-overlapping case. 
The boundaries of each rectangular subdomain are defined as follows
\begin{align}
x^{\min}_{i,j} = \frac{x^{\max}_j - x^{\min}_j}{N_j-1} \left(i-\frac{\delta}{2}\right), \quad  x^{\max}_{i,j} = \frac{x^{\max}_j - x^{\min}_j}{N_j-1} \left(i+\frac{\delta}{2}\right),
\end{align}
where $N_j$ is the number of subdomains in each direction, with $\prod_{j=1}^{d} N_j=N$. 
Let $\{\omega_i\}_{i=1}^{N}$ be a collection of window functions forming a partition of unity subordinate to $\{D_i\}_{i=1}^{N}$, i.e., $0 \leq \omega_i \leq 1$ and $\sum_{i=1}^{N} \omega_i = 1$ on $D$, with $\operatorname{supp}(\omega_i)\subset D_i$. For each $i$, define the local perforated subdomains $\Omega_i := D_i \cap \Omega$. Since $\Omega \subset D$ and $\bigcup_{i=1}^{N} D_i = D$, we obtain $\bigcup_{i=1}^{N} \Omega_i = \Omega$, so that $\{\Omega_i\}_{i=1}^{N}$ forms an overlapping decomposition of $\Omega$. We restrict the partition of unity from $D$ to $\Omega$ by defining $\omega_i^{r} := \omega_i|_{\Omega}$. Clearly, $\{\omega_i^{r}\}_{i=1}^{N}$ still form a partition of unity on $\Omega$. Moreover, $\operatorname{supp}(\omega_i^{r})=\operatorname{supp}(\omega_i)\cap \Omega\subset \Omega_i$. Hence, $\{\omega_i^{r}\}_{i=1}^{N}$ is subordinate to the decomposition $\{\Omega_i\}_{i=1}^{N}$ of $\Omega$. For simplicity of notation, we subsequently write $\omega_i$ in place of $\omega_i^{r}$.

Following \cite{moseley2023finite}, the FBPINN ansatz for $\boldsymbol{v}_{\theta} \in \mathcal{NN}_{\Theta}$, where $\boldsymbol{v}_{\theta} : \Omega \rightarrow \mathbb{R}^{d}$, is given by
\begin{align}\label{FBPINN ansatz}
\boldsymbol{v}_{\theta}(x)
=
\sum_{i=1}^{N}
\omega_i(x)\,
\mathrm{unnorm}
\circ
\boldsymbol{NN}_i
\circ
\mathrm{norm}_i(x),
\qquad x \in \Omega,
\end{align}
where $\boldsymbol{NN}_i(\mathrm{norm}_i(\cdot);\theta_i) : \Omega_{i} \rightarrow \mathbb{R}^{d}$ denotes a local neural network assigned to $\Omega_i$, $\theta = \{\theta_i\}_{i=1}^{N}$ is the collection of network parameters, and the normalisation functions $\mathrm{norm}_i$ are defined using the ``parent'' cubic subdomains $D_i = \prod_{j=1}^{d} [x^{\min}_{i,j}, x^{\max}_{i,j}]$ as follows:
\begin{align}\label{scaling transform}
\mathrm{norm}_i : D_i \to [-1,1]^d,
\qquad
x \mapsto \tilde{x} := \frac{x - \mu_i}{\sigma_i}.
\end{align}
Here, $\mu_i =(\mu_{i,1},\cdots, \mu_{i,d}) \in\mathbb R^d$ and $\sigma_i \in\mathbb R$ are
\begin{align}
\mu_i = \frac{x^{\max}_i + x^{\min}_i}{2},
\qquad
\sigma_i
=
\frac{x^{\max}_{i,j}-x^{\min}_{i,j}}{2},
\end{align}
where, due to the isotropy of the cubic subdomains, $\sigma_i$ is independent of the coordinate direction $j$. In general, it is preferable to choose smooth bijective functions (diffeomorphisms) for the normalisation \eqref{scaling transform}. The $\mathrm{unnorm}$ function in \eqref{FBPINN ansatz} maps the local neural network outputs to be within the range $[-1,1]$, and thereby also balances the training of the velocity and pressure networks. Lastly, a variety of window functions can be used in \eqref{FBPINN ansatz}, as provided by \cite{moseley2023finite}. The current study uses the cosine window function
\begin{equation}\label{window functions}
\omega_i(x)
=
\frac{\hat{\omega}_i(x)}
{\sum_{k=1}^{N}\hat{\omega}_k(x)},
\qquad
\hat{\omega}_i(x)
=\prod_{j=1}^{d}
\left(
\frac{1+\cos\left(\pi \tilde{x}_{j}\right)}{2}
\right)^2,
\end{equation}
where $\tilde{x}_{j}=(x_j-\mu_{i,j})/\sigma_i$ denotes the $j$-th normalised coordinate associated with the subdomain $\tilde{\Omega}_i:=\mathrm{norm}_i(\Omega_i)$.

The ansatz \eqref{FBPINN ansatz} also possesses universal approximation properties. 
\begin{proposition}[Universal approximation]\label{UA proposition}
Let $\Omega \subset D \subset \mathbb{R}^{d}$ be a perforated domain, and let $\{D_i\}_{i=1}^{N}$ be an overlapping cubic subdomain partition of $D$. Assume that the maximum diameter of the perforations is strictly smaller than the side length of each subdomain $D_i$, and let $\mathcal{NN}_{\Theta}$ denote the class of FBPINN ansatz functions\footnote{we use the  unnormalisation function $\mathrm{unnorm}=\mathrm{id}$; however, the choice of unnormalisation does not affect the approximation capabilities.} of the form \eqref{FBPINN ansatz}. Then, for every $\boldsymbol v \in H^{s}(\Omega)^d$ and every $\varepsilon > 0$, there exists $\boldsymbol v_\theta \in \mathcal{NN}_{\Theta}$ such that $\lVert \boldsymbol v - \boldsymbol v_\theta \rVert_{H^{s}(\Omega)^d} < \varepsilon$.
\end{proposition}
\begin{proof}
See~\ref{APP1}
\end{proof}
\noindent
Despite the qualitative nature of Proposition~\ref{UA proposition}, one practical conclusion is that it is desirable to avoid subdomain partitions in which the boundary of a perforation intersects a subdomain face near its diameter, since this may create acute corners (and hence non-Lipschitz boundaries) in the resulting subdomain elements of the partition, making approximation with smooth neural networks more challenging.

Multilayer perceptron neural networks are known for their spectral bias, as they preferentially learn low-frequency components while delaying convergence on high-frequency components. FBPINNs, however, are capable of efficiently mitigating spectral bias due to domain decomposition principles and the scaling induced by the normalisation \eqref{scaling transform}. Subdomain rescaling mitigates spectral bias by transforming a high-frequency approximation problem into a collection of lower-frequency local problems, which are easier to learn. This frequency transformation can be characterised explicitly using the Fourier transform. For the latter, we use the following Fourier transform convention
\begin{align}\label{Fourier transform}
\mathcal{F}_{x}[f](\xi) = \int_{\mathbb{R}^d} f(x)e^{-2\pi \mathrm{i}x\cdot \xi}\,dx, 
\end{align}
where $\mathrm{i} = \sqrt{-1}$ denotes the imaginary unit. For vector-valued functions, the Fourier transform is applied componentwise. Since each window function $\omega_i$ is compactly supported in $\Omega_i$, we regard the windowed local contributions as functions on $\mathbb R^d$ by extending them by zero outside their supports. Consequently, the FBPINN ansatz \eqref{FBPINN ansatz} is compactly supported in $\Omega$, and its Fourier transform is well-defined componentwise on $\mathbb R^d$. 
\begin{theorem}[Fourier transform of the FBPINN ansatz]\label{Fourier transform proposition}
Consider the FBPINN ansatz \eqref{FBPINN ansatz}, together with the local normalisation map \eqref{scaling transform}. Furthermore, assume that the output unnormalisation is linear with $\alpha_{\mathrm{un}}>0$:
\begin{align}
\mathrm{unnorm}(z)=\alpha_{\mathrm{un}} z.
\end{align}
Then, for every frequency $\xi \in \mathbb{R}^{d}$, the Fourier transform of \eqref{FBPINN ansatz} is given by
\begin{align}\label{FBPINN in spectral domain}
\mathcal{F}[\boldsymbol{v}_\theta](\xi)
=
\alpha_{\mathrm{un}} 
\sum_{i=1}^{N}
\sigma_{i}^d e^{- 2\pi \mathrm{i} \, \mu_{i}\cdot\xi}
\bigg( \mathcal{F}_{\tilde{x}}[\tilde{\omega}_{i}] \ast \mathcal{F}_{\tilde{x}}[\boldsymbol{NN}_{i}] \bigg) (\sigma_{i} \xi),
\end{align}
where $\mathcal{F}_{\tilde{x}}$ denotes the Fourier transform with respect to the local (normalised) variable $\tilde{x}$, $\tilde \omega_{i}(\tilde{x}): = \omega_{i}(\mu_{i} + \sigma_{i} \tilde{x})$, and the convolution is given by
\begin{align}\label{convolution in F space}
\bigg( \mathcal{F}_{\tilde{x}}[\tilde{\omega}_{i}] \ast \mathcal{F}_{\tilde{x}}[\boldsymbol{NN}_{i}] \bigg) (\sigma_{i} \xi) =   \int_{\tilde{\Omega}_i} \omega_i(\mu+\sigma_{i}\tilde{x}) \boldsymbol{NN}_i(\tilde{x}) e^{-2 \pi \mathrm{i} \, \tilde{x}\cdot(\sigma_{i} \xi)} \,d\tilde{x}.
\end{align}
Since $\tilde{\omega}_i$ is scalar-valued and $\boldsymbol{NN}_i$ is vector-valued, the convolution is understood componentwise.
\end{theorem}
\begin{proof}
See~\ref{APP2}
\end{proof}
\noindent
Before interpreting Theorem~\ref{Fourier transform proposition}, we briefly recall the uncertainty principle from Fourier analysis. The uncertainty principle states that a function cannot be simultaneously highly localised in both physical space and frequency (Fourier) space. For $f\in L^2(\mathbb R^d)$ with finite second moments in both physical and Fourier space, one classical form of the uncertainty principle is given by the inequality 
\begin{align}\label{uncertainty principle}
\left(
\int_{\mathbb R^d} \lVert x \rVert ^2  \, |f(x)|^2\,dx
\right)
\left(
\int_{\mathbb R^d} \lVert \xi \rVert^2  \, |\mathcal{F}_{x}[f](\xi)|^2\,d\xi
\right)
\geq C \|f\|_{L^2}^4,
\end{align}
where $C>0$ is some constant (see, e.g. \cite{folland1997uncertainty}). Here, the first integral measures the spatial spread of $f$, while the second measures the spread of its Fourier transform $\mathcal{F}_{x}[f]$. The inequality therefore quantifies the trade-off between localisation in physical space and localisation in Fourier space: for $f\neq 0$, the product of these two measures of spread is bounded away from zero. Thus, one of the two quantities may be made small only at the expense of increasing the other. In particular, stronger localisation in physical space necessarily induces spreading in Fourier space, and vice versa.

Through the lens of Fourier analysis, we now observe that learning with FBPINNs involves two competing spectral mechanisms. First, the spectral representation \eqref{FBPINN in spectral domain} reveals that the frequency variable $\xi\in\mathbb R^d$ enters the local networks in the rescaled form $\sigma_i \xi$. Since the local subdomain scales are typically much smaller than the global domain scale, the effective frequencies $\sigma_i\xi$ seen by the local subnetworks are correspondingly reduced. Consequently, high-frequency features in the global coordinates $x$ appear as lower-frequency features in the local coordinates $\tilde{x}$, thereby providing an efficient mechanism for mitigating spectral bias. Second, multiplication by the local window function induces convolution in Fourier space \eqref{convolution in F space}. Using the definition of convolution, we obtain
\begin{align}\label{convolution 2}
\bigg(
\mathcal{F}_{\tilde{x}}[\tilde{\omega}_i]
*
\mathcal{F}_{\tilde{x}}[\boldsymbol{NN}_i]
\bigg)(\sigma_i\xi)
=
\int_{\mathbb R^d}
\mathcal{F}_{\tilde{x}}[\tilde{\omega}_i](\sigma_i\xi-\zeta)
\mathcal{F}_{\tilde{x}}[\boldsymbol{NN}_i](\zeta)
\,d\zeta.
\end{align}
Hence, the spectral contribution at frequency $\sigma_i\xi$ depends not only on the local neural network spectrum at the same frequency, but also on neighbouring frequencies $\zeta$, weighted by the Fourier transform of the window function. The Fourier transform $\mathcal{F}_{\tilde{x}}[\tilde{\omega}_i]$ therefore acts as a spectral coupling kernel. By the uncertainty principle, stronger spatial localisation of the window function $\tilde{\omega}_i$ necessarily produces broader spectral support for $\mathcal{F}_{\tilde{x}}[\tilde{\omega}_i]$. Consequently, the convolution kernel spreads over a wider range of frequencies, increasing frequency mixing and spectral leakage between neighbouring spectral modes. To demonstrate the application of the uncertainty principle, in Fig.~\ref{fig:localisation principle}, we plot 1D analogues of the (overlapping) window functions \eqref{window functions} on the interval $[0,1]$ for both a coarse ($N=3$) and a finer ($N=9$) subdomain partition, using small and large overlap ratios $\delta=1.1$ and $\delta=2.5$. One can clearly observe that smaller overlaps introduce fatter tails in the Fourier spectrum of a representative window (chosen to be the middle window of the partition in Fig. \ref{fig:localisation principle}) -  an effect that is further intensified for finer subdomain partitions. At the same time, coarser subdomain partitioning and larger overlaps localise the spectrum more effectively around zero and flatten the spectral tails.
\begin{figure}
    \centering
    \includegraphics[width=0.95\linewidth]{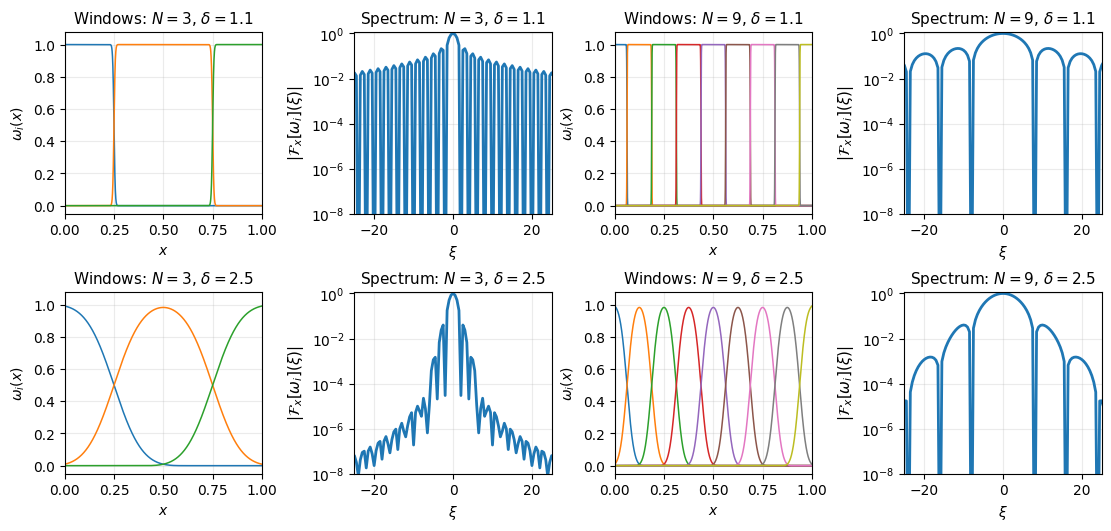}
    \caption{Effect of subdomain partition fineness and overlap on windows and their Fourier spectra: shrinking physical support of the window functions by either increasing the number of subdomains $N$ or decreasing the overlap $\delta$ broadens their representative Fourier spectrum. The plotted spectra correspond to the Fourier spectra of the windows located at the centres of the respective partitions.}
    \label{fig:localisation principle}
\end{figure}
\begin{rem}
We note that the convolution \eqref{convolution in F space} in Theorem~\ref{Fourier transform proposition} and in \eqref{convolution 2} is written in the local normalised variables $\tilde x$, but Fig.~\ref{fig:localisation principle} shows the window functions in the original physical variable $x$. However, it is easy to show via change of variables that the local and global Fourier spectra are related by:
\begin{align}\label{Fourier scaling relation}
\mathcal F_x[\omega_i](\xi)
=
\sigma_i^d e^{-2\pi \mathrm{i}\mu_i\cdot \xi}
\mathcal F_{\tilde x}[\tilde\omega_i](\sigma_i\xi).
\end{align}
The spectra in Fig.~\ref{fig:localisation principle} are plotted for the physical windows $\omega_i(x)$ in the global variable $x$. By the Fourier scaling relation \eqref{Fourier scaling relation}, these spectra are equivalent to the local spectra of $\tilde\omega_i$ up to the rescaling $\xi\mapsto\sigma_i\xi$ and a constant factor. Indeed, since $|e^{-2\pi \mathrm{i}\mu_i\cdot \xi}|=1$, we have $|\mathcal F_x[\omega_i](\xi)|=\sigma_i^d \, |\mathcal F_{\tilde x}[\tilde\omega_i](\sigma_i\xi)|$. Thus, the stronger spatial localisation of the windows leads to broader Fourier spectra, whether the spectra are viewed in local or global frequency variables.
\end{rem}

One practical conclusion suggested by the above analysis is that aggressive subdomain partitioning improves spectral separation and yields a collection of simpler local learning problems, but generally requires larger overlaps to avoid excessive convolution-induced frequency mixing, which may otherwise make the optimisation more challenging. 
In particular, insufficient overlap narrows the transition regions over which neighbouring subnetworks are blended, resulting in steeper spatial variations of the coupling windows. 
This, in turn, increases the high-frequency content of $\mathcal{F}_{\tilde{x}}[\tilde{\omega}_i]$ and introduces artificial frequency coupling that is not intrinsic to the underlying PDE solution. 
Excessive (convolution-induced) spectral leakage partially re-couples neighbouring frequency bands, reducing the effectiveness of the decomposition in mitigating spectral bias and potentially leading to slower or less stable optimisation. Increasing the overlap smooths the transitions between neighbouring subnetworks and reduces the resulting convolution-induced spectral coupling, thereby yielding more stable training for fine subdomain partitions. 

Although the discussion above addresses several aspects of choosing the number of subdomains and the amount of overlap in the continuous setting, an additional layer of discretisation in FBPINNs is introduced by the collocation method and the associated sampling strategy. 
In most domain decomposition methods for PINNs, a trade-off exists between overfitting due to insufficient training data (number of collocation points per subdomain) and the reduced complexity of the target function achieved through domain decomposition \cite{hu2023augmented}. If too few subdomains are used, there is less separation of the low- and high-frequency features, making the training susceptible to spectral bias. Having too many subdomains improves spectral separation, but increases the number of parameters to train and generally reduces the number of collocation points in each subdomain (as the overall set of collocation points is distributed among the subdomains), potentially leading to overfitting of localised target functions and poor predictions. In this regard, we remark that a certain amount of spectral leakage in frequency space may actually help regularise the training process, thus making the practical choice of $\delta$ less trivial. We also note that using larger overlaps increases the computational cost, as evaluating the solution at a given collocation point involves summing the solution (weighted by the window function) from all the subdomains it falls into. Therefore, multiple subdomain network evaluations are required at the same point in the overlapping regions, with the computational cost increasing with the amount of overlap. In our experiments, decomposition of the domain into subdomains whose size is comparable to the smallest characteristic length scale, for example, slightly larger than the diameter of the perforations, combined with a relatively large overlap ratio of $\delta = 2$, provides a reasonable choice for effective training.

\subsection{Hard-constraint design for FBPINNs}

The standard hard-constraint strategy in PINNs enforces boundary conditions by multiplying the network output by a distance function that vanishes on the boundary, allowing Dirichlet boundary conditions to be satisfied by construction. We illustrate it here for the velocity network $\boldsymbol{u}_{\theta}$:
\begin{align}\label{Intro 1}
\boldsymbol{u}_{\theta}(x) &= \mathcal{C}[\bar{\boldsymbol{u}}_{\theta}](x): = l_{\partial \Omega}(x)\bar{\boldsymbol{u}}_{\theta}(x) + \boldsymbol{g}(x) , 
\end{align}
where $\bar{\boldsymbol{u}}_{\theta}$ denotes the raw output of the neural network,  $\boldsymbol{u}_{\theta}$ is the hard-constraint neural network ansatz, $l_{\partial \Omega}(x)$ is a distance function to the boundary $\partial \Omega$  satisfying $l_{\partial \Omega}(x) = 0$ for $x \in \partial \Omega$, and $\boldsymbol{g}(x)$ represents the prescribed boundary data on $\partial \Omega$. Thus, $\mathcal{C}$ defines an affine operator applied to $\bar{\boldsymbol{u}}_{\theta}$, whose output must be sufficiently smooth to guarantee the existence of the derivatives required in the PINN training. This can be conveniently achieved by choosing smooth distance functions. We proceed with the generic hard-constrained form for enforcing the no-slip boundary condition
\begin{align}\label{generic HC}
\boldsymbol{u}_{\theta}(x)
&=
\mathcal{C}^{\mathrm p}[\bar{\boldsymbol{u}}_{\theta}](x)
=
l_{\partial \Omega^{\mathrm p}}(x)\,
\bar{\boldsymbol{u}}_{\theta}(x),
\end{align}
where $\mathcal{C}^{\mathrm p}$ is a linear constraining operator and $l_{\partial \Omega^{\mathrm p}}$ is a smooth boundary-distance function satisfying $l_{\partial \Omega^{\mathrm p}}(x)=0$ for $x\in\partial\Omega^{\mathrm p}$ and
$l_{\partial \Omega^{\mathrm p}}(x)>0$ for $x\in\Omega$. Owing to the hard constraint \eqref{generic HC}, the loss term associated with the no-slip condition in \eqref{PINN soft BC} is therefore removed. However, we note that $\boldsymbol{u}_{\theta}$ no longer enters \eqref{PINN soft BC} directly, but rather through its derivatives. To investigate the role of hard constraints on the optimisation problem, we expand the Laplacian of the constrained velocity and obtain
\begin{align}\label{eq: local laplacian I}
\Delta\bigl(l_{\partial \Omega^{\mathrm p}}\,
\bar{\boldsymbol{u}}_{\theta}\bigr)
= 
l_{\partial \Omega^{\mathrm p}}\,\Delta \bar{\boldsymbol{u}}_{\theta}
+
2\nabla l_{\partial \Omega^{\mathrm p}}\cdot\nabla \bar{\boldsymbol{u}}_{\theta}
+ \bar{\boldsymbol{u}}_{\theta}\,\Delta l_{\partial \Omega^{\mathrm p}}.
\end{align}
Note that the distance function $l_{\partial\Omega^{\mathrm p}}$ directly shapes the Laplacian of the constrained network output through $l_{\partial\Omega^{\mathrm p}}$, $\nabla l_{\partial\Omega^{\mathrm p}}$, and $\Delta l_{\partial\Omega^{\mathrm p}}$. Therefore, these terms inject the geometry of the no-slip condition into the residual \eqref{momentum residual} and can be interpreted as an inductive bias for the Laplacian of the constrained velocity. Consequently, not every function $l_{\partial \Omega^{\mathrm p}}$ that vanishes on $\partial\Omega^{\mathrm{p}}$ is equally suitable, and several design choices must be made. In particular, due to smoothness requirements and the dense arrangement of the perforations, $l_{\partial \Omega^{\mathrm{p}}}$ may remain relatively small in a neighbourhood of the perforation boundaries. Thus, the term $l_{\partial\Omega^{\mathrm p}}\, \Delta\bar{\boldsymbol{u}}_{\theta}$ in \eqref{eq: local laplacian I} may be damped precisely where the solution develops strong gradients and large curvature due to the no-slip condition. These near-boundary regions contain high-frequency spatial features that are already difficult for neural networks to learn because of spectral bias. On the other hand, $\Delta l_{\partial \Omega^{\mathrm p}}$ may become large and induce an overly restrictive bias, which is beneficial only when the curvature induced by $l_{\partial \Omega^{\mathrm p}}$ is compatible with that of the velocity field near the boundaries. In view of the damping effect introduced by $l_{\partial \Omega^{\mathrm p}}$, poor design choices may therefore create an imbalance among the terms in \eqref{eq: local laplacian I}, introducing stiffness into the residual \eqref{momentum residual} during training.

We now describe the construction of $l_{\partial\Omega^{\mathrm{p}}}$. For each perforation $\mathcal{P}_k$, we use a smooth disk function
\begin{align}\label{eq:disk_function}
\phi_{\mathrm{disk}}(x,c_k)
&=
\frac{\|x-c_k\|^2-R^2}{2R},
\end{align}
where $c_k$ and $R$ denote the perforation centre and radius, respectively, as defined in Section~\ref{Sec 2.1}. We define the aggregated inverse distance function, following \cite{sukumar2022exact}, and its corresponding hyperbolic tangent modulation
\begin{align}\label{eq:local_distance_function}
l_{\mathrm{inv}}(x)
&=
\left(
\sum_{k =1}^{K}
\phi_{\mathrm{disk}}(x,c_k)^{-m}
\right)^{-1/m}, \quad l_{\partial \Omega^{\mathrm{p}}}(x) = 
\tanh\!\left(a\,l_{\mathrm{inv}}(x)\right).
\end{align}
For $x \in \partial\Omega^{\mathrm{p}}$, the values of $l_{\mathrm{inv}}$ are understood in the sense of its continuous extension. Here, the parameter $m$ controls the degree of localisation in the aggregation of the individual disk functions. In fact, for $x \in \Omega$, one obtains the following pointwise limit \cite{sukumar2022exact}: 
\begin{align}\label{eq: minimum distance}
\underset{m\to\infty}{\lim}l_{\mathrm{inv}}(x) =  \min_{1\leq k\leq K}
\phi_{\mathrm{disk}}(x,c_k).
\end{align}
Thus, as $m$ increases, $l_{\mathrm{inv}}$ becomes increasingly dominated by the closest perforation boundary and approaches a smooth approximation of the minimum distance \eqref{eq: minimum distance}. Smaller values of $m$ retain a stronger cumulative contribution from several nearby perforations, which may lead to more pronounced damping in densely perforated regions.
At the same time, the parameter $a$ in \eqref{eq:local_distance_function} controls the subsequent hyperbolic tangent modulation. Larger values of $a$ make $l_{\partial\Omega^{\mathrm p}}$ saturate more rapidly away from $\partial\Omega^{\mathrm{p}}$, confining the damping effect to a thinner boundary layer. In   Fig.~\ref{fig:SDF design} (the top and middle rows), we compare $l_{\mathrm{inv}}$ and $l_{\partial \Omega^{\mathrm p}}$, defined in \eqref{eq:local_distance_function}, for different values of $m$ and $a$. We observe that larger values of these parameters yield stronger localisation, thereby reducing excessive damping in regions containing multiple perforations.  This, in turn, is expected to improve the learning of $\Delta\bar{\boldsymbol{u}}_{\theta}$ in \eqref{eq: local laplacian I}, particularly near $\partial\Omega^{\mathrm{p}}$.

The effect of the hard constraints on the remaining two terms in \eqref{eq: local laplacian I}, namely $2\nabla l_{\partial\Omega^{\mathrm p}}\cdot \nabla\bar{\boldsymbol{u}}_{\theta}$ and $\Delta l_{\partial\Omega^{\mathrm p}}\,\bar{\boldsymbol{u}}_{\theta}$, is more subtle. Computing the derivatives of the hard-constraint factor in \eqref{eq:local_distance_function}, we obtain
\begin{align}
\nabla l_{\partial \Omega^{\mathrm{p}}}
=
a\,\operatorname{sech}^2(a\,l_{\mathrm{inv}})\nabla l_{\mathrm{inv}}, \quad \Delta l_{\partial \Omega^{\mathrm{p}}}
=
a\,\operatorname{sech}^2(a\,l_{\mathrm{inv}})\Delta l_{\mathrm{inv}}
-
2a^2
\operatorname{sech}^2(a\,l_\mathrm{inv})
\tanh(a\,l_\mathrm{inv})
\|\nabla l_{\mathrm{inv}}\|^2.
\end{align}
We note that $\nabla l_{\mathrm{inv}}$ and $\Delta l_{\mathrm{inv}}$ are bounded for a finite $m$, provided that the perforations are mutually separated. Since $0 \leq \operatorname{sech}^{2}(a s)\leq 1$ and $\operatorname{sech}^{2}(a s)\to 0$ as $|s|\to\infty$, with the decay rate controlled by $a$, the derivative of the hyperbolic tangent localises the derivatives of $l_{\partial \Omega^{\mathrm{p}}}$. In particular, both $\nabla l_{\partial \Omega^{\mathrm{p}}}$ and $\Delta l_{\partial \Omega^{\mathrm{p}}}$ decay away from $\partial \Omega^{\mathrm p}$ once $a \,l_{\mathrm{inv}}(x)$ is sufficiently large. At the same time, the limiting minimum-distance function \eqref{eq: minimum distance} is generally non-smooth at points where several perforations are equally close. At such switching sets, the (generalised) derivative exists only in a set-valued sense, while the second derivative contains measure-valued contributions, analogous to Dirac-type singularities. Thus, increasing $m$ in \eqref{eq:local_distance_function} makes $l_{\mathrm{inv}}$ approach \eqref{eq: minimum distance} and can produce larger values of $\Delta l_{\mathrm{inv}}$. These curvature effects are inherited by $\Delta l_{\partial\Omega^{\mathrm p}}$ and can be further amplified by the parameter $a$ in the transition regions near the perforation boundaries, where $\operatorname{sech}^2(a \, l_{\mathrm{inv}})$ remains non-negligible. Similarly, tight localisation of $l_{\partial \Omega^{\mathrm{p}}}$ yields steeper transitions in $\nabla l_{\partial \Omega^{\mathrm{p}}}$, amplifying the gradient of the raw output $\bar{\boldsymbol{u}}_{\theta}$ in \eqref{eq: local laplacian I}. In summary, tight localisation of $l_{\partial \Omega^{\mathrm{p}}}$ provides better control of the first term in \eqref{eq: local laplacian I}, but at the cost of additional complexity in the remaining two terms. To support this analysis, in Fig.~\ref{fig:SDF design} (bottom row), we plot the logarithms of $\Delta l_{\partial \Omega^{\mathrm{p}}}$ over $\Omega$. The plots show that increasing $m$ and $a$ produces curvature contours with highly localised shapes and large magnitudes. In fact, for $m=12$ and $a=25$, the magnitudes differ by up to seven orders of magnitude (this is particularly visible near points that are equidistant from two neighbouring perforations). This may introduce significant stiffness into the residual by imposing an overly restrictive inductive bias, contributing to overfitting or stripe-like artefacts during training.

\begin{figure}[H]
    \centering
    \includegraphics[width=0.8\linewidth]{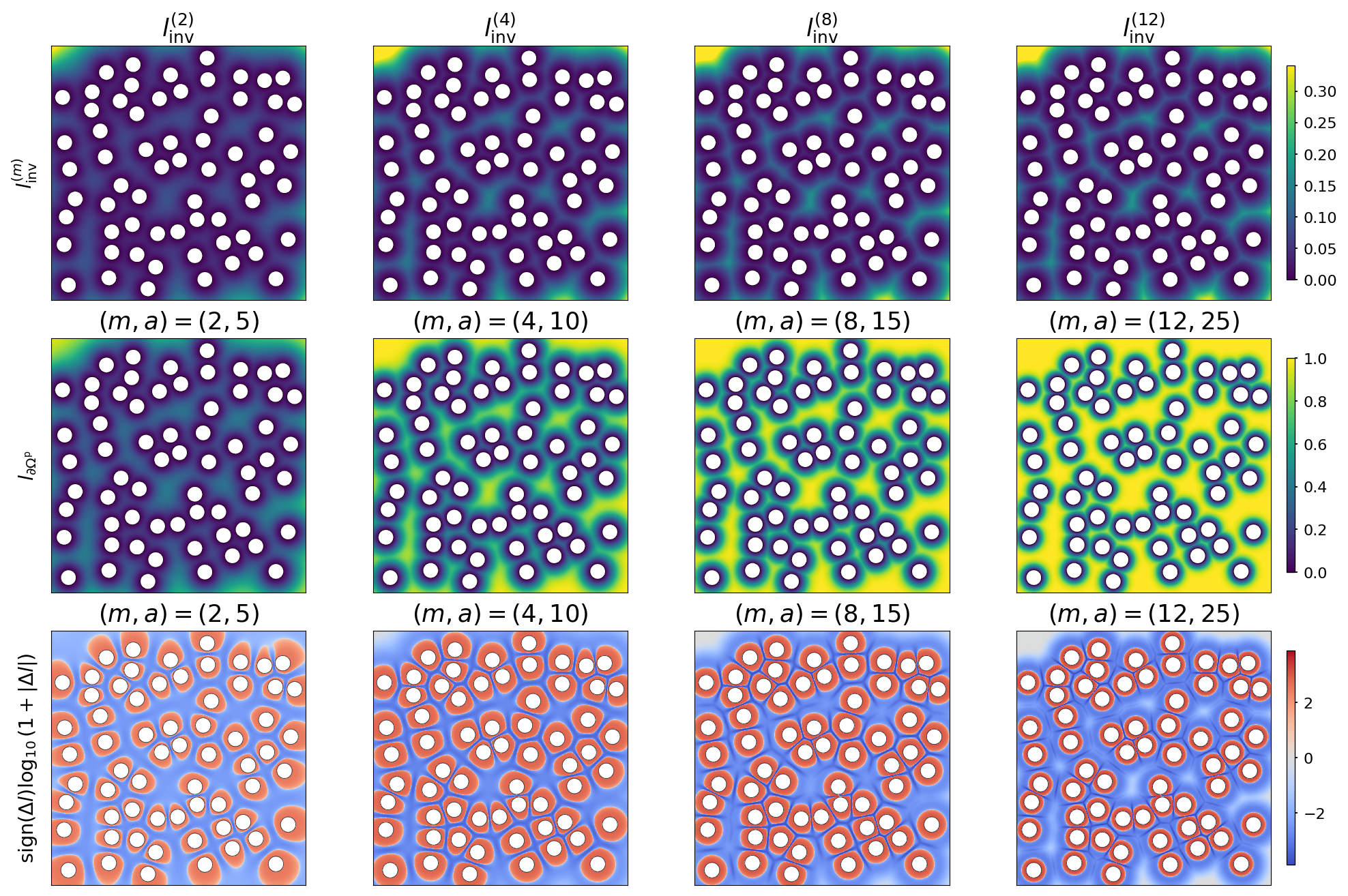}
    \caption{Localisation and curvature of $l_{\partial \Omega^{\mathrm{p}}}$ ($K=64$, $R=0.03$). The top row shows $l_{\mathrm{inv}}^{(m)}$ for $m=2,4,8,12$, while the middle row shows $l_{\partial\Omega^{\mathrm p}}=\tanh(a,l_{\mathrm{inv}}^{(m)})$ for $(m,a)=(2,5),(4,10),(8,15),(12,25)$. The bottom row displays the signed logarithmic transform $\operatorname{sign}(\Delta l_{\partial\Omega^{\mathrm p}})\log_{10}(1+|\Delta l_{\partial\Omega^{\mathrm p}}|)$ for the same parameter pairs. Larger values of $m$ and $a$ lead to stronger localisation, but also generate sharper curvature patterns in $\Delta l_{\partial\Omega^{\mathrm p}}$.}
    \label{fig:SDF design}
\end{figure}

We proceed by modifying the FBPINN ansatz \eqref{FBPINN ansatz} to obtain the hard-constrained neural network class \eqref{hard constrained class}, which satisfies exactly both the boundary conditions on the wall boundary $\partial \Omega^{\mathrm{w}}$ and the no-slip boundary condition on $\partial \Omega^{\mathrm{p}}$. A distinctive feature of the FBPINN ansatz is that hard constraints can be imposed not only on the overall neural network function, but also incorporated into the local subnetworks. 
For the velocity approximation $\boldsymbol{u}_{\theta} \in \mathcal{NN}_{\Theta}^{\dagger}$, the general form of the hard-constrained FBPINN network is given by
\begin{align}\label{HC FBPINN ansatz}
\boldsymbol{u}_{\theta}(x) = \mathcal{C}^{\mathrm{w}}\bigg[
\sum_{i=1}^{N}
\omega_i(x) \cdot
\mathrm{unnorm}^{\boldsymbol u}
\circ
\mathcal C_i^{\mathrm p}
\!\left[
\boldsymbol{NN}_{i}^{\boldsymbol u}
\right]
\circ
\mathrm{norm}_i(x)
\bigg] = \mathcal{C}^{\mathrm{w}}\bigg[\sum_{i=1}^{N}
\omega_i(x) \, \mathcal C_i^{\mathrm p}[\boldsymbol{z}_{i}](x)  \bigg].
\end{align}
Here, the local networks are denoted by $\boldsymbol{z}_{i}(x):= \mathrm{unnorm}^{\boldsymbol u}\circ \boldsymbol{NN}_{i}^{\boldsymbol u} \circ \mathrm{norm}_i(x)$ for brevity, the (globally constraining) affine operator $\mathcal C^{\mathrm w}$ is applied to the global output to impose the boundary condition on $\partial\Omega^{\mathrm w}$, while the (locally constraining) linear operators $\mathcal C_i^{\mathrm p}$ are applied to the individual subnetworks to impose the no-slip condition on $\partial\Omega^{\mathrm p}$. Since $\mathrm{unnorm}^{\boldsymbol u}$ acts as multiplication by a constant, the above notation is consistent.

The local hard constraint for the $i$-th subnetwork is constructed only from those perforation boundaries on which the corresponding window function is active. Precisely, for $i=1,\dots,N$, we define the local perforation index set
\begin{align}\label{eq:local_perforation_index_set}
\mathcal J_i
&=
\left\{
j\in\{1,\dots,K\}
:
\partial \mathcal{P}_j\cap \operatorname{supp}(\omega_i)\neq\emptyset
\right\}.
\end{align} 
Thus, for the element $\Omega_{i}$ of the subdomain partition, $\Gamma_{i} = \cup_{j \in \mathcal{J}_{i}} \partial \mathcal{P}_{j}$ is the union of the perforation boundaries relevant to the $i$-th subnetwork. Similar to \eqref{eq:local_distance_function}, we define the local approximate distance function
\begin{align}\label{eq:local_distance_function_1}
l_{\mathrm{inv}}^{i}(x)
&=
\left(
\sum_{j\in\mathcal J_i}
\left(
\frac{1}{\phi_{\mathrm{disk}}(x,c_j)}
\right)^m
\right)^{-1/m}, \quad l_{\partial \Omega^{\mathrm{p}}}^{i}(x) = \begin{cases}
\tanh\!\left(a\,l_{\mathrm{inv}}^{i}(x)\right),
& \mathcal J_i\neq\emptyset, \\[0.6em]
1,
& \mathcal J_i=\emptyset.
\end{cases}
\end{align}
For $x \in \Omega_{i}$, the corresponding local hard-constraining operator is then defined by
\begin{align}
\label{eq:local_hard_constraint_operator}
\mathcal C_i^{\mathrm p}
\left[
\boldsymbol{z}_{i}
\right](x)
=
\begin{cases}
l_{\partial \Omega^{\mathrm{p}}}^{i}(x)
\boldsymbol{z}_{i}(x),
& \mathcal J_i\neq\emptyset, \\[0.6em]
\boldsymbol{z}_{i}(x),
& \mathcal J_i=\emptyset.
\end{cases}
\end{align}
To implement such local constraints, it is sufficient to provide each local subnetwork with its corresponding set $\mathcal J_i$, which can be stored, for example, through the perforation centres associated with $\Gamma_{i}$. 
However, in the overlapping regions where the subdomain network outputs are summed together, the mismatch in the local approximate distance function can introduce artificial features to the overall ansatz. In the current study, instead of applying the no-slip condition as local hard constraints, we use hyperbolic tangent modulation with an appropriate choice of $a$ to localise the distance functions to the vicinity of the perforations, and apply both the wall and no-slip boundary hard constraints globally as in \eqref{Intro 1}. The investigation of local hard constraints in FBPINNs remains as future work.

We also note that, when combining different hard constraints, such as those for the wall boundary and the local constraints on the perforations, it is crucial to ensure their compatibility. For example, the lifting function $\boldsymbol{g}$ in \eqref{Intro 1}, used to impose the wall boundary condition, may have non-zero support on the perforation boundaries, yielding a non-zero velocity there. We discuss this issue in the numerical section for a particular choice of Dirichlet boundary conditions on $\partial \Omega^{\mathrm{w}}$.

For the pressure variable $p_{\psi} \in \mathcal{NN}_{\Psi}$, we also use the FBPINN architecture together with the normalisation
\begin{align}\label{pressure FBPINN}
p_{\psi}(x):=\bar{p}_{\psi}(x) - \frac{1}{|\Omega|}\int_{\Omega} \bar{p}_{\psi}(x) \, dx, \quad \bar{p}_{\psi}(x) = \sum_{i=1}^{N} w_{i}(x) \cdot \mathrm{unnorm}^{p} \, \circ \, NN_{i}^{p} \, \circ \, \mathrm{norm}_{i}(x)
\end{align}
to fix the constant and ensure the uniqueness of the pressure. The integral in this normalisation is approximated using a Monte Carlo approach over the given set of collocation points.

\begin{rem}
We note that enforcing hard constraints into the neural network ansatz may lead to over-constraining the underlying problem. For example, assume that one requires $p=0$ on the outlet of $\Omega= (0,1)^{2}$, and use $p_{\psi}(x)  = \mathcal{C}(x)\bar{p}_{\psi}(x)$ with the distance function $\mathcal{C}(x) = (1-x_1)$ to achieve it. According to the chain rule, we obtain
\begin{align*}
\frac{\partial p_{\psi}(x)}{\partial x_1} =
- \bar{p}_{\psi}(x) +
(1 - x_1) \cdot \frac{\partial \bar{p}_{\psi}(x)}{\partial x_1}, \quad 
\frac{\partial p_{\psi}(x)}{\partial x_2} =
(1 - x_1) \cdot \frac{\partial \bar{p}_{\psi}(x)}{\partial x_2}.
\end{align*}
Therefore, $\frac{\partial p_{\psi}(\boldsymbol{x})}{\partial x_2} = 0$ on the outlet, which is an artefact of the ansatz not present in the initial Stokes formulation. 
\end{rem}

\subsection{Non-dimensionalisation and characteristic length scales}\label{section:nondim}
Data normalisation is essential in deep learning for stable and effective training. In practice, inputs to the neural network are often rescaled to lie within a fixed range, such as $[-1,1]$, and  this is a "non-physical" scaling of the inputs only. In contrast, non-dimensionalisation rescales the entire physical system so that all variables and governing equations are expressed on comparable magnitudes. Since first- and second-order derivatives scale as $l_{\mathrm{c}}^{-1}$ and $l_{\mathrm{c}}^{-2}$, respectively, the pressure gradient and viscous terms scale as
\begin{align}
\nabla p \sim P_{\mathrm{c}} l_{\mathrm{c}}^{-1}, \qquad 
\mu_{\mathcal{D}}  \Delta \boldsymbol{u} \sim \mu_{\mathcal{D}}  U_{\mathrm{c}} l_{\mathrm{c}}^{-2},
\end{align}
where $l_{\mathrm{c}}$ is the characteristic length scale, and $P_{\mathrm{c}}$ and $U_{\mathrm{c}}$ are characteristic pressure and velocity magnitudes. Hence,
\begin{align}\label{scaling ratio}
\frac{\lVert \mu_{\mathcal{D}} \, \Delta \boldsymbol{u}\rVert}{\lVert \nabla p\rVert}
\sim
\frac{\mu_{\mathcal{D}} U_{\mathrm{c}}}{P_{\mathrm{c}} l_{\mathrm{c}}}.
\end{align}
The Stokes pressure scale $P_{\mathrm{c}}\sim \mu_{\mathcal{D}}  U_{\mathrm{c}}/l_{\mathrm{c}}$ is precisely the choice that makes this ratio of order one. However, $P_{\mathrm{c}}$ and $U_{\mathrm{c}}$ are not known before solving the problem and must be estimated from physical or geometric considerations. Without such scaling, the velocity and pressure networks are typically initialised with comparable output magnitudes, although the physically consistent pressure scale may differ substantially from the velocity scale. Under this initial neural network parametrisation, $P_{\mathrm{c}}\sim U_{\mathrm{c}}$, and therefore the ratio in \eqref{scaling ratio} scales as $\mu_{\mathcal{D}}  \, l_{\mathrm{c}}^{-1}$. In our applications, the characteristic length is chosen as the size of the smallest geometric feature, namely the fibre diameter, since this is the scale on which the strongest velocity variations are expected. Thus, for small $l_{\mathrm{c}}$, the ratio in \eqref{scaling ratio} is large, and the unscaled momentum residual \eqref{momentum residual} may initially be dominated by the viscous contribution. Consequently, training may be biased towards updates of the velocity network, whose parameter adjustments yield a larger immediate reduction of the residual, thereby slowing down the convergence of the pressure network. Since pressure and velocity are coupled only through the momentum residual in the standard PINN formulation, non-dimensionalisation provides a direct way to improve convergence, particularly for the pressure network, by reducing residual stiffness. 

We proceed with the following choice of non-dimensional variables
\begin{equation}\label{nondim_variables}
x^* = \frac{x}{l_\mathrm{c}}, \quad \boldsymbol{u^*} = \frac{\boldsymbol{u}}{g_x}, \quad p^* = \frac{p l_\mathrm{c}}{\mu_{\mathcal{D}}  \,g_x}, \quad \nabla^*=l_\mathrm{c}\nabla, \quad \Delta^*=l_\mathrm{c}^2\Delta,
\end{equation}
where $g_x$ (taken here as an estimate of $U_{\mathrm{c}}$) is the horizontal component of the prescribed velocity field $\boldsymbol{g}$. Observe that in the above scaling, $P_{\mathrm{c}} = \mu_{\mathcal{D}}\,g_{x} /l_{\mathrm{c}}$, thus making the ratio \eqref{scaling ratio} of order one. The corresponding non-dimensional Stokes equation is then given by
\begin{equation}\label{nondim Stokes equation}
\begin{aligned}
\frac{\mu_{\mathcal{D}} \, g_x}{l_\mathrm{c}^2}\left(- \Delta^* \boldsymbol{u}^* + \nabla^* p^*\right) &= \boldsymbol{f} \quad \text{in } \Omega^*, \\
\frac{g_x}{l_\mathrm{c}} \operatorname{div}^* \boldsymbol{u}^* &= 0 \quad \text{in } \Omega^*, \\
g_x \, \boldsymbol{u}^* &= 0 \quad \text{on } \partial \Omega^{\mathrm{p}*}, \\
g_x \, \boldsymbol{u}^* &= \boldsymbol{g} \quad \text{on } \partial \Omega^{\mathrm{w}*},
\end{aligned}
\end{equation}
where $\Omega^*$, $\Omega^{\mathrm{p}*}$ and $\Omega^{\mathrm{w}*}$ denote their corresponding non-dimensionalised domains. From this point onwards, we drop the symbol denoting the non-dimensional form, and all variables are expressed in non-dimensional form.  

Generally, it is evident that the choice of characteristic quantities in \eqref{nondim_variables} will affect the relative scaling of the PDE residuals, potentially increasing or decreasing stiffness in the discrete gradient flow \eqref{gradient flow} that further depends on the loss weights $\lambda$, see also \cite{rohrhofer2023data}. In our case, the non-dimensionalisation scales the PDE residuals by different factors: the momentum equation scales with ${\mu_{\mathcal{D}} \, g_x}/{l_\mathrm{c}^2}$, the divergence equation scales with ${g_x}/{l_\mathrm{c}}$ and the no-slip and wall boundary conditions scale with $g_x$. We take advantage of the inherent scaling that exists in the physical system and set $\lambda_\mathrm{div}=({\mu_{\mathcal{D}} }/{l_\mathrm{c}})^{2}$ and $\lambda_{b}=({\mu_{\mathcal{D}} }/{l_\mathrm{c}}^2)^{2}$ as a first estimate (given that the squared norms are used to define the losses \eqref{discrete losses}). The adaptive loss weighting scheme \eqref{standard loss balancing} will then continue to adjust the loss weights to ensure balanced training. Here, it is important to clarify that loss weighting schemes balance between the different loss terms in the problem, whereas non-dimensionalisation aims to place the different terms in the PDE residual on comparable scales. However, both help to balance the loss gradients and stabilise learning. 


\subsection{Adaptive collocation point sampling}
The number and distribution of the collocation points have been shown to affect the accuracy of predictions significantly. Resampling collocation points is a well-established method for preventing overfitting, but more recently, methods have been developed to adaptively sample collocation points based on the value of the PDE residuals to strategically place points where training is required to minimise overfitting and maximise data efficiency. 
Similar to adaptive mesh refinement used in finite element methods, adaptive collocation point sampling has been used to successfully improve the accuracy of predictions using PINNs \citep{hanna2022residual,lu2021deepxde, wu2023rar}.
In the residual-based adaptive refinement (RAR) method \citep{lu2021deepxde}, new collocation points are added in regions with the largest PDE residuals. This effectively increases the contribution of these regions to the empirical residual loss, acting as an implicit form of loss rebalancing. Such adaptive refinement is particularly relevant in regions with sharp solution features, for example near steep fronts in Burgers' equation, where additional collocation points may be needed to resolve the solution accurately.
In the Stokes problem, steep gradients exist near the boundaries of the perforations and can be introduced due to the distance functions where the solution is forced to change rapidly (e.g., at the domain corners). Specific to FBPINNs, the gradients can be artificially steep at the subdomain overlaps, especially at initialisation, where the solution changes significantly between adjacent subdomains. Adaptive sampling helps mitigate the effects of steep gradients, as well as numerical artefacts that may arise due to overfitting in regions with insufficient collocation points.

In the current study, in addition to the random uniformly sampled collocation points $\{x^{r}_{i}\}_{i=1}^{M_{r}} \subset \Omega$, a smaller group of points is adaptively resampled every $T_{\mathrm{RAR}}$ iterations according to the probability density function proportional to the PDE residuals following \citep{wu2023rar}. 
A set of dense points $\mathcal{S}_0$ is randomly sampled ($100,000$ points), then the PDE residuals are evaluated for these points. Since there are multiple PDE residuals: the divergence-free and the momentum loss terms \eqref{PINN soft BC}, the sum of the normalised absolute residuals is used to identify regions where the overall residual is highest,
\begin{equation}\label{RAR_residual}
\varepsilon_\textrm{r} = \sum_{i=1}^C \frac{|{r}_i|-\min|{r}_i|}{\max|{r}_i| - \min|{r}_i|},
\end{equation}
where ${r}_i$ is PDE residual, and $C$ is the number of PDE residuals in the system. The probability density function is evaluated by
\begin{equation}\label{eq:pdf}
    \pi(x) = \frac{\varepsilon_\mathrm{r}(x)}{\int \varepsilon_\mathrm{r}(x) dx}.
\end{equation}
Then a subset of points $\mathcal{S}$ is randomly sampled from $\mathcal{S}_0$ using the probability mass function $\tilde \pi(x)=\pi(x)/\sum_{x\in{\mathcal{S}}_{0}}\pi(x)$ and added to the set of collocation points for training.

Residual adaptive refinement can automatically identify regions with high losses and add collocation points to speed up training and reduce overall errors, especially near the perforation boundaries. RAR allows for fewer collocation points overall than purely randomly sampled points, as the adaptive scheme targets where training is required. 
As presented in Fig.~\ref{fig:RAR}, it was observed that early in training (after 5000 iterations), the adaptive points are located in areas with insufficient collocation points, or in areas with steep gradients introduced by hard boundary constraints. Later in training (after 70,000 iterations), the adaptive points are more evenly distributed, with a higher concentration near the perforations, thereby assisting in capturing the micro-scale features. 

\begin{figure}[H]
\centering
\includegraphics[width=0.75\linewidth]{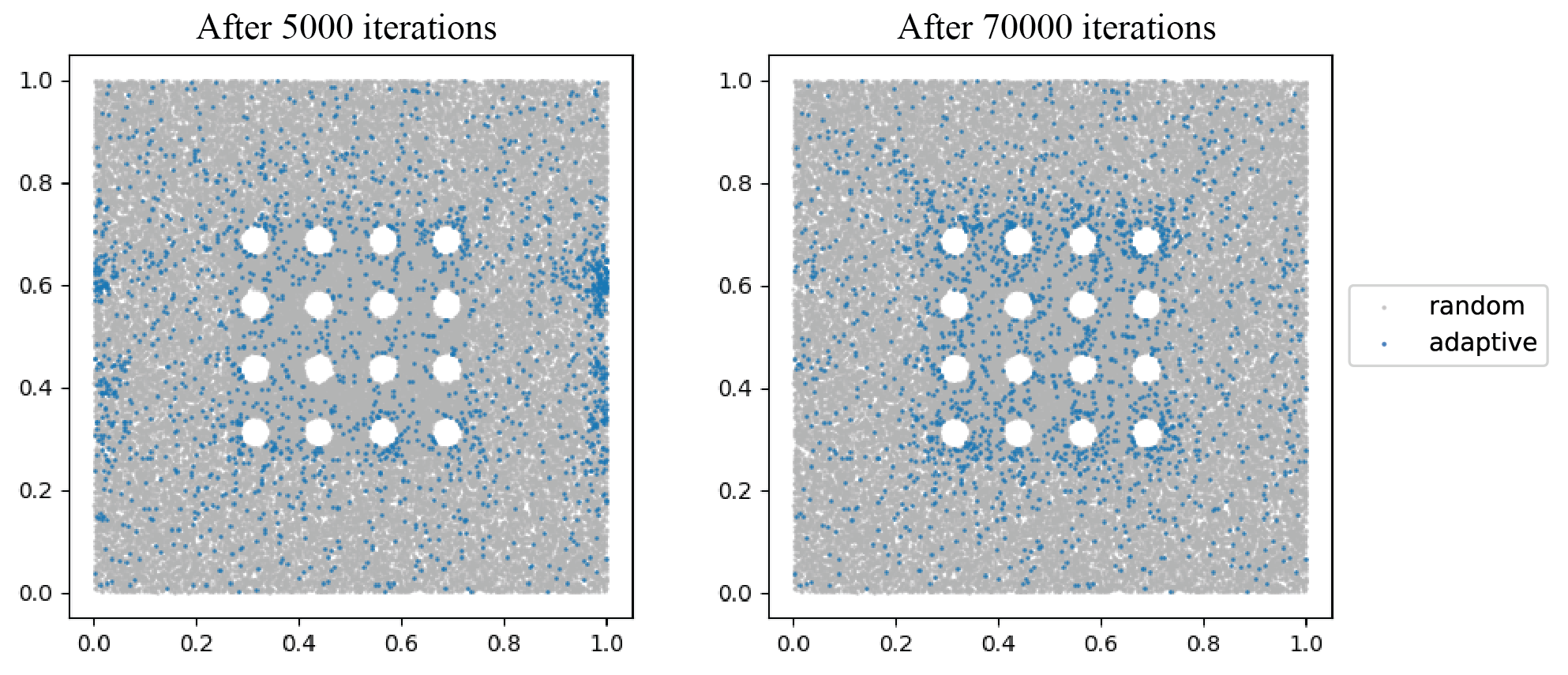}
\caption{\label{fig:RAR}Example of the randomly and adaptively-sampled collocation points after 5,000 and 70,000 iterations of training for 16 periodic perforations.}
\end{figure}


The hard-constrained FBPINN solver runs according to Algorithm \ref{alg:algorithm}. The solver receives as input the initial neural network parameters $\theta^0$ and $\psi^0$ for the velocity and pressure networks, respectively, the maximum number of training iterations $k_\text{max}$, and a number of user-defined update frequencies, including the batch resampling $T_\textrm{s}\in \mathbb N$, RAR sampling $T_\textrm{RAR}\in \mathbb N$, and loss weight update $T_\textrm{w}\in \mathbb N$ frequencies. The collocation points are sampled in the non-dimensionalised domain, and the inputs to each subdomain network are normalised. The local constraining operator is applied to the subdomain outputs, which are then unnormalised and multiplied by the window function. The subdomain outputs are summed, the global constraining operator is applied, and the required partial derivatives for the PDE are computed with automatic differentiation. Finally, the loss function and its gradient are evaluated, and the training parameters are updated via an optimiser of choice. After training for $k_\text{max}$ iterations, the trained parameters, along with the corresponding approximated velocity and pressure functions, are returned as outputs of the solver. These approximations can then be evaluated at any set of points within the domain to obtain the predictions.

\begin{algorithm}
\caption{Hard-constrained FBPINN solver}
\label{alg:algorithm}
\textbf{Inputs:} Initial neural network parameters ($\theta^0$, $\psi^0$), maximum number of iterations $k_{\text{max}}$, resampling frequency $T_\textrm{s}$, RAR sampling frequency $T_\textrm{RAR}$, loss weight update frequency $T_\mathrm{w}$ \\
\textbf{Outputs:} Trained neural network parameters ($\theta$, $\psi$)
\begin{algorithmic}[1]
\STATE{$k\gets 0$}
\WHILE { $0 \leq k \leq k_{\text{max}}-1$} \vspace{0.4em}
    \IF {$k\mod T_\mathrm{s}=0$} 
    \STATE{Randomly sample a set of points $\mathcal{T}$}
    \ENDIF
    \IF {$k \mod T_\textrm{RAR}=0$} 
    \STATE{Randomly sample a set of dense points $\mathcal{S}_0$}
    \STATE{Compute PDE residuals for the points in $\mathcal{S}_0$ according to \eqref{RAR_residual}}
    \STATE{Compute $\pi(x)$ using \eqref{eq:pdf} and randomly sample a set of points $\mathcal{S}$ from $\mathcal{S}_0$ based on the probability mass function $\tilde \pi(x)=\pi(x)/\sum_{x\in{\mathcal{S}}_{0}}\pi(x)$}
    \STATE{$\mathcal{T}\gets \mathcal{T} \cup \mathcal{S}$}
    \ENDIF
    \IF {$k \mod T_\mathrm{w} = 0$}
    \STATE{Update: $\lambda^k \gets $ Gradient-based weight scaling $(\nabla_{\theta,\psi}\mathcal{L}_{\lambda^k}(\theta^k,\psi^k),\lambda^k)$}
    \ENDIF

\STATE{Compute $\nabla_{\theta,\psi}\mathcal{L}_{\lambda_k}(\theta^k,\psi^k)$ using automatic differentiation}
\STATE{$\theta^{k+1},\psi^{k+1}\gets$ Optimiser($\nabla_{\theta,\psi}\mathcal{L}_{\lambda_k}(\theta^k,\psi^k)$, optimiser hyperparameters)}\\
\STATE{$k\leftarrow k+1$}\\
\ENDWHILE
  \end{algorithmic}
\end{algorithm}

\section{Numerical experiments}

In this section, we present our numerical simulations of the Stokes equations in two dimensions with FBPINNs. We consider a (dimensionless) unit square domain, while pointing out that solutions in a square domain of any size can be recovered by multiplying the pressure obtained on a unit domain by the side length of the domain, as evident from the dimensional analysis in Section \ref{section:nondim}. 
For the data in the Stokes equation \eqref{Stokes equation}, we choose $\boldsymbol{f}:=0$ for the body force term in $\Omega \subset \mathbb{R}^{2}$ and $\boldsymbol{g}:=(1,0)$ for the Dirichlet boundary condition on the wall boundary $\partial \Omega^{\mathrm{w}}$. Clearly, $\boldsymbol{g}$ satisfies the compatibility condition \eqref{compatibility condition}. Geometries with periodically and randomly scattered perforations are considered. For the periodic setting, three geometries with an increasing number (16, 36 and 64) of periodically arranged perforations are examined. In this setting, doubling the number of perforations
per direction also doubles the solution frequencies in the microstructure, thereby increasing the challenge for PINN approximation due to spectral bias. This periodic setting then introduces a convenient new benchmark for testing multi-scale PINNs. We further validate the methodology for randomly scattered perforations (36, 64, and 100 perforations), which better reflect practical applications. 

The Taylor--Hood finite element approximation is used to compute the reference solutions for validation purposes. This choice aligns well with the expected approximation quality \eqref{H1 rate} of our hard-constrained FBPINN approach, making such a comparison meaningful since neither method enforces the divergence-free condition exactly. The FBPINN model is implemented using the JAX framework \cite{bradbury2018jax} following the structure of \cite{moseley2023finite}. 

To encode $\boldsymbol{g}=(1,0)$ on $\partial \Omega^{\mathrm{w}}$ into the FBPINN ansatz \eqref{HC FBPINN ansatz}, we use the distance function $l_{\Omega^{\mathrm{w}}}(x) = 16x_{1}(1-x_{1})x_{2}(1-x_{2})$, and, for $\boldsymbol{v} : \Omega \rightarrow \mathbb{R}^{2}$, we define the following affine operator
\begin{align}\label{eq:wall_hard}
C^{\mathrm{w}}[\boldsymbol{v}](x) = \tanh(a^{\mathrm{w}} \, l_{\partial \Omega^{\mathrm{w}}}(x)) \, \boldsymbol{v}(x) + \bar{\boldsymbol{g}}(x), \quad \bar{\boldsymbol{g}}(x) =  \boldsymbol{g}[1- \tanh(a^{\mathrm{w}} \, l_{\partial\Omega^{\mathrm{w}}}(x))],
\end{align}
where $a^{\mathrm{w}}>0$ is a sufficiently large parameter that renders $\bar{\boldsymbol{g}}(x)$ negligible away from $\partial \Omega$, preventing a non-zero contribution of $\bar{\boldsymbol{g}}$ to the no-slip boundary, while also satisfying $\bar{\boldsymbol{g}} = \boldsymbol{g}$ on $\partial \Omega$. We use $a^{\mathrm{w}}=10$ in \eqref{eq:wall_hard}, and $m=2$ and $a=5$ in \eqref{eq:local_distance_function} for all our cases. For other aspects of boundary condition enforcement, we follow the methodology outlined in the previous sections. 
Lastly, by unnormalising the velocity by the characteristic length using $\mathrm{unnorm}^{\boldsymbol{u}}(z)=l_c\,z$ in \eqref{HC FBPINN ansatz}, we could achieve consistently good training for all test cases presented in this paper. A different choice of the unnormalisation parameter will still show satisfactory training, but with a slower convergence.
The same network structure is used for all test cases in this section. Each subdomain network has a single hidden layer with 32 neurons and $\tanh$ activation for both velocity and pressure. We use the Adam optimiser with an initial learning rate of 0.01 and an exponential decay rate of 0.9 every 2000 iterations. 
The penalty parameters for scaling the loss terms are updated every 200 iterations according to \eqref{standard loss balancing}. We train using a full set of collocation points (with varying numbers of points for the different test cases) and resample every 5000 iterations. We also adaptively sample 2000 points every 5000 iterations according to the probability density function in \eqref{eq:pdf}. These adaptive points are included in the total count when defining the number of collocation points; for example, for a total of 60,000 collocation points, 58,000 are sampled randomly and 2000 are sampled adaptively. 

\subsection{Periodic perforations}
Fig. \ref{fig:dom} shows a schematic of the geometry for a unit domain with 16 perforations with a diameter of 0.06 and an exemplary 11 by 11 decomposition of the domain. Each subdomain is centred on the vertices of the underlying grid. On the right of Fig. \ref{fig:dom} is a schematic representation of four overlapping subdomains that cover the region marked by the black square, for an overlap ratio of $\delta = 1.5$.
Since $L^2$ errors place a heavier weight on large errors, the higher magnitude flow around the perforations can dominate the error measure. To better quantify the errors within the microstructure, $L^1$ errors are evaluated in a smaller region confined to the perforations, marked as a grey square bounded by a dashed border in Fig. \ref{fig:dom}. The side lengths of this microstructure region are 0.46, 0.56 and 0.78 for the 16-, 36- and 64-perforation cases, respectively. We randomly sample 60,000 collocation points from the overall domain, with 20,000 additional points sampled in the microstructure region. 

\begin{figure}[!ht]
\centering
\includegraphics[width=0.65\linewidth]{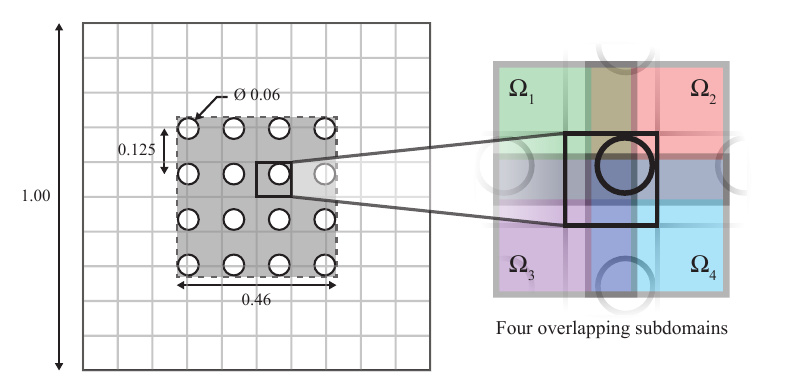}
\caption{\label{fig:dom} Unit-domain geometry for the 16-perforation case, showing the microstructure region in grey, where the $L^1$ errors are evaluated, and four representative overlapping subdomains.}
\end{figure}

In order to reduce the influence of spectral bias, Fourier feature embeddings have been used in previous studies, which transform the input variables into high-frequency signals before feeding them as input to the PINN \cite{wang2023expert,tancik2020fourier}. The application of Fourier features has been shown to improve the prediction of sharp gradients and reduce test errors significantly, but the scale of the input features, which is a user-specified hyperparameter (or alternatively, trainable, at the cost of an increased number of parameters), ideally needs to match the frequency range of interest, which may not be known \textit{a priori}. As a reference, we compare our FBPINN approach with regular PINNs (both soft- and hard-constrained variants) consisting of two hidden layers, each with 128 neurons, with 64 Fourier input features and random weight factorisation, following the implementation in \cite{wang2023expert}. We also consider FBPINNs with soft constraints to further investigate the influence of the hard constraints. The soft-constrained PINN and FBPINN models enforce the no-slip boundary condition on the surface of the perforations through soft penalisation as in \eqref{discrete losses}, using 200 randomly sampled collocation points around each perforation. The wall boundary conditions are enforced exactly in all models following \eqref{eq:wall_hard}. Note that the regular PINN model uses an initial learning rate of 0.001, as training was unstable with a higher learning rate. 
Furthermore, as discussed in Section \ref{section:nondim}, we set the initial loss weights as $\lambda_\mathrm{div}=({\mu_{\mathcal{D}} }/{l_\mathrm{c}})^{2}=278$ and $\lambda_{b}=({\mu_{\mathcal{D}} }/{l_\mathrm{c}}^2)^{2}=\num{7.72E04}$.

In the numerical experiments that follow, the size of each subdomain is chosen to be approximately 1.5 times the diameter of the perforations, leading to a domain consisting of 24 by 24 subdomains for an overlap ratio $\delta = 2$. This results in 576 subdomains and a total of 167,616 parameters to train for the FBPINN, whereas the regular PINN has 66,950 trainable parameters. The training time on an NVIDIA A100 GPU was 42 and 20 minutes for the regular PINN and the FBPINN (both soft and hard-constrained) models, respectively, for the 64-perforation case after 100,000 iterations. Combining domain decomposition with vectorised computation in the FBPINN model significantly reduces the computation time compared to regular PINNs, resulting in shorter training times despite the larger number of training parameters.


In Figs. \ref{fig:comparepinn16}, \ref{fig:comparepinn36} and \ref{fig:comparepinn64}, the $L^2$ errors in the overall domain and $L^1$ errors in the microstructure region using the soft- and hard-constrained variants of the regular PINN and FBPINN models are compared for 16, 36 and 64 periodic perforations, respectively. For the 16-perforation case, all four models can accurately model the multi-scale flow, with the FBPINN with hard constraints having the fastest training convergence and lowest $L^1$ errors. 
As the number of perforations (i.e., the frequency) in the domain increases, PINNs demonstrate a significant loss in performance, whereas hard-constrained FBPINNs maintain their superior accuracy and training efficiency with increasing numbers of perforations, achieving less than 5\% errors for all cases after only 10,000 iterations. On the other hand, soft-constrained FBPINNs show an increase in errors as training advances (most prominent for the 64-perforation case), indicating evidence of conflicting gradients between the different loss terms.

Although training convergence can be improved by applying the no-slip condition as hard constraints, conventional PINNs with hard constraints (i.e., without domain decomposition) portray a progressive increase in errors as the number of perforations increases. By comparing the results of hard-constrained FBPINNs with those of hard-constrained PINNs and soft-constrained FBPINNs, it is clear that tailored hard constraints, domain decomposition, and spatial localisation are all crucial for effective training that is less sensitive to the number of perforations in the domain.

\begin{figure}[H]
\centering
\includegraphics[width=\linewidth]{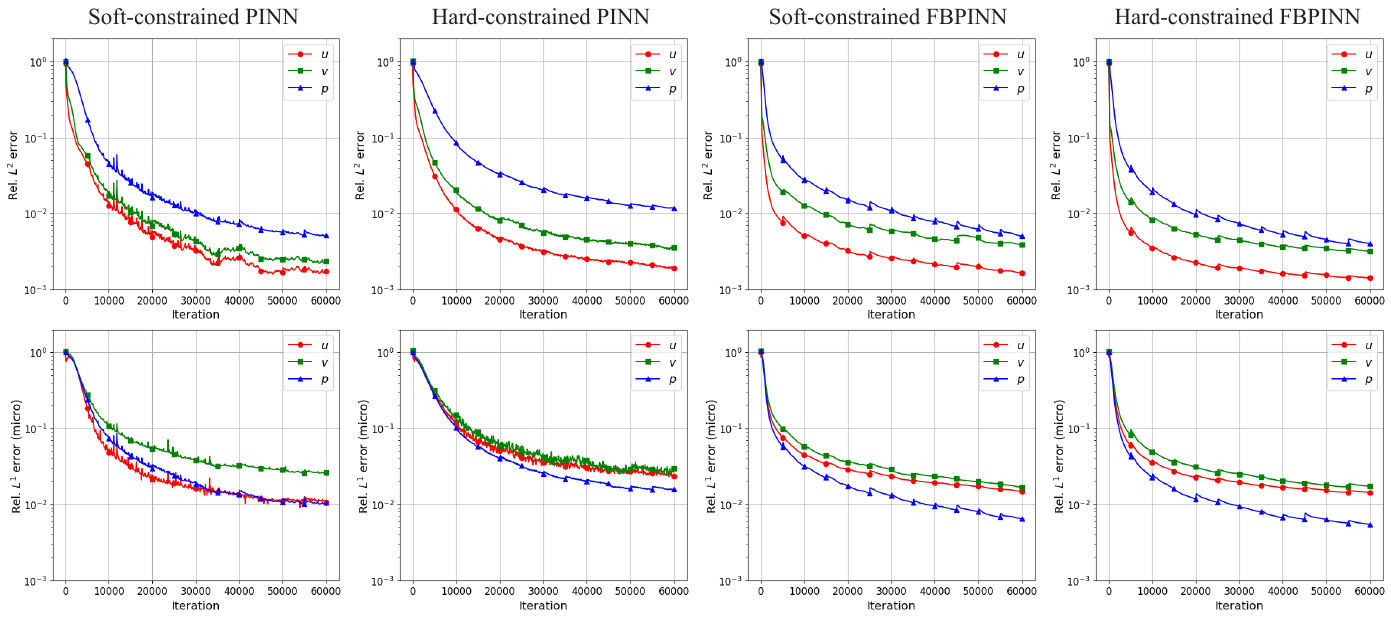}
\caption{\label{fig:comparepinn16} Comparison of the $L^2$ and microstructure $L^1$ errors for the soft- and hard-constrained variants of PINNs and FBPINNs for the 16-perforation case. The quantities $u$, $v$ and $p$ represent the errors in the horizontal velocity, vertical velocity, and pressure, respectively.}
\end{figure}

\begin{figure}[H]
\centering
\includegraphics[width=\linewidth]{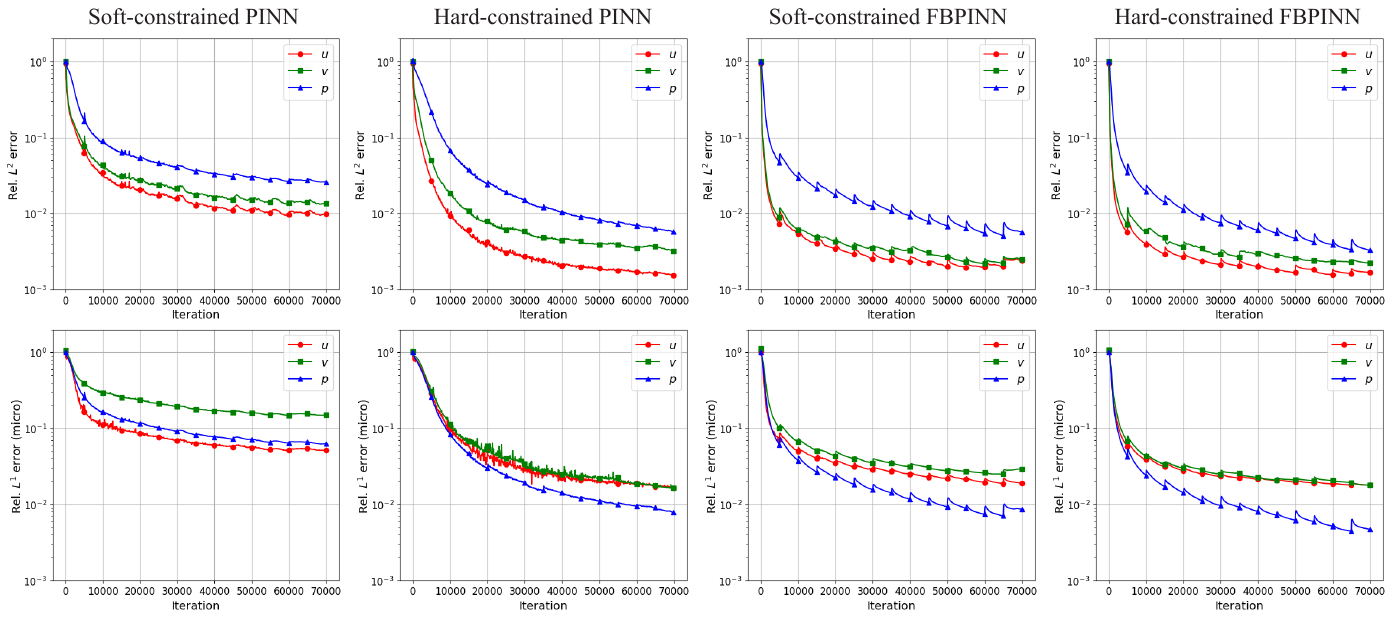}
\caption{\label{fig:comparepinn36} Comparison of the $L^2$ and microstructure $L^1$ errors for the soft- and hard-constrained variants of PINNs and FBPINNs for the 36-perforation case. The quantities $u$, $v$ and $p$ represent the errors in the horizontal velocity, vertical velocity, and pressure, respectively.}
\end{figure}

\begin{figure}[H]
\centering
\includegraphics[width=\linewidth]{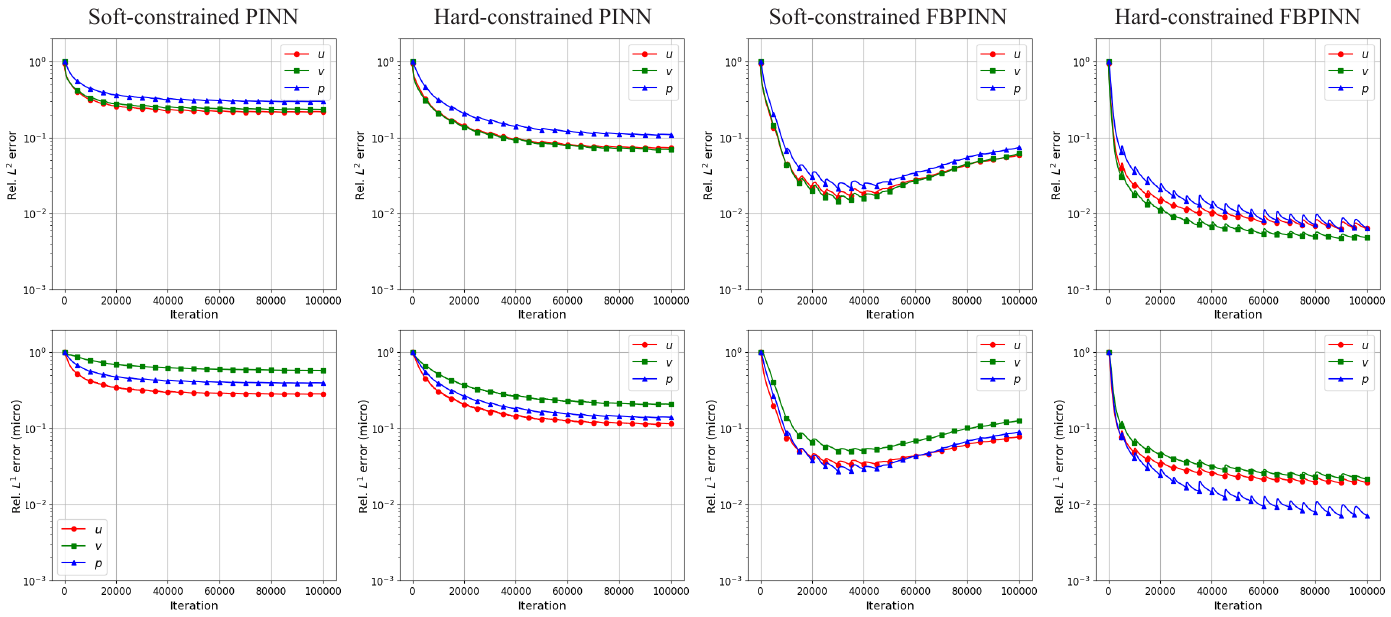}
\caption{\label{fig:comparepinn64}Comparison of the $L^2$ and microstructure $L^1$ errors for the soft- and hard-constrained variants of PINNs and FBPINNs for the 64-perforation case. The quantities $u$, $v$ and $p$ represent the errors in the horizontal velocity, vertical velocity, and pressure, respectively.}
\end{figure}

The contours of the velocity and pressure from the hard-constrained FBPINN models and reference FE solutions, with the corresponding pointwise absolute errors, are illustrated in Figs. \ref{fig:16_global}, \ref{fig:36_global} and \ref{fig:64_global} for 16, 36 and 64 perforations, respectively. The FBPINN model shows very good agreement with the FE results, regardless of the number of perforations. 
The solutions in the microstructure region are also presented in Figs. \ref{fig:16_local} and \ref{fig:64_local} for 16 and 64 perforations, respectively, showing that the FBPINN approach can accurately capture the complex flow patterns that arise in the highly perforated region. 

\begin{figure}[H]
\centering
\includegraphics[width=0.65\linewidth]{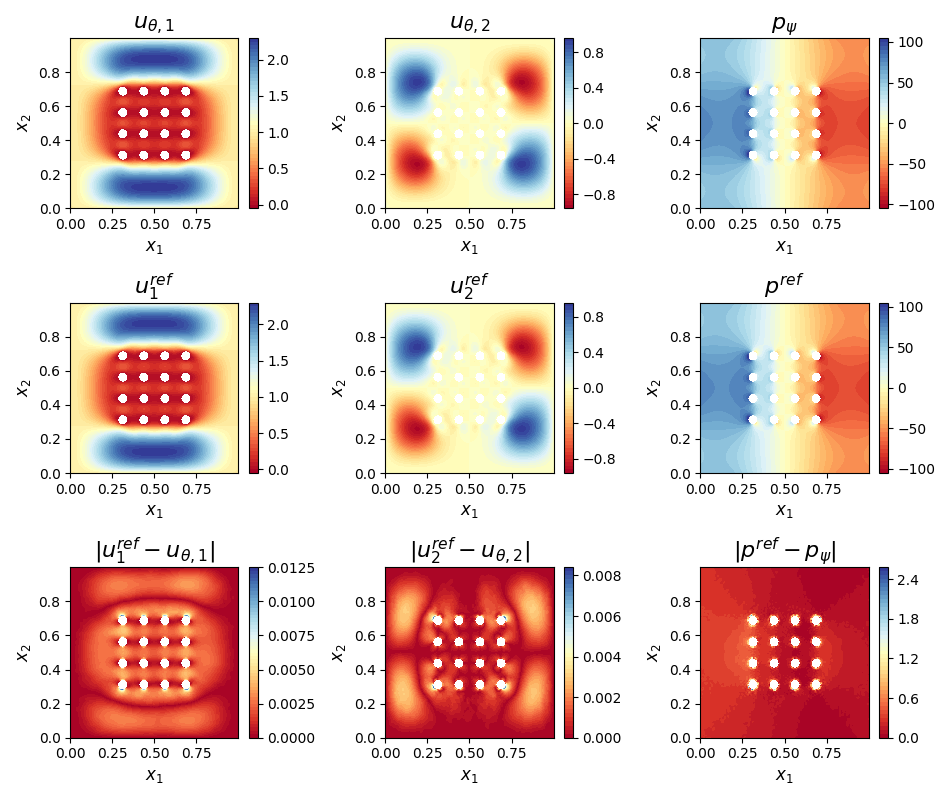}
\caption{\label{fig:16_global} Comparison of the hard-constrained FBPINN predictions and reference solution for the velocity ($u_1$: horizontal, $u_2$: vertical) and pressure in the periodic 16-perforation case.}
\end{figure}

\begin{figure}[H]
\centering
\includegraphics[width=0.65\linewidth]{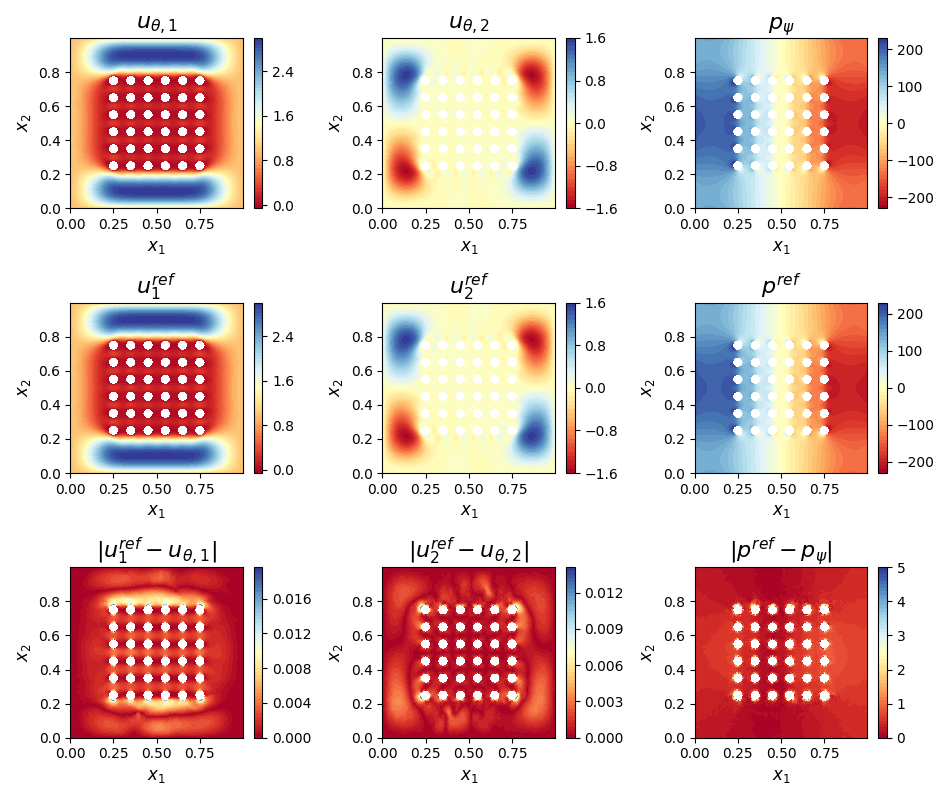}
\caption{\label{fig:36_global} Comparison of the hard-constrained FBPINN predictions and reference solution for the velocity ($u_1$: horizontal, $u_2$: vertical) and pressure in the periodic 36-perforation case.}
\end{figure}

\begin{figure}[H]
\centering
\includegraphics[width=0.65\linewidth]{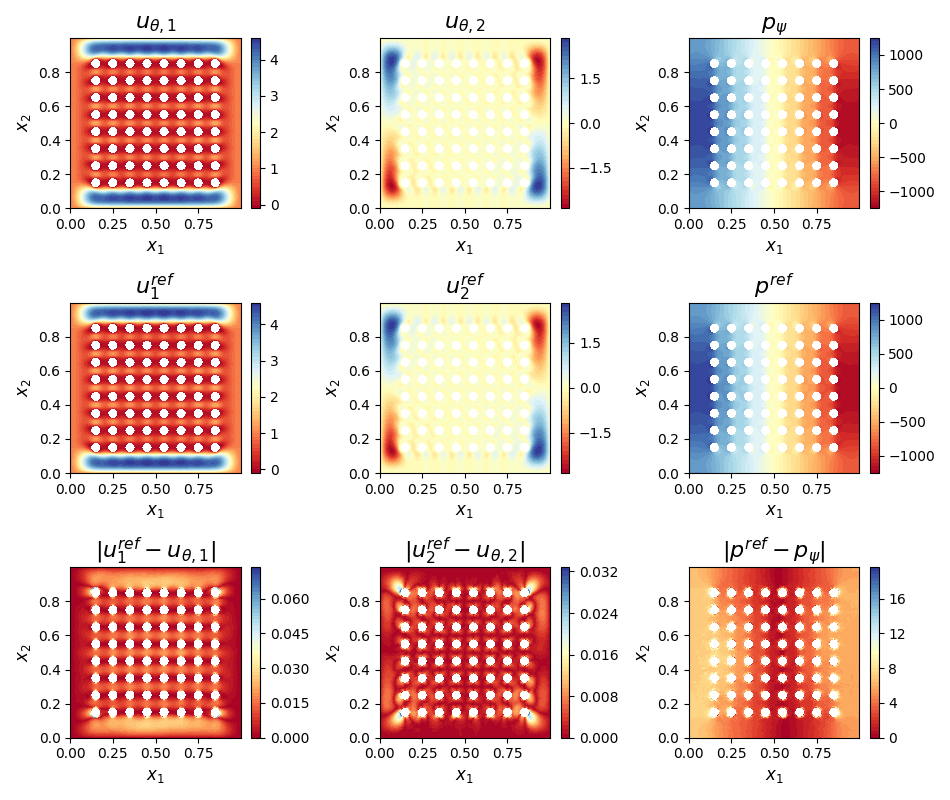}
\caption{\label{fig:64_global} Comparison of the hard-constrained FBPINN predictions and reference solution for the velocity ($u_1$: horizontal, $u_2$: vertical) and pressure in the periodic 64-perforation case.}
\end{figure}

\begin{figure}[H]
\centering
\includegraphics[width=0.65\linewidth]{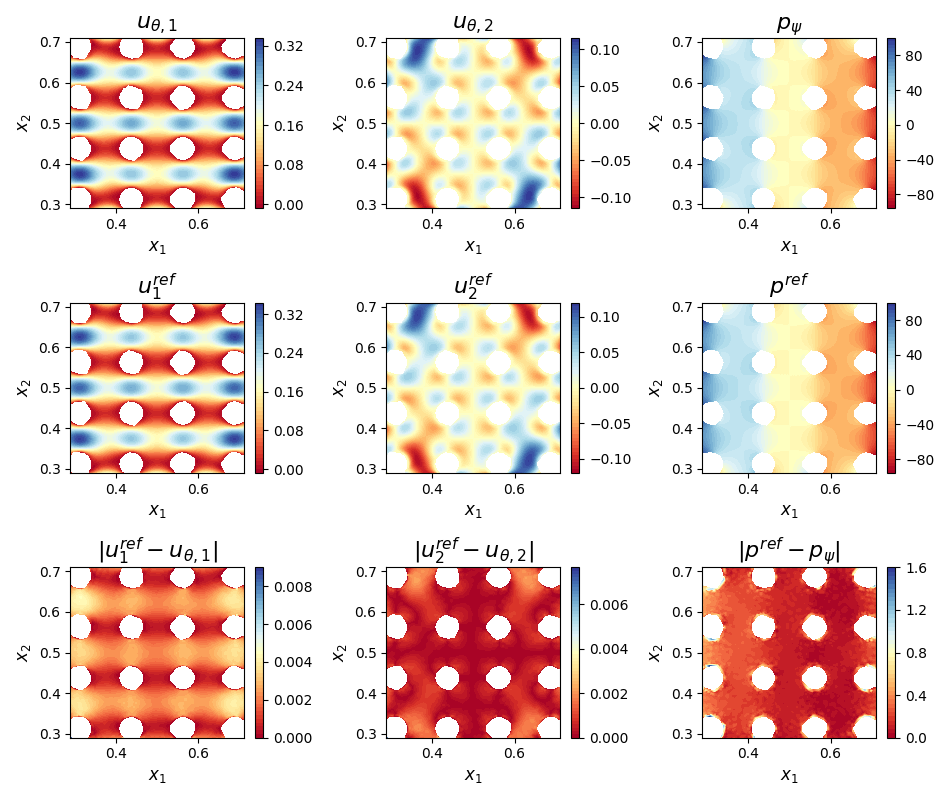}
\caption{\label{fig:16_local} Comparison of the hard-constrained FBPINN predictions and reference solution in the microstructure region for the velocity ($u_1$: horizontal, $u_2$: vertical) and pressure in the periodic 16-perforation case.}
\end{figure}


\begin{figure}[H]
\centering
\includegraphics[width=0.65\linewidth]{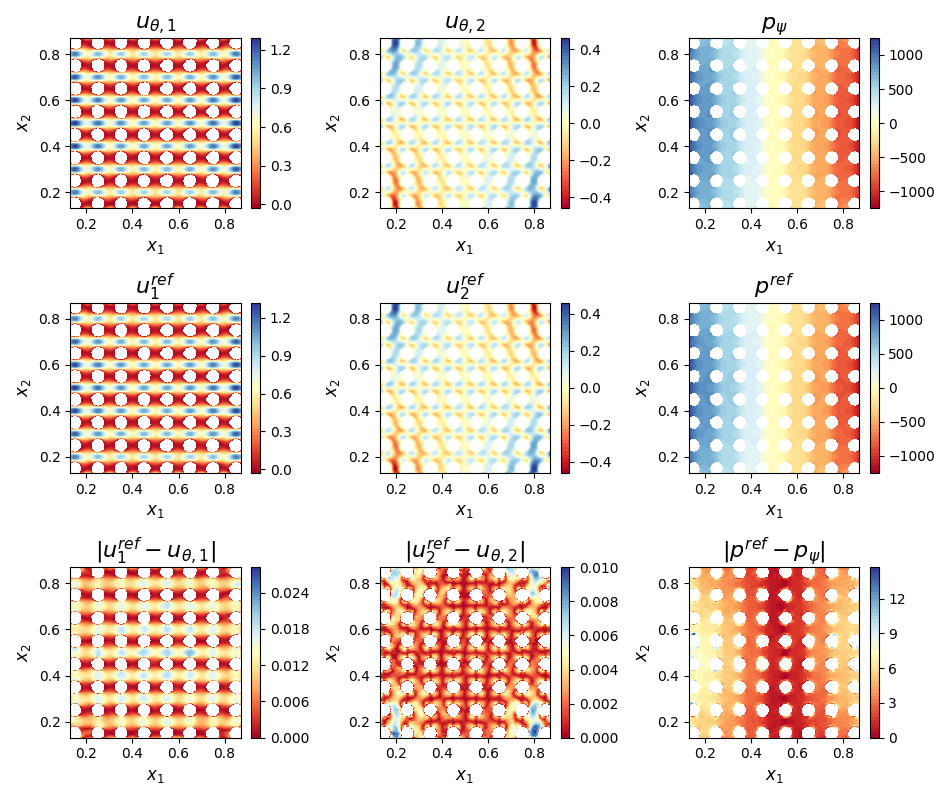}
\caption{\label{fig:64_local} Comparison of the hard-constrained FBPINN predictions and reference solution in the microstructure region for the velocity ($u_1$: horizontal, $u_2$: vertical) and pressure in the periodic 64-perforation case.}
\end{figure}


\subsubsection{Gradient conflicts during training of PINNs}
Two types of gradient conflict can arise during the training of PINNs: unbalanced gradient magnitudes and conflicting gradient directions \cite{wang2025gradient}.
Conflicts in the magnitudes of the loss gradients are mitigated through appropriate scaling of the physical problem through non-dimensionalisation and the adaptive loss weighting scheme \eqref{standard loss balancing}. 
Histograms of the back-propagated gradients $\nabla_{\theta} \mathcal{L}_{\rm div}(\theta)$, $\nabla_{\theta} \mathcal{L}_r(\theta, \psi)$, $\nabla_{\psi} \mathcal{L}_r(\theta, \psi)$, and $\nabla_{\theta} \mathcal{L}_{\rm b}(\theta)$ at 50,000 iterations for the 64-perforation case using the soft-constrained FBPINN model (Fig. \ref{fig:alignment}a), and a plot of the adaptive loss weights (Fig. \ref{fig:alignment}b, see "soft FBPINN"), illustrate how the loss weights $\lambda$ act to balance the loss gradients during training. Additionally, the plot of the loss weights in Fig. \ref{fig:alignment}b, given for all four models considered in this study, confirms that the initial estimates of the divergence-free and boundary loss weights described in Section \ref{section:nondim} were indeed close to the values obtained through the adaptive scheme, demonstrating dimensional analysis to be an effective approach for estimating the loss weights in physical problems. 

To illustrate the degree of directional conflict between the different loss gradients, the gradient alignment score (an extension of cosine similarity to multiple vectors) as defined in \cite{wang2025gradient} is evaluated. 
For vectors $v_1, v_2,...,v_n$, the alignment score is defined as
\begin{align}\label{alignmentscore}
\mathcal{A}(v_1, v_2,...,v_n) = 2 \left\lVert \frac{\sum_{i=1}^n \frac{v_i}{\lVert v_i \rVert} }{n} \right\rVert^2-1.
\end{align}
The alignment score takes values in $[-1, 1]$, where 1 indicates perfect alignment, and $-1$ indicates complete opposition.
The alignment scores, with exponential-moving-average smoothing, for PINNs and FBPINNs with soft and hard boundary constraints are plotted in Fig.~\ref{fig:alignment}c. 
For the soft-constrained FBPINN model, the score remains close to $-0.7$ throughout training, indicating poor alignment, which explains the increase in errors observed in Fig.~\ref{fig:comparepinn64}. In contrast, the score is positive for the first 35,000 iterations for the hard-constrained FBPINN model, aligning with the fast initial decrease in error observed in Fig.~\ref{fig:comparepinn64}. However, the hard-constrained PINN has a higher alignment score than the hard-constrained FBPINN, close to 0.7, which may seem to contradict the higher errors observed for the PINN in Fig.~\ref{fig:comparepinn64}. Since the alignment score only quantifies directional conflicts, and not the magnitudes of the loss gradients, it can be biased if the loss components become small or start to plateau. The loss curves for the hard-constrained PINN and FBPINN models are presented in Fig. \ref{fig:alignment}d, showing significantly lower losses for the hard-constrained FBPINN than the PINN model. As training progresses and the easy descent directions are exhausted, the loss terms increasingly compete for the remaining degrees of freedom in parameter space, and gradient conflict emerges, and is reflected as a decrease in alignment score for the hard-constrained FBPINN model. 

First-order optimisers move in the direction of the weighted sum of the loss gradients, making training susceptible to gradient conflicts and causing the optimiser to oscillate between multiple local minima. Furthermore, the large penalty weight $\lambda_b$ required to enforce the no-slip boundary condition in the soft-constrained models (Fig. \ref{fig:alignment}b) deteriorates the conditioning of the Hessian, resulting in a more complicated loss landscape and yielding an inefficient or unstable optimisation trajectory. In contrast, as shown by the error curves in Fig. \ref{fig:comparepinn64}, hard constraints enable much more efficient learning, leading to accurate predictions with $L^2$ errors of less than 1\%. 
\begin{figure}
\centering
\includegraphics[width=\linewidth]{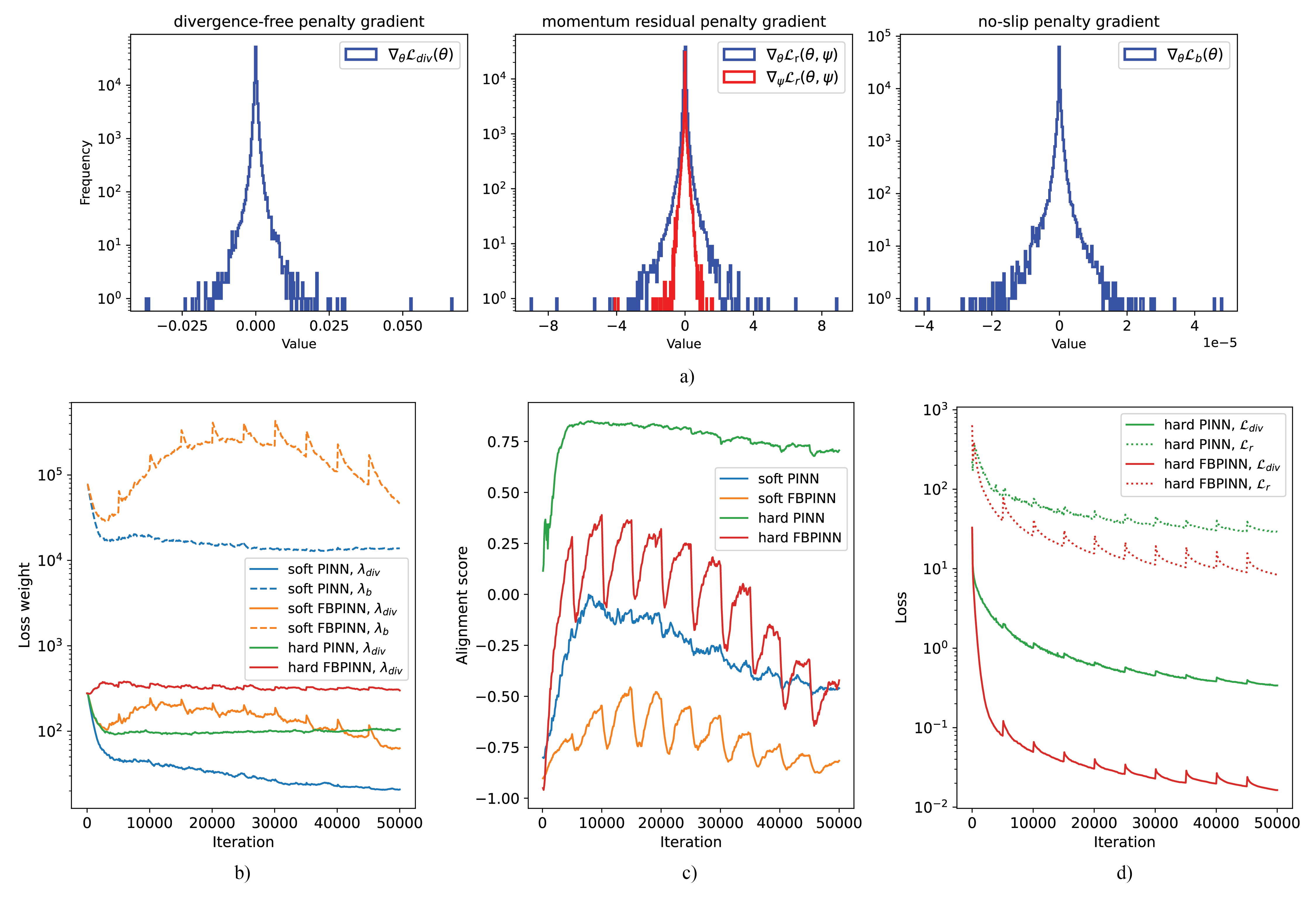}
\caption{\label{fig:alignment} Analysis of the results for the 64-perforation case. a) Histograms of the back-propagated loss gradients after 50,000 iterations for the soft-constrained FBPINN model. b) Adaptive loss weights and c) alignment scores for the soft- and hard-constrained variants of the PINN and FBPINN models. d) Loss curves for the hard-constrained PINN and FBPINN models.}
\end{figure}





\subsubsection{Sensitivity analysis on the degree of localisation in FBPINNs}
The training of hard-constrained FBPINNs is influenced by the degree of localisation introduced by both the amount of subdomain overlap controlled by the window functions, and the boundary distance functions used to enforce the hard constraints on the perforations.
To investigate the influence of these localisation parameters, we conduct multiple tests of the 64-perforation case, using five different random seeds for each test to assess the robustness of the models.
Table \ref{table:overlap_ratio} presents the $L^2$ errors in the overall domain, $L^1$ errors in the microstructure region and the compute time, averaged over the 5 random seeds, for subdomain overlap ratios $\delta$ = 1.2, 1.4, 2.0 and 2.4. As discussed in Section \ref{section:DD}, increasing the amount of overlap leads to more stable training and more accurate predictions, but at the cost of longer training times; increasing the overlap ratio from 2.0 to 2.4 leads to a 21\% reduction in $L^2$ errors on average, but a 58\% increase in the compute time. An overlap ratio of 2.0 is recommended for a good balance between training efficiency and accuracy.
Tables \ref{table:global_alpha} and \ref{table:pinn_alpha} provide results for hard-constrained FBPINNs and PINNs, respectively, for different values of the hard-constraint parameters: $a$, ranging from no $\tanh$ modulation to $a = 25$, and $m$, ranging from 2 to 12, where increasing both parameters increases the degree of localisation (Fig. \ref{fig:SDF design}). By comparing the case with no $\tanh$ modulation with the case with $a=5$, it is evident that stronger localisation of the distance function $l_{\partial \Omega^{\mathrm{p}}}$ leads to more accurate predictions. However, as the degree of localisation increases further with the increase in $a$ and $m$, the errors increase, most likely due to the stiffness introduced by the increased complexity of $\Delta l_{\partial \Omega^{\mathrm{p}}}$. This trend is consistent for both PINNs and FBPINNs. 

FBPINNs outperform conventional PINNs not only in accuracy and speed, but also in robustness to network initialisation. 
Tables \ref{table:fbpinn_alpha_l2_error_range} and \ref{table:pinn_alpha_l2_error_range} give the seed-wise minimum and maximum relative $L^2$ errors for hard-constrained FBPINNs and PINNs, respectively. While PINNs show good results for some seeds, the training is prone to instabilities that depend heavily on the choice of random seed, which controls the randomised Fourier input features. Overall, FBPINNs are more robust to the choice of random seed and consistently give highly accurate predictions.

\begin{table}[!ht]
\centering
\caption{\centering{Relative $L^2$ and $L^1$ errors and compute times for varying overlap ratios $\delta$ for 64 periodically arranged perforations. The results are averaged over the 5 seeds.}}
\label{table:overlap_ratio}
\small
\begin{tabular}{cccccccc}
\toprule
$\delta$ & \multicolumn{3}{c}{Relative $L^2$ error} &  \multicolumn{3}{c}{Relative $L^1$ error}  & Compute time \\
& $u_1$ & $u_2$ & $p$ & $u_1$ & $u_2$ & $p$ & (min)\\
\midrule
1.2 & \num{9.23E-2} & \num{1.45E-1} & \num{2.50E-1} & \num{3.02E-1} & \num{3.34E-1} & \num{2.52E-1} & 8.59 \\
1.4 & \num{2.45E-2} & \num{4.18E-2} & \num{6.32E-2} & \num{9.04E-2} & \num{1.08E-1} & \num{5.99E-2} & 10.32 \\
2.0 & \num{6.42E-03} & \num{4.90E-03} & \num{6.89E-03} & \num{1.97E-02} & \num{2.28E-02} & \num{7.57E-03} & 13.95 \\
2.4 & \num{5.22E-3} & \num{4.45E-3} & \num{4.43E-3} & \num{1.70E-2} & \num{1.81E-2} & \num{4.73E-3} & 22.00 \\
\bottomrule
\end{tabular}
\end{table}


\begin{table}[h!]
\centering
\caption{\centering{Effect of the global hard-constraint parameters $a$ and $m$ for hard-constrained FBPINNs. The first row represents a case without hyperbolic tangent modulation. The results are averaged over the 5 seeds.}}
\label{table:global_alpha}
\small
\begin{tabular}{ccccccccc}
\toprule
$a$ & $m$ & \multicolumn{3}{c}{Relative $L^2$ error} & \multicolumn{3}{c}{Relative $L^1$ error} & Compute time \\
& & $u_1$ & $u_2$ & $p$ & $u_1$ & $u_2$ & $p$ & (min)\\
\midrule
- & 2 & \num{8.05E-03} & \num{1.20E-02} & \num{1.51E-02} & \num{3.14E-02} & \num{3.72E-02} & \num{1.56E-02} & 13.85 \\
5 & 2 & \num{6.42E-03} & \num{4.90E-03} & \num{6.89E-03} & \num{1.97E-02} & \num{2.28E-02} & \num{7.57E-03} & 13.95 \\
10 & 4 & \num{5.85E-03} & \num{4.99E-03} & \num{6.19E-03} & \num{1.91E-02} & \num{2.32E-02} & \num{6.86E-03} & 14.04 \\
15 & 8 & \num{5.65E-03} & \num{6.81E-03} & \num{8.08E-03} & \num{2.17E-02} & \num{2.96E-02} & \num{8.65E-03} & 14.45 \\
25 & 12 & \num{1.54E-02} & \num{3.23E-02} & \num{3.60E-02} & \num{5.54E-02} & \num{1.05E-01} & \num{3.67E-02} & 14.48 \\
\bottomrule
\end{tabular}

\end{table}

\begin{table}[h!]
\centering
\caption{\centering{Effect of the global hard-constraint parameters $a$ and $m$ for hard-constrained PINNs. The first row represents a case without hyperbolic tangent modulation. The results are averaged over 5 random seeds.}}
\label{table:pinn_alpha}
\small
\begin{tabular}{ccccccccc}
\toprule
$a$ & $m$ & \multicolumn{3}{c}{Relative $L^2$ error} & \multicolumn{3}{c}{Relative $L^1$ error} & Compute time \\
& & $u_1$ & $u_2$ & $p$ & $u_1$ & $u_2$ & $p$ & (min)\\
\midrule
-  & 2  & \num{3.40E-01} & \num{5.38E-01} & \num{8.37E-01} & \num{8.55E-01} & \num{8.05E-01} & \num{8.58E-01} & 49.21 \\
5  & 2  & \num{7.10E-02} & \num{1.12E-01} & \num{1.60E-01} & \num{1.91E-01} & \num{1.93E-01} & \num{1.68E-01} & 50.05 \\
10 & 4  & \num{2.47E-01} & \num{3.31E-01} & \num{4.10E-01} & \num{4.00E-01} & \num{4.16E-01} & \num{4.16E-01} & 50.14 \\
15 & 8  & \num{3.12E-01} & \num{4.82E-01} & \num{8.25E-01} & \num{8.59E-01} & \num{8.51E-01} & \num{8.59E-01} & 50.25 \\
25 & 12 & \num{3.91E-01} & \num{5.59E-01} & \num{1.01E+00} & \num{9.77E-01} & \num{9.69E-01} & \num{1.00E+00} & 49.93 \\
\bottomrule
\end{tabular}
\end{table}

\begin{table}[h!]
\centering
\caption{\centering{Seed-wise minimum and maximum relative $L^2$ errors for hard-constrained FBPINNs.} }
\label{table:fbpinn_alpha_l2_error_range}
\small
\begin{tabular}{ccccc}
\toprule
$a$ & $m$ & \multicolumn{3}{c}{Relative $L^2$ error range}  \\
& & $u_1$ & $u_2$ & $p$ \\
\midrule
-  & 2
& [\num{7.31E-03}, \, \num{9.88E-03}]
& [\num{6.26E-03}, \, \num{1.71E-02}]
& [\num{1.25E-02}, \, \num{1.80E-02}]
\\

5  & 2
& [\num{6.23E-03}, \, \num{6.55E-03}]
& [\num{4.87E-03}, \, \num{5.00E-03}]
& [\num{6.50E-03}, \, \num{7.17E-03}]
\\

10 & 4
& [\num{5.66E-03}, \, \num{6.06E-03}]
& [\num{4.93E-03}, \, \num{5.12E-03}]
& [\num{5.77E-03}, \, \num{6.61E-03}]
\\

15 & 8
& [\num{5.58E-03}, \, \num{5.74E-03}]
& [\num{6.35E-03}, \, \num{7.83E-03}]
& [\num{7.64E-03}, \, \num{8.70E-03}]
\\

25 & 12
& [\num{1.39E-02}, \, \num{1.64E-02}]
& [\num{2.84E-02}, \, \num{3.55E-02}]
& [\num{3.26E-02}, \, \num{3.81E-02}]
\\
\bottomrule
\end{tabular}
\end{table}

\begin{table}[H]
\centering
\caption{\centering{Seed-wise minimum and maximum relative $L^2$ errors for hard-constrained PINNs.}}
\label{table:pinn_alpha_l2_error_range}
\small
\begin{tabular}{ccccc}
\toprule
$a$ & $m$ & \multicolumn{3}{c}{Relative $L^2$ error range}  \\
& & $u_1$ & $u_2$ & $p$ \\
\midrule
-  & 2
& $[\num{2.28E-01}, \, \num{4.27E-01}]$
& $[\num{3.54E-01}, \, \num{6.92E-01}]$
& $[\num{5.73E-01}, \, \num{9.99E-01}]$
\\

5  & 2
& $[\num{4.60E-03}, \, \num{3.34E-01}]$
& $[\num{4.80E-03}, \, \num{5.42E-01}]$
& $[\num{3.20E-03}, \, \num{7.79E-01}]$
\\

10 & 4
& $[\num{6.40E-03}, \, \num{7.74E-01}]$
& $[\num{7.40E-03}, \, \num{9.89E-01}]$
& $[\num{9.90E-03}, \, \num{1.13E+00}]$
\\

15 & 8
& $[\num{1.62E-01}, \, \num{3.61E-01}]$
& $[\num{2.68E-01}, \, \num{5.46E-01}]$
& $[\num{4.03E-01}, \, \num{9.85E-01}]$
\\

25 & 12
& $[\num{2.33E-01}, \, \num{6.63E-01}]$
& $[\num{2.64E-01}, \, \num{9.16E-01}]$
& $[\num{9.96E-01}, \, \num{1.05E+00}]$
\\
\bottomrule
\end{tabular}
\end{table}

\subsection{Random perforations}
In this section, three cases with random perforations are studied to represent more realistic arrangements of fibres within fibre bundles in composite reinforcements. In these random arrangements, there can be regions where fibres are less densely packed, leading to preferential flow pathways with higher velocities. 
The random arrangements are constructed by randomly generating each circular perforation and sequentially adding them to the domain without overlap. We consider 36, 64 and 100 perforations, progressively adding them to a subregion with a side length of 0.68 centred in the domain, to observe the effect of the increased frequency (or equivalently, the fibre volume fraction) on training convergence. The test case with 100 perforations corresponds to a fibre volume fraction of 40\%. 
In order to accommodate larger numbers of perforations in the unit domain, the fibre diameter is reduced to 0.05, and the number of subdomains in each dimension is increased to 28 accordingly, leading to a total of 784 subdomains and 228,144 parameters to train in total. In all cases, 120,000 collocation points are randomly sampled from the overall domain, with an additional 20,000 points in the microstructure region.

In Fig. \ref{fig:random_compare}, the $L^2$ errors and $L^1$ errors in the microstructure are plotted for 36, 64 and 100 random perforations. Although the errors increase slightly with increasing number of perforations, even with 100 perforations the hard-constrained FBPINN model demonstrates good training convergence, achieving $L^2$ errors of less than 5\% after only 10,000 iterations and less than 1\% after 70,000 iterations.

\begin{figure}
\centering
\includegraphics[width=0.87\linewidth]{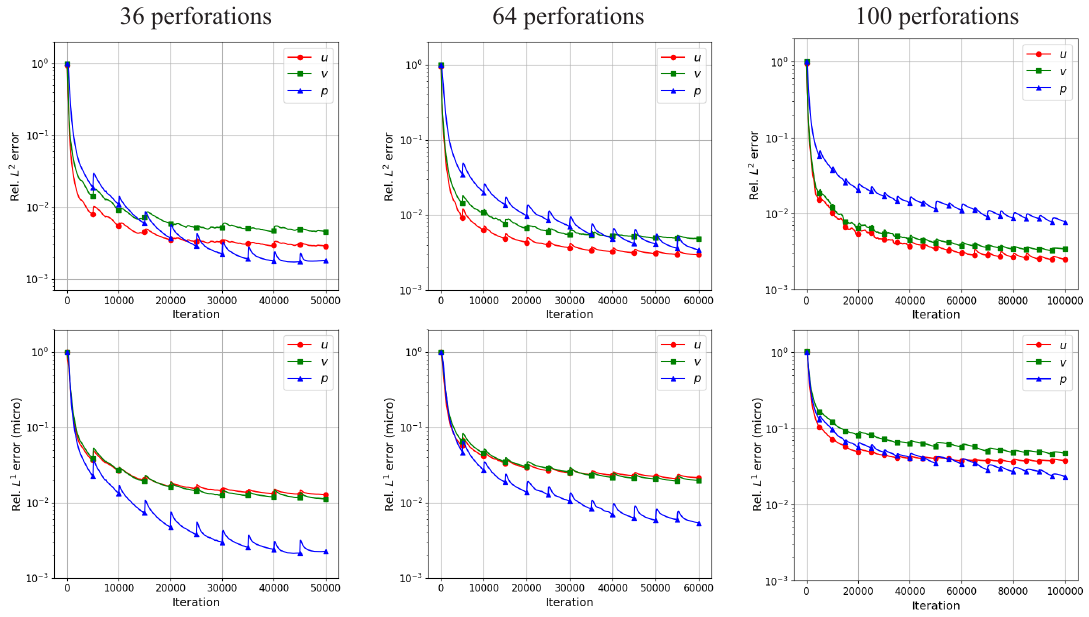}
\caption{\label{fig:random_compare} Comparison of the overall $L^2$ errors and microstructure $L^1$ errors for the hard-constrained FBPINN model in random arrangements of 36, 64, and 100 perforations. The quantities $u$, $v$, and $p$ denote the errors in the horizontal velocity, vertical velocity, and pressure, respectively.}
\end{figure}

Figs. \ref{fig:64rand_local} and \ref{fig:100rand_local} show the velocity and pressure contours in the microstructure region for 64 and 100 perforations, respectively. The predictions in the overall domain for all three random arrangements are provided in \ref{appendix:rand_solutions}. 
As the number of perforations increases, the resistance to flow increases, leading to a higher pressure gradient across the perforated region and lower velocities between the perforations. Consequently, both the increase in solution frequency within the microstructure and the difference in velocity magnitude between the microstructure and the surrounding region make such flows difficult for PINNs to approximate. We show, however, that our FBPINN approach captures the multi-scale fluid flow structure quite well.  


\begin{figure}[H]
\centering
\includegraphics[width=0.7\linewidth]{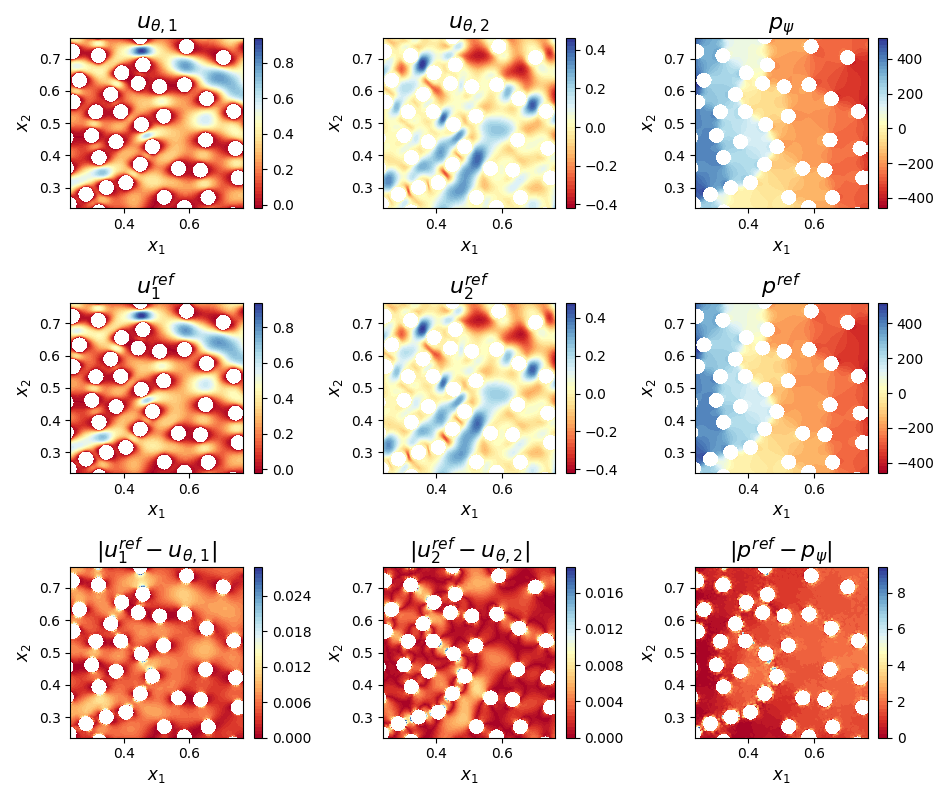}
\caption{\label{fig:64rand_local} Comparison of the hard-constrained FBPINN predictions and reference solution in the microstructure region for the velocity ($u_1$: horizontal, $u_2$: vertical) and pressure for a random arrangement of 64 perforations.}
\end{figure}
\begin{figure}[H]
\centering
\includegraphics[width=0.7\linewidth]{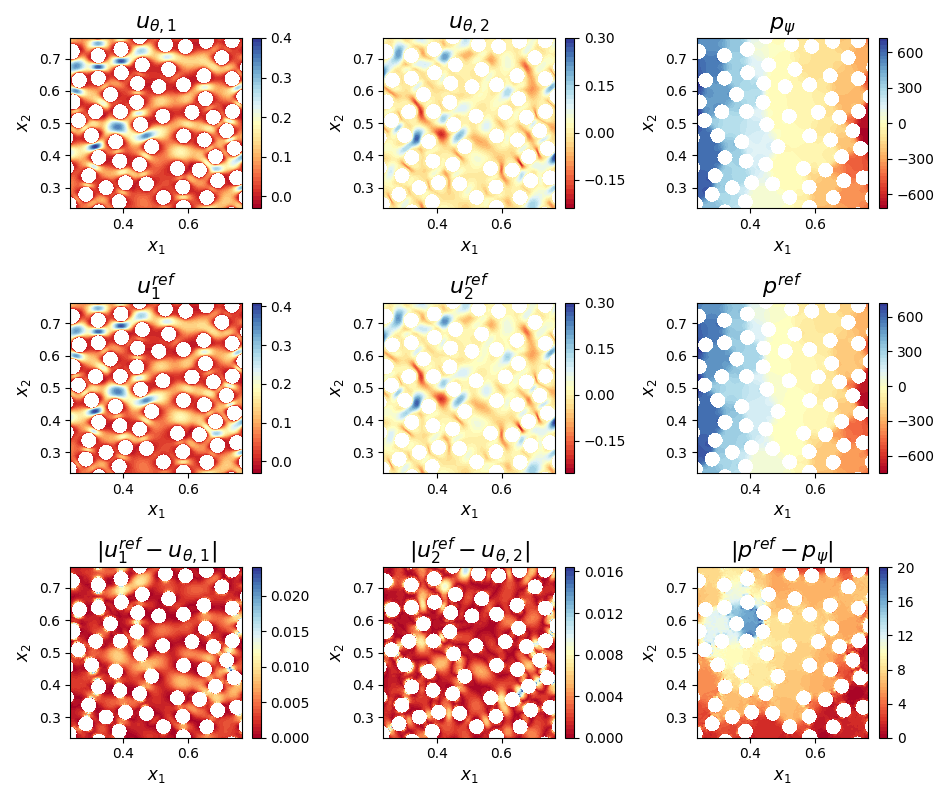}
\caption{\label{fig:100rand_local} Comparison of the hard-constrained FBPINN predictions and reference solution in the microstructure region for the velocity ($u_1$: horizontal, $u_2$: vertical) and pressure for a random arrangement of 100 perforations.}
\end{figure}

To verify that our approach provides accurate predictions of the flow field independent of the domain geometry, we perform further tests on different random arrangements of 100 perforations. Five different arrangements with the resulting vertical velocity component predictions are visualised in Fig.~\ref{fig:random_vel}. The $L^2$ errors and $L^1$ errors in the microstructure after 100,000 iterations of training are summarised in Table \ref{table:random_ablation}, showing good accuracy and only small variations for the different geometries.

\begin{table}[h!]
\centering
\caption{\centering{Relative $L^2$ and $L^1$ errors for different random arrangements of 100 perforations.}}
\small
\begin{tabular}{ccccccc}
\toprule
Test case & \multicolumn{3}{c}{Relative $L^2$ error} & \multicolumn{3}{c}{Relative $L^1$ error } \\
& $u_1$ & $u_2$ & $p$ & $u_1$ & $u_2$ & $p$ \\
\midrule
1 & \num{2.50E-03} & \num{3.42E-03} & \num{7.77E-03} 
& \num{3.78E-02} & \num{4.72E-02} & \num{2.32E-02} \\
2 & \num{2.83E-03} & \num{3.35E-03} & \num{5.17E-03} 
& \num{3.71E-02} & \num{3.61E-02} & \num{8.69E-03} \\
3 & \num{2.61E-03} & \num{3.40E-03} & \num{5.09E-03} 
& \num{3.92E-02} & \num{4.03E-02} & \num{9.85E-03} \\
4 & \num{3.74E-03} & \num{4.34E-03} & \num{9.37E-03} 
& \num{4.28E-02} & \num{4.84E-02} & \num{1.78E-02} \\
5 & \num{3.22E-03} & \num{3.71E-03} & \num{5.53E-03} 
& \num{3.66E-02} & \num{3.93E-02} & \num{9.36E-03} \\
\bottomrule
\end{tabular}
\label{table:random_ablation}
\end{table}

\begin{figure}[H]
\centering
\includegraphics[width=\linewidth]{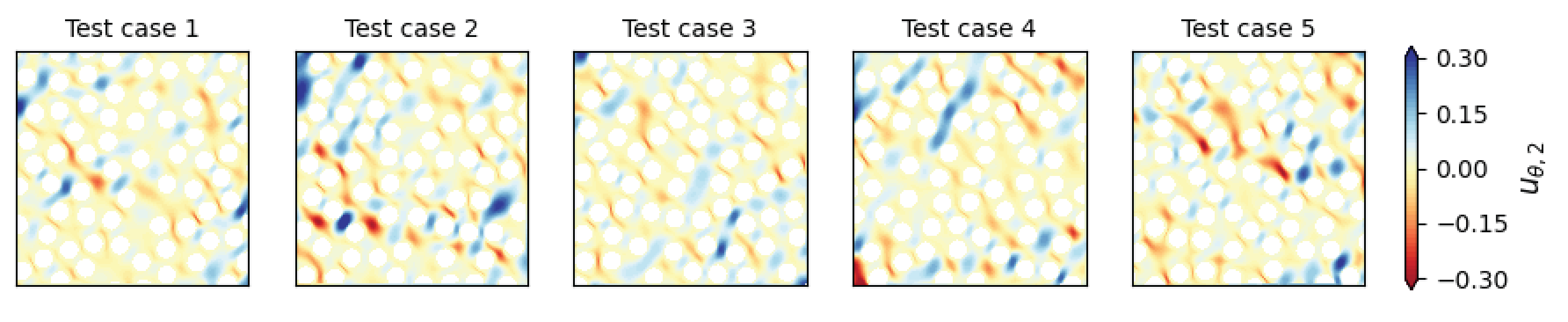}
\caption{\label{fig:random_vel} FBPINN model predictions for the vertical velocity component for five different random arrangements of 100 perforations.}
\end{figure}

        

\section{Conclusions}
The current study presents hard-constrained FBPINNs as a practical approach for approximating viscous fluid flow in highly perforated domains, demonstrating significantly faster training, higher accuracy, and better robustness -- to both network initialisation and the choice of hyperparameters -- compared to modern regular PINN architectures. Each ingredient of the method addresses a distinct source of difficulty.
Enforcing the no-slip condition as a hard constraint eliminates the associated penalty term and thereby reduces the stiffness of the discrete gradient flow and removes the boundary-induced gradient conflicts that we observe for the soft-constrained variants. The finite-basis ansatz mitigates spectral bias through the localisation of the window functions and the frequency rescaling induced by the local input normalisation, which replaces a single high-frequency global problem by a collection of lower-frequency local ones. Non-dimensionalisation, together with the output unnormalisation, places the viscous and pressure gradient contributions in the momentum residual on comparable scales, further reducing residual stiffness and yielding balanced learning of the velocity and pressure
fields. Hyperbolic tangent modulation of the aggregated distance function confines the no-slip constraint to the vicinity of the perforations and shapes a suitable geometric prior for the Laplacian of the constrained velocity. Finally, a carefully chosen subdomain overlap and residual-based adaptive collocation refinement improve accuracy while keeping the computational cost balanced. Together, these choices yield accurate predictions consistently across different numbers and arrangements of perforations, and show that hard constraints alone are not sufficient: domain decomposition and scaling principles are equally crucial for effective training.

\section*{Data availability}
The code required to reproduce the results will be made available upon publication of the manuscript.

\section*{Acknowledgements} D. Korolev and M. Duhovic acknowledge the support of the Federal Ministry of Education and Research, Germany (funding reference: 01IS24081) under project HybridSolver. J. Lee and S.S. Kim were supported by the National R\&D Program through the National Research Foundation of Korea (NRF) funded by the Ministry of Science and ICT (RS-2023-00260461).

\appendix
\section{Proofs}

\subsection{Proof of Proposition~\ref{UA proposition}}\label{APP1}

Let $\Omega \subset D$ denote the perforated domain, $\{D_i\}_{i=1}^{N}$ its overlapping rectangular subdomain partition, and $\{\omega_i\}_{i=1}^{N}$ a subordinate partition of unity satisfying $0 \leq \omega_i \leq 1$ and $\sum_{i=1}^{N} \omega_i = 1$. Define $\Omega_i := D_i \cap \Omega$ for each $i \in \{1, \dots, N\}$. Since the diameter of the perforations is smaller than the side length of $D_i$, each intersection of $\Omega$ with $D_i$ either forms a rectangular cut or a cut by a smooth arc segment of the perforated boundary. In both cases, this implies that each $\Omega_i$ is Lipschitz. Moreover, restricting the partition of unity $\{\omega_i\}_{i=1}^N$ on $D$ to $\Omega$ by setting $\omega_i^{r}:= \omega_i|_{\Omega}$ yields a partition of unity on $\Omega$ subordinate to $\Omega_i := D_i \cap \Omega$. With a slight abuse of notation, we then set $\omega_{i}:=\omega^{r}_{i}$ for each $i$.

We proceed by constructing an FBPINN ansatz $\boldsymbol{v}_{\theta}$ that can approximate $\boldsymbol{v} \in H^{s}(\Omega)^{d}$ in the Sobolev norm  $\lVert \,\cdot\,\rVert_{H^{s}(\Omega)^{d}}$. Let $T_i : \Omega_i \rightarrow \tilde{\Omega}_i$ be the normalisation map, satisfying the properties of a diffeomorphism. Then, restricting $\boldsymbol{v}$ to $\Omega_{i}$ preserves the Sobolev class, i.e, $\boldsymbol{v}^{i}:=\restr{\boldsymbol{v}}{\Omega_{i}} \in H^{s}(\Omega_{i})^{d}$. Since $T_{i}$ is a diffeomorphism, we have $\tilde{\boldsymbol{v}}^{i}:= \boldsymbol{v}^{i} \circ T_{i}^{-1} \in H^{s}(\tilde{\Omega}_{i})^{d}$, and each $\tilde{\Omega}_{i}:=T(\Omega_{i})$ is also Lipschitz. Recall that for each $k>s$, $H^{k}(\tilde{\Omega}_{i})^{d}$ is dense in $H^{s}(\tilde{\Omega}_{i})^{d}$. In addition, since $\tilde{\Omega}_{i}$ is a bounded Lipschitz domain, for each $\delta_{i}>0$ and each $\boldsymbol{w}^{i} \in H^k(\tilde{\Omega}_{i})^{d}$ with $\lVert \tilde{\boldsymbol{w}}^{i} - \tilde{\boldsymbol{v}}^{i} \rVert_{H^{s}(\tilde{\Omega}_{i})^{d}} < \delta_{i}/2$ (by density), there exists (see \cite[Theorem 4.9]{guhring2021approximation}) a fully connected $\tanh$ neural network $\boldsymbol{NN}^{i}: \tilde{\Omega}_{i} \rightarrow \mathbb{R}^{d}$ with parameters $\theta^{i}$ such that
\begin{align}\label{estimate 1}
\lVert \tilde{\boldsymbol{v}}^{i} - \boldsymbol{NN}^{i}\rVert_{H^{s}(\tilde{\Omega}_{i})^{d}} \leq \lVert \tilde{\boldsymbol{v}}^{i} - \tilde{\boldsymbol{w}}^{i} \rVert_{H^{s}(\tilde{\Omega}_{i})^{d}}  \, + \, \rVert \tilde{\boldsymbol{w}}^{i} - \boldsymbol{NN}^{i}\rVert_{H^{s}(\tilde{\Omega}_{i})^{d}} <    \delta_{i}.
\end{align}
In addition, by Sobolev norm equivalence under smooth changes of variables, we get 
\begin{align}\label{estimate 2}
\lVert \boldsymbol{v}^{i} - \boldsymbol{v}(\cdot; {\theta^{i}}) \rVert_{H^{s}(\Omega_{i})} \leq C(T_{i})\lVert  \tilde{\boldsymbol{v}}^{i} - \boldsymbol{NN}^{i} \rVert_{H^{s}(\tilde{\Omega}_{i})},
\end{align}
where $\boldsymbol{v}(\cdot; {\theta^{i}}):= \boldsymbol{NN}^{i} \circ T_{i}$ and $C(T_{i})>0$ is some constant which depends on $T_{i}$. Furthermore, since $\sum_{i=1}^{N} \omega_i = 1$, observe that $\boldsymbol{v} = \sum_{i=1}^{N} \omega_{i} \, \boldsymbol{v}$. Thus, 
\begin{align}\label{estimate 3}
\lVert \boldsymbol{v} - \boldsymbol{v}_{\theta}\rVert_{H^{s}(\Omega)^{d}} = \big\lVert \sum_{i=1}^{N} \omega_{i} (\boldsymbol{v} - \boldsymbol{v}_{\theta}) \big\rVert_{H^{s}(\Omega)^{d}} \leq \sum_{i=1}^{N} \lVert \omega_{i} (\boldsymbol{v} - \boldsymbol{v}_{\theta})\rVert_{H^{s}(\Omega)^{d}} \leq \sum_{i=1}^{N} C_{\omega_{i}} \lVert \boldsymbol{v}^{i} - \boldsymbol{v}(\cdot; \theta^{i})\rVert_{H^{s}(\Omega_{i})^{d}},
\end{align}
where $C_{\omega_{i}}>0$ is some constant which depends on the derivatives of $\omega_i$ up to the $s$-th order. Since $\delta_{i}$ was arbitrary, let $\varepsilon>0$ be arbitrary and fix $\delta_{i} = \frac{\varepsilon}{N C(T_{i}) C_{\omega_{i}}}$. For each $i \in \{1, \dots ,N \}$, construct $\boldsymbol{NN}_{i}$ satisfying \eqref{estimate 1} and consider the respective finite-basis neural network $\boldsymbol{v}_{\theta}:= \sum_{i=1}^{N} \omega_{i} \, \boldsymbol{v}(\cdot; \theta^{i})$ with $\theta = \{\theta^{i} \}_{i=1}^{N}$. By applying the estimates \eqref{estimate 1} and \eqref{estimate 2} to the estimate \eqref{estimate 3}, we obtain  
\begin{align*}
\lVert \boldsymbol{v} - \boldsymbol{v}_{\theta}\rVert_{H^{s}(\Omega)^{d}} \leq \sum_{i=1}^{N} C(T_{i}) \, C_{\omega_{i}} \, \lVert \tilde{\boldsymbol{v}}_{i} - \boldsymbol{NN}^{i}\rVert_{H^{s}(\tilde{\Omega}_{i})^{d}} <  \frac{\varepsilon}{N}\sum_{i=1}^{N} \,1 = \varepsilon.
\end{align*}
Since $\varepsilon$ was arbitrarily chosen, the claim follows. 

\subsection{Proof of Theorem~\ref{Fourier transform proposition}}\label{APP2}
Using the FBPINN ansatz and the linear unnormalisation, we have
\begin{align*}
\boldsymbol{v}_\theta(x)
=
\sum_{i=1}^{N}
\omega_i(x)
\left[
\alpha \,\boldsymbol{NN}_i\!\left(\frac{x-\mu_i}{\sigma_i}\right)
\right].
\end{align*}
The linearity of the Fourier transform and $\mathrm{supp}(\omega_i) \subset \Omega_i$ of each $\omega_i$ yield
\begin{align*}
\mathcal{F}[\boldsymbol{v}_\theta](\xi)
&=
\alpha_{\mathrm{un}} 
\sum_{i=1}^{N}
\int_{\Omega_{i}}
\omega_i(x)
\boldsymbol{NN}_i\!\left(\frac{x-\mu_i}{\sigma_i}\right)
e^{-2\pi \mathrm{i}x\cdot \xi}\,dx .
\end{align*}
We further use the local change of variables
\begin{align*}
\tilde{x}
=
\frac{x-\mu_i}{\sigma_i},
\qquad
x
=
\mu_i+\sigma_i\tilde{x},
\qquad
dx
=
\sigma_i^d\,d\tilde{x}.
\end{align*}
Thus, the subdomain $\Omega_i$ is mapped to the scaled subdomain $\tilde{\Omega}_i
:=\{\tilde{x}\in\mathbb{R}^d: \mu_i+\sigma_i\tilde{x}\in \Omega_i \}$. Define the local-coordinate partition-of-unity window $\tilde{\omega}_i(\tilde{x}):=\omega_i(\mu_i+\sigma_i\tilde{x})$. In these terms, we obtain
\begin{align*}
\int_{\Omega_{i}}
\omega_i(x)
\boldsymbol{NN}_i\!\left(\frac{x-\mu_i}{\sigma_i}\right)
e^{-2\pi \mathrm{i}x\cdot \xi}\,dx
=\sigma_i^d e^{-2\pi \mathrm{i}\mu_i\cdot \xi}
\int_{\tilde{\Omega}_i}
\tilde{\omega}_i(\tilde{x})
\boldsymbol{NN}_i(\tilde{x})
e^{-2\pi \mathrm{i}\tilde{x}\cdot(\sigma_i \,\xi)}
\,d\tilde{x}.
\end{align*}
Since $\operatorname{supp}(\tilde{\omega}_i)\subset \tilde{\Omega}_i$, we may regard $\tilde{\omega}_i\boldsymbol{NN}_i$ as a compactly supported function on $\mathbb R^d$. Therefore, we get
\begin{align*}
\int_{\tilde{\Omega}_i}
\tilde{\omega}_i(\tilde{x})
\boldsymbol{NN}_i(\tilde{x})
e^{-2\pi \mathrm{i}\tilde{x}\cdot(\sigma_i \,\xi)}
\,d\tilde{x}
=:
\mathcal{F}_{\tilde{x}}
[
\tilde{\omega}_i\boldsymbol{NN}_i
](\sigma_i \,\xi).
\end{align*}
Since the Fourier transform of the product of two functions is equal to the convolution of their Fourier transforms in Fourier space, we get
\begin{align*}
\mathcal{F}_{\tilde{x}}
\big[
\tilde{\omega}_i\boldsymbol{NN}_i
\big](\sigma_i\xi)
=
\bigg(
\mathcal{F}_{\tilde{x}}[\tilde{\omega}_i]
*
\mathcal{F}_{\tilde{x}}[\boldsymbol{NN}_i]
\bigg)(\sigma_i\xi).
\end{align*}
Combining the previous identities yields
\begin{align*}
\mathcal{F}[\boldsymbol{v}_\theta](\xi)
=\alpha_{\mathrm{un}}  \sum_{i=1}^{N}
\sigma_i^d e^{-2\pi \mathrm{i}\mu_i\cdot \xi}
\bigg(
\mathcal{F}_{\tilde{x}}[\tilde{\omega}_i]
*
\mathcal{F}_{\tilde{x}}[\boldsymbol{NN}_i]
\bigg)(\sigma_i\xi).
\end{align*}
This proves the claim.

\label{app1}

\section{Solutions for random arrangements of perforations} \label{appendix:rand_solutions}
The velocity and pressure contours for the random arrangements of 36, 64 and 100 perforations from Section 4.2 are given in Figs. \ref{fig:36rand_global}, \ref{fig:64rand_global} and \ref{fig:100rand_global}, respectively. 
\begin{figure}[H]
\centering
\includegraphics[width=0.7\linewidth]{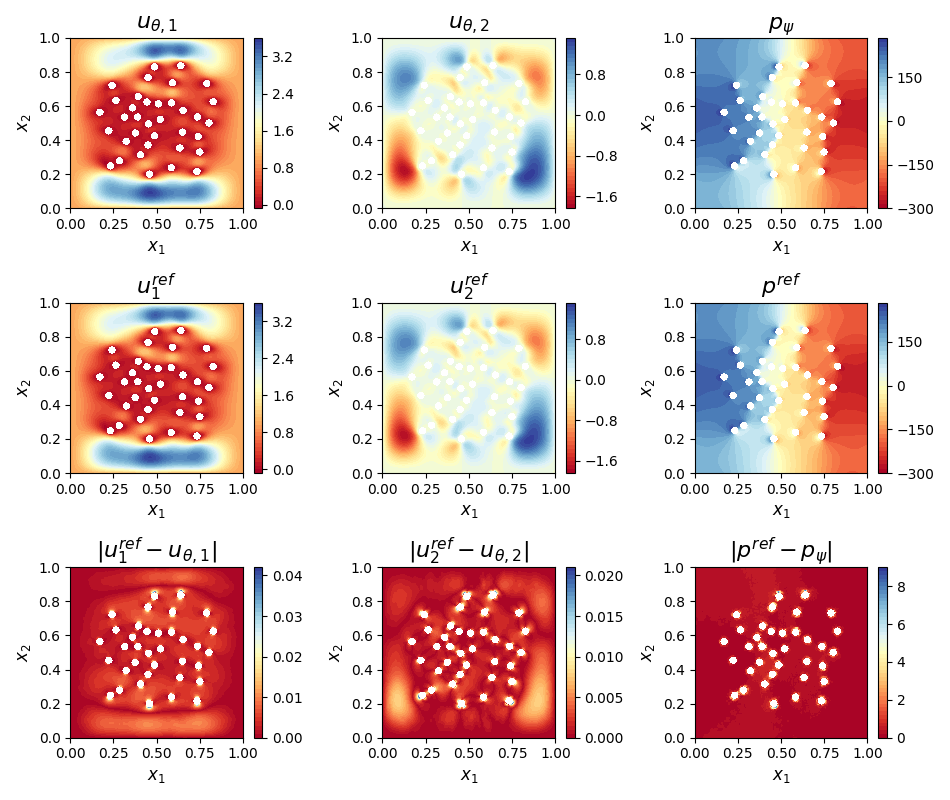}
\caption{\label{fig:36rand_global} Comparison of the hard-constrained FBPINN predictions and reference solution for the velocity ($u_1$: horizontal, $u_2$: vertical) and pressure for a random arrangement of 36 perforations}
\end{figure}
\begin{figure}[H]
\centering
\includegraphics[width=0.7\linewidth]{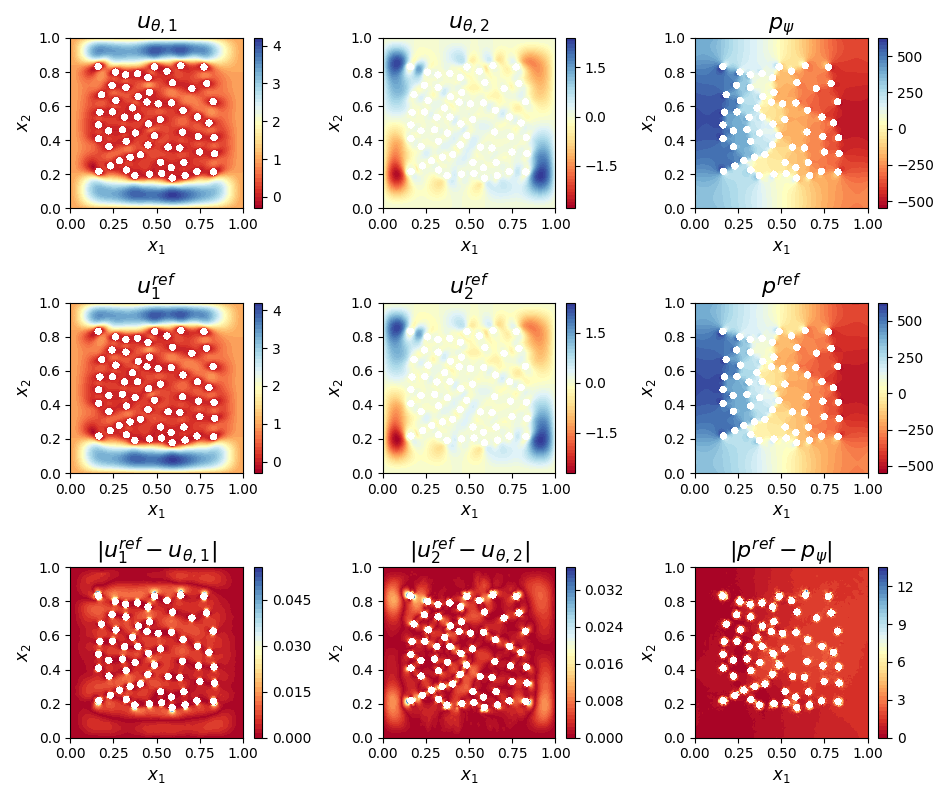}
\caption{\label{fig:64rand_global} Comparison of the hard-constrained FBPINN predictions and reference solution for the velocity ($u_1$: horizontal, $u_2$: vertical) and pressure for a random arrangement of 64 perforations}
\end{figure}
\begin{figure}[H]
\centering
\includegraphics[width=0.7\linewidth]{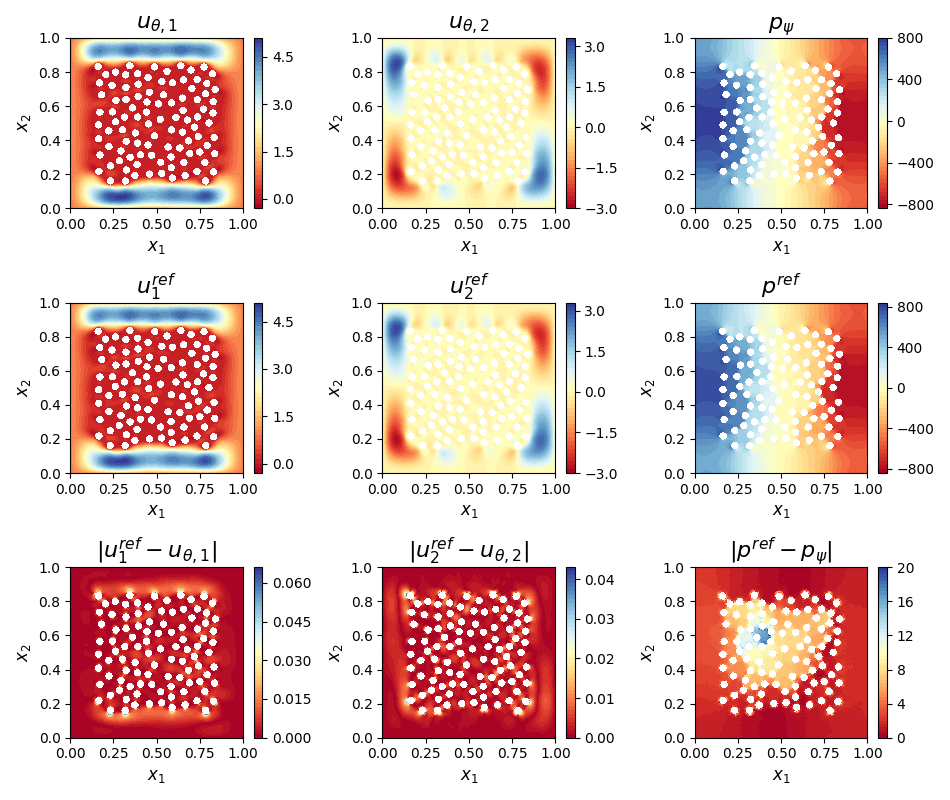}
\caption{\label{fig:100rand_global} Comparison of the hard-constrained FBPINN predictions and reference solution for the velocity ($u_1$: horizontal, $u_2$: vertical) and pressure for a random arrangement of 100 perforations}
\end{figure}



\bibliographystyle{elsarticle-num} 
\bibliography{bib.bib}

@article{moseley2023finite,
  title={Finite basis physics-informed neural networks ({FBPINNs}): a scalable domain decomposition approach for solving differential equations},
  author={Moseley, Ben and Markham, Andrew and Nissen-Meyer, Tarje},
  journal={Advances in Computational Mathematics},
  volume={49},
  number={4},
  pages={62},
  year={2023},
  publisher={Springer}
}

@article{hu2023augmented,
  title={Augmented Physics-Informed Neural Networks ({APINNs}): A gating network-based soft domain decomposition methodology},
  author={Hu, Zheyuan and Jagtap, Ameya D and Karniadakis, George Em and Kawaguchi, Kenji},
  journal={Engineering Applications of Artificial Intelligence},
  volume={126},
  pages={107183},
  year={2023},
  publisher={Elsevier}
}

@article{zeinhofer2025unified,
  title={A unified framework for the error analysis of physics-informed neural networks},
  author={Zeinhofer, Marius and Masri, Rami and Mardal, Kent--Andr{\'e}},
  journal={IMA Journal of Numerical Analysis},
  volume={45},
  number={5},
  pages={2988--3025},
  year={2025},
  publisher={Oxford University Press}
}

@article{de2024numerical,
  title={Numerical analysis of physics-informed neural networks and related models in physics-informed machine learning},
  author={De Ryck, Tim and Mishra, Siddhartha},
  journal={Acta Numerica},
  volume={33},
  pages={633--713},
  year={2024},
  publisher={Cambridge University Press}
}

@article{de2024error,
  title={Error estimates for physics-informed neural networks approximating the {N}avier--{S}tokes equations},
  author={De Ryck, Tim and Jagtap, Ameya D and Mishra, Siddhartha},
  journal={IMA Journal of Numerical Analysis},
  volume={44},
  number={1},
  pages={83--119},
  year={2024},
  publisher={Oxford University Press}
}

@article{mishra2023estimates,
  title={Estimates on the generalization error of physics-informed neural networks for approximating {PDE}s},
  author={Mishra, Siddhartha and Molinaro, Roberto},
  journal={IMA Journal of Numerical Analysis},
  volume={43},
  number={1},
  pages={1--43},
  year={2023},
  publisher={Oxford University Press}
}

@article{shin2023error,
  title={Error estimates of residual minimization using neural networks for linear {PDE}s},
  author={Shin, Yeonjong and Zhang, Zhongqiang and Karniadakis, George Em},
  journal={Journal of Machine Learning for Modeling and Computing},
  volume={4},
  number={4},
  year={2023},
  publisher={Begel House Inc.},
  pages={73-101}
}

@article{wang2021understanding,
  title={Understanding and mitigating gradient flow pathologies in physics-informed neural networks},
  author={Wang, Sifan and Teng, Yujun and Perdikaris, Paris},
  journal={SIAM Journal on Scientific Computing},
  volume={43},
  number={5},
  pages={A3055--A3081},
  year={2021},
  publisher={SIAM}
}

@article{wang2023expert,
  title={An expert's guide to training physics-informed neural networks},
  author={Wang, Sifan and Sankaran, Shyam and Wang, Hanwen and Perdikaris, Paris},
  journal={arXiv preprint arXiv:2308.08468},
  year={2023}
}

@article{korolev2026hybrid,
title = {Hybrid machine learning based scale bridging framework for permeability prediction of fibrous structures},
journal = {Composites Part A: Applied Science and Manufacturing},
volume = {202},
year = {2026},
issn = {1359-835X},
pages = {109458},
author = {Denis Korolev and Tim Schmidt and Dinesh K. Natarajan and Stefano Cassola and David May and Miro Duhovic and Michael Hintermüller}
}

@article{lu2021physics,
  title={Physics-informed neural networks with hard constraints for inverse design},
  author={Lu, Lu and Pestourie, Raphael and Yao, Wenjie and Wang, Zhicheng and Verdugo, Francesc and Johnson, Steven G},
  journal={SIAM Journal on Scientific Computing},
  volume={43},
  number={6},
  pages={B1105--B1132},
  year={2021},
  publisher={SIAM}
}

@article{liu2022unified,
  title={A unified hard-constraint framework for solving geometrically complex {PDEs}},
  author={Liu, Songming and Zhongkai, Hao and Ying, Chengyang and Su, Hang and Zhu, Jun and Cheng, Ze},
  journal={NeurIPS},
  volume={35},
  pages={20287--20299},
  year={2022}
}

@article{berg2018unified,
title = {A unified deep artificial neural network approach to partial differential equations in complex geometries},
journal = {Neurocomputing},
volume = {317},
pages = {28-41},
year = {2018},
author = {Berg, J. and Nyström, K},
}

@article{sun2020surrogate,
  title={Surrogate modeling for fluid flows based on physics-constrained deep learning without simulation data},
  author={Luning Sun and Han Gao and Shaowu Pan and Jian-Xun Wang},
  journal={Computer Methods in Applied Mechanics and Engineering},
  volume={361},
  pages={112732},
  year={2020},
}

@article{wang2023exact,
  title={Exact Dirichlet boundary physics-informed neural network {EPINN} for solid mechanics},
  author={Wang, Jiaji and Mo, YL and Izzuddin, Bassam and Kim, Chul-Woo},
  journal={Computer Methods in Applied Mechanics and Engineering},
  volume={414},
  pages={116184},
  year={2023},
  publisher={Elsevier}
}

@article{sukumar2022exact,
  title={Exact imposition of boundary conditions with distance functions in physics-informed deep neural networks},
  author={Sukumar, Natarajan and Srivastava, Ankit},
  journal={Computer Methods in Applied Mechanics and Engineering},
  volume={389},
  pages={114333},
  year={2022},
  publisher={Elsevier}
}

@article{bischof2025multi,
  title={Multi-objective loss balancing for physics-informed deep learning},
  author={Bischof, Rafael and Kraus, Michael A},
  journal={Computer Methods in Applied Mechanics and Engineering},
  volume={439},
  pages={117914},
  year={2025},
  publisher={Elsevier}
}

@inproceedings{rathore2024challenges,
  author    = {Rathore, Pratik and Lei, Weimu and Frangella, Zachary and Lu, Lu and Udell, Madeleine},
  title     = {Challenges in training {PINNs}: a loss landscape perspective},
  booktitle = {Proceedings of the 41st International Conference on Machine Learning},
  series    = {ICML},
  year      = {2024},
  address   = {Vienna, Austria}
}

@article{wright1999numerical,
  title={Numerical optimization},
  author={Wright, Stephen and Nocedal, Jorge},
  journal={Springer Science},
  volume={35},
  number={67-68},
  pages={7},
  year={1999}
}

@article{kingma2014adam,
  title={Adam: A method for stochastic optimization},
  author={Kingma, Diederik P},
  journal={arXiv preprint arXiv:1412.6980},
  year={2014}
}

@article{straub2025hard,
  title={Hard-constraining {Neumann} boundary conditions in physics-informed neural networks via {Fourier} feature embeddings},
  author={Straub, Christopher and Brendel, Philipp and Medvedev, Vlad and Rosskopf, Andreas},
  journal={arXiv preprint arXiv:2504.01093},
  year={2025}
}

@article{xie2024physics,
  title={Physics-specialized neural network with hard constraints for solving multi-material diffusion problems},
  author={Xie, Yuchen and Chi, Honghang and Wang, Yahui and Ma, Yu},
  journal={Computer Methods in Applied Mechanics and Engineering},
  volume={430},
  pages={117223},
  year={2024},
  publisher={Elsevier}
}

@article{raissi2019physics,
  title={Physics-informed neural networks: A deep learning framework for solving forward and inverse problems involving nonlinear partial differential equations},
  author={Raissi, Maziar and Perdikaris, Paris and Karniadakis, George E},
  journal={Journal of Computational physics},
  volume={378},
  pages={686--707},
  year={2019},
  publisher={Elsevier}
}

@article{toscano2025pinns,
  title={From {PINNs} to {PIKANs}: {R}ecent advances in physics-informed machine learning},
  author={Toscano, Juan Diego and Oommen, Vivek and Varghese, Alan John and Zou, Zongren and Ahmadi Daryakenari, Nazanin and Wu, Chenxi and Karniadakis, George Em},
  journal={Machine Learning for Computational Science and Engineering},
  volume={1},
  number={1},
  pages={1--43},
  year={2025},
  publisher={Springer}
}

@article{hintermuller2026hybrid,
  title={A hybrid physics-informed neural network based multiscale solver as a partial differential equation constrained optimization problem},
  author={Hinterm{\"u}ller, Michael and Korolev, Denis},
  journal={ESAIM: Control, Optimisation and Calculus of Variations},
  volume={32},
  pages={18},
  year={2026},
  publisher={EDP Sciences}
}

@book{johnson2009numerical,
  title={Numerical solution of partial differential equations by the finite element method},
  author={Johnson, Claes},
  year={2009},
  publisher={Courier Corporation}
}

@article{rohrhofer2023data,
  title={Data vs. {P}hysics: {T}he {A}pparent {P}areto {F}ront of {P}hysics-{I}nformed {N}eural {N}etworks}, 
  author={Rohrhofer, Franz M. and Posch, Stefan and Gößnitzer, Clemens and Geiger, Bernhard C.},
  journal={IEEE Access}, 
  volume={11},
  pages={86252-86261},
  year={2023}
}

@article{lu2021deepxde,
author = {Lu, Lu and Meng, Xuhui and Mao, Zhiping and Karniadakis, George Em},
title = {Deep{XDE}: {A} {D}eep {L}earning {L}ibrary for solving {D}ifferential {E}quations},
journal = {SIAM Review},
volume = {63},
number = {1},
pages = {208-228},
year = {2021},
}

@article{wu2023rar,
author = {Chenxi Wu and Min Zhu and Qinyang Tan and Yadhu Kartha and Lu Lu},
title = {A comprehensive study of non-adaptive and residual-based adaptive sampling for physics-informed neural networks},
journal = {Computer Methods in Applied Mechanics and Engineering},
volume = {403},
pages = {115671},
year = {2023},
}

@article{schmidt2025numerical,
  title={Numerical data generation for building machine learning models for permeability estimation of fibrous structures},
  author={Schmidt, Tim and Natarajan, Dinesh Krishna and Duhovic, Miro and Cassola, Stefano and Nuske, Marlon and May, David},
  journal={Polymer Composites},
  year={2025},
  pages={S104-S120},
  volume = {46},
}

@article{yong2025permbenchmark,
title = {Towards standardisation of the out-of-plane permeability measurement for reinforcement textiles},
journal = {Composites Part A: Applied Science and Manufacturing},
volume = {190},
pages = {108630},
year = {2025},
author = {A.X.H. Yong and A. Endruweit and A. George and D. May and Y.A. Aksoy and M.A. Ali and T. Allen and P. Baral and C. Betteridge and C. Brauner and others},
}

@article{syerko2023imagebased,
title = {Benchmark exercise on image-based permeability determination of engineering textiles: Microscale predictions},
journal = {Composites Part A: Applied Science and Manufacturing},
volume = {167},
pages = {107397},
year = {2023},
author = {E. Syerko and T. Schmidt and D. May and C. Binetruy and S.G. Advani and S. Lomov and L. Silva and S. Abaimov and N. Aissa and I. Akhatov and others},
}

@article{annamalai2025extension,
  title={On the extension of in-plane permeability calibration to out-of-plane measurements: {A}dvancements in additively manufactured textile-structured porous media for liquid composite moulding},
  author={Annamalai, Prabakaran and Gangipamula, Venkatesh and Ashebir, Demeke Abay and Sattar, Md Abdus and Lomov, Stepan V and May, David and Bodaghi, Masoud and Nikzad, Mostafa},
  journal={Composites Part A: Applied Science and Manufacturing},
  volume={193},
  pages={108863},
  year={2025},
  publisher={Elsevier}
}

@article{jo2024permeability,
title = {Prediction of transverse permeability in representative volume elements with closely arranged fibers through the application of {D}elaunay-triangulation and electrical-circuit analogy},
journal = {Composite Structures},
volume = {334},
pages = {117984},
year = {2024},
author = {Hyeonseong Jo and Sangyoon Bae and Hyunsoo Hong and Wonvin Kim and Seong Su Kim},
}

@article{caglar2022deeplearning,
title = {Deep learning accelerated prediction of the permeability of fibrous microstructures},
journal = {Composites Part A: Applied Science and Manufacturing},
volume = {158},
pages = {106973},
year = {2022},
author = {Baris Caglar and Guillaume Broggi and Muhammad A. Ali and Laurent Orgéas and Véronique Michaud},
}

@article{jean2026imagebased,
title = {An image-based deep learning framework for flow field prediction in arbitrary-sized fibrous microstructures},
journal = {Composites Part A: Applied Science and Manufacturing},
volume = {200},
pages = {109337},
year = {2026},
author = {Jimmy Gaspard Jean and Guillaume Broggi and Baris Caglar},
}

@article{LEE2025108857,
title = {Physics-informed neural networks for real-time simulation of transverse {L}iquid {C}omposite {M}oulding processes and permeability measurements},
journal = {Composites Part A: Applied Science and Manufacturing},
volume = {193},
pages = {108857},
year = {2025},
author = {J. Lee and M. Duhovic and D. May and T. Allen and P. Kelly},
}

@article{wu2024semiconductor,
    author  = "Erjun Wu and Bo Wang and Shuai Zhang and Yu Su and Xiaodong Chen",
    title   = "Microscale underfill dynamics and void formation of high-density
flip-chip packaging: Experiments and simulations",
    year    = "2024",
    journal = "Physics of Fluids",
    volume  = "36",
    pages   = "032117"
}

@article{asif2024heatexchanger,
    author  = "Asif, M. and Jamshed, S. and Dhiman, A.K",
    title   = "Heat transfer across an array of cylinders arranged in inline and staggered formation in a heat exchanger: {E}ffect of nanoparticle volume fraction, nanoparticle diameter, and {R}ichardson number",
    year    = "2024",
    journal = "The European Physical Journal Plus",
    volume  = "139",
    pages   = "601"
}

@article{horgue2013microreactor,
    author  = "Pierre Horgue and Frédéric Augier and Paul Duru and Marc Prat and Michel Quintard",
    title   = "Experimental and numerical study of two-phase flows in arrays
of cylinders",
    year    = "2013",
    journal = "Chemical Engineering Science",
    volume  = "102",
    pages   = "335-345"
}

@article{krishnamurthy2007microreactor,
    author  = "Santosh Krishnamurthy and Yoav Peles",
    title   = "Gas-liquid two-phase flow across a bank of micropillars",
    year    = "2007",
    journal = "Physics of Fluids",
    volume  = "19",
    pages   = "043302"
}

@article{shukla2025neurosem,
  title={Neuro{SEM}: {A} hybrid framework for simulating multiphysics problems by coupling {PINN}s and spectral elements},
  author={Shukla, Khemraj and Zou, Zongren and Chan, Chi Hin and Pandey, Additi and Wang, Zhicheng and Karniadakis, George Em},
  journal={Computer Methods in Applied Mechanics and Engineering},
  volume={433},
  pages={117498},
  year={2025},
  publisher={Elsevier}
}

@book{hunter2001applied,
  title={Applied analysis},
  author={Hunter, John K and Nachtergaele, Bruno},
  year={2001},
  publisher={World Scientific}
}

@article{wolf2022homogenization,
  title={Homogenization of the {S}tokes system in a non-periodically perforated domain},
  author={Wolf, Sylvain},
  journal={Multiscale Modeling \& Simulation},
  volume={20},
  number={1},
  pages={72--106},
  year={2022},
  publisher={SIAM}
}

@inproceedings{DBLP:conf/iclr/RyckBMB24,
  author       = {Tim De Ryck and
                  Florent Bonnet and
                  Siddhartha Mishra and
                  Emmanuel de B{\'{e}}zenac},
  title        = {An operator preconditioning perspective on training in physics-informed
                  machine learning},
  booktitle    = {ICLR},
  year         = {2024},

}

@article{cai2021physics,
  title={Physics-informed neural networks ({PINNs}) for fluid mechanics: A review},
  author={Cai, Shengze and Mao, Zhiping and Wang, Zhicheng and Yin, Minglang and Karniadakis, George Em},
  journal={Acta Mechanica Sinica},
  volume={37},
  number={12},
  pages={1727--1738},
  year={2021},
  publisher={Springer}
}

@article{cai2021flow,
  title={Flow over an espresso cup: inferring {3-D} velocity and pressure fields from tomographic background oriented {S}chlieren via physics-informed neural networks},
  author={Cai, Shengze and Wang, Zhicheng and Fuest, Frederik and Jeon, Young Jin and Gray, Callum and Karniadakis, George Em},
  journal={Journal of Fluid Mechanics},
  volume={915},
  pages={A102},
  year={2021},
  publisher={Cambridge University Press}
}

@article{wang2025simulating,
  title={Simulating three-dimensional turbulence with physics-informed neural networks},
  author={Wang, Sifan and Sankaran, Shyam and Fan, Xiantao and Stinis, Panos and Perdikaris, Paris},
  journal={arXiv preprint arXiv:2507.08972},
  year={2025}
}

@article{zhu2024physics,
  title={Physics-informed neural networks for incompressible flows with moving boundaries},
  author={Zhu, Yongzheng and Kong, Weizhen and Deng, Jian and Bian, Xin},
  journal={Physics of Fluids},
  volume={36},
  number={1},
  year={2024},
  publisher={AIP Publishing}
}

@article{botarelli2025using,
  title={Using {P}hysics-{I}nformed neural networks for solving {Navier-Stokes} equations in fluid dynamic complex scenarios},
  author={Botarelli, Tommaso and Fanfani, Marco and Nesi, Paolo and Pinelli, Lorenzo},
  journal={Engineering Applications of Artificial Intelligence},
  volume={148},
  pages={110347},
  year={2025},
  publisher={Elsevier}
}

@article{griebel2010homogenization,
  title={Homogenization and numerical simulation of flow in geometries with textile microstructures},
  author={Griebel, Michael and Klitz, Margit},
  journal={Multiscale Modeling \& Simulation},
  volume={8},
  number={4},
  pages={1439--1460},
  year={2010},
  publisher={SIAM}
}

@article{bodaghi2016statistics,
  title={On the statistics of transverse permeability of randomly distributed fibers},
  author={Bodaghi, Masoud and Catalanotti, Giuseppe and Correia, Nuno},
  journal={Composite Structures},
  volume={158},
  pages={323--332},
  year={2016},
  publisher={Elsevier}
}

@inproceedings{rahaman2019spectral,
  title={On the spectral bias of neural networks},
  author={Rahaman, Nasim and Baratin, Aristide and Arpit, Devansh and Draxler, Felix and Lin, Min and Hamprecht, Fred and Bengio, Yoshua and Courville, Aaron},
  booktitle={International conference on machine learning},
  pages={5301--5310},
  year={2019},
  organization={PMLR}
}

@article{wang2021eigenvector,
  title={On the eigenvector bias of {F}ourier feature networks: {F}rom regression to solving multi-scale {PDE}s with physics-informed neural networks},
  author={Wang, Sifan and Wang, Hanwen and Perdikaris, Paris},
  journal={Computer Methods in Applied Mechanics and Engineering},
  volume={384},
  pages={113938},
  year={2021},
  publisher={Elsevier}
}

@article{moon2025physics,
  title={Physics-informed neural operators for generalizable and label-free inference of temperature-dependent thermoelectric properties},
  author={Moon, Hyeonbin and Lee, Songho and Demeke, Wabi and Ryu, Byungki and Ryu, Seunghwa},
  journal={npj Computational Materials},
  volume={11},
  number={1},
  pages={272},
  year={2025},
  publisher={Nature Publishing Group UK London}
}

@article{zhao2025physics,
  title={Physics-informed neural networks for solving inverse problems in phase field models},
  author={Zhao, BR and Sun, DK and Wu, H and Qin, CJ and Fei, QG},
  journal={Neural Networks},
  pages={107665},
  year={2025},
  volume={190},
  publisher={Elsevier}
}

@article{HANNA2024108019,
title = {A self-supervised learning framework based on physics-informed and convolutional neural networks to identify local anisotropic permeability tensor from textiles {2D} images for filling pattern prediction},
journal = {Composites Part A: Applied Science and Manufacturing},
volume = {179},
pages = {108019},
year = {2024},
issn = {1359-835X},
author = {John M. Hanna and José V. Aguado and Sebastien Comas-Cardona and Yves {Le Guennec} and Domenico Borzacchiello},
}

@article{hanna2022residual,
  title={Residual-based adaptivity for two-phase flow simulation in porous media using physics-informed neural networks},
  author={Hanna, John M and Aguado, Jose V and Comas-Cardona, Sebastien and Askri, Ramzi and Borzacchiello, Domenico},
  journal={Computer Methods in Applied Mechanics and Engineering},
  volume={396},
  pages={115100},
  year={2022},
  publisher={Elsevier}
}

@article{haghighat2023constitutive,
  title={Constitutive model characterization and discovery using physics-informed deep learning},
  author={Haghighat, Ehsan and Abouali, Sahar and Vaziri, Reza},
  journal={Engineering Applications of Artificial Intelligence},
  volume={120},
  pages={105828},
  year={2023},
  publisher={Elsevier}
}

@article{liu2024config,
  title={Config: Towards conflict-free training of physics informed neural networks},
  author={Liu, Qiang and Chu, Mengyu and Thuerey, Nils},
  journal={arXiv preprint arXiv:2408.11104},
  year={2024}
}

@article{wang2025gradient,
  title={Gradient alignment in physics-informed neural networks: A second-order optimization perspective},
  author={Wang, Sifan and Bhartari, Ananyae Kumar and Li, Bowen and Perdikaris, Paris},
  journal={arXiv preprint arXiv:2502.00604},
  year={2025}
}

@article{xu2023transfer,
  title={Transfer learning based physics-informed neural networks for solving inverse problems in engineering structures under different loading scenarios},
  author={Xu, Chen and Cao, Ba Trung and Yuan, Yong and Meschke, G{\"u}nther},
  journal={Computer Methods in Applied Mechanics and Engineering},
  volume={405},
  pages={115852},
  year={2023},
  publisher={Elsevier}
}

@article{yu2020gradient,
  title={Gradient surgery for multi-task learning},
  author={Yu, Tianhe and Kumar, Saurabh and Gupta, Abhishek and Levine, Sergey and Hausman, Karol and Finn, Chelsea},
  journal={Advances in neural information processing systems},
  volume={33},
  pages={5824--5836},
  year={2020}
}

@article{guhring2021approximation,
  title={Approximation rates for neural networks with encodable weights in smoothness spaces},
  author={G{\"u}hring, Ingo and Raslan, Mones},
  journal={Neural Networks},
  volume={134},
  pages={107--130},
  year={2021},
  publisher={Elsevier}
}

@article{klawonn2024domain,
  title={Machine learning and domain decomposition methods ‑ a survey},
  author={Axel Klawonn and Martin Lanser and Janine Weber},
  journal={Computational Science and Engineering},
  volume={1},
  pages={2},
  year={2024}
}

@article{jagtap2020xpinn,
  title={Extended Physics-Informed Neural Networks ({XPINNs}): A Generalized Space-Time Domain Decomposition Based Deep Learning Framework for Nonlinear Partial Differential Equations},
  author={Ameya D. Jagtap and George Em Karniadakis},
  journal={Communications in Computational Physics},
  volume={28},
  number={5},
  pages={2002-2041},
  year={2020}
}

@article{zhu2019surrogate,
title = {Physics-constrained deep learning for high-dimensional surrogate modeling and uncertainty quantification without labeled data},
journal = {Journal of Computational Physics},
volume = {394},
pages = {56-81},
year = {2019},
author = {Zhu, Y. and Zabaras, N. and Koutsourelakis, P. and Perdikaris, P},
}

@article{tancik2020fourier,
  title={Fourier features let networks learn high frequency functions in low dimensional domains},
  author={Tancik, Matthew and Srinivasan, Pratul and Mildenhall, Ben and Fridovich-Keil, Sara and Raghavan, Nithin and Singhal, Utkarsh and Ramamoorthi, Ravi and Barron, Jonathan and Ng, Ren},
  journal={Advances in neural information processing systems},
  volume={33},
  pages={7537--7547},
  year={2020}
}

@article{lagaris1998artificial,
  title={Artificial neural networks for solving ordinary and partial differential equations},
  author={Lagaris, Isaac E and Likas, Aristidis and Fotiadis, Dimitrios I},
  journal={IEEE transactions on neural networks},
  volume={9},
  number={5},
  pages={987--1000},
  year={1998},
  publisher={IEEE}
}

@article{folland1997uncertainty,
  title={The uncertainty principle: a mathematical survey},
  author={Folland, Gerald B and Sitaram, Alladi},
  journal={Journal of Fourier analysis and applications},
  volume={3},
  number={3},
  pages={207--238},
  year={1997},
  publisher={Springer}
}

@inproceedings{muller2022notes,
  title={Notes on exact boundary values in residual minimisation},
  author={M{\"u}ller, Johannes and Zeinhofer, Marius},
  booktitle={Mathematical and Scientific Machine Learning},
  pages={231--240},
  year={2022},
  organization={PMLR}
}

@article{bradbury2018jax,
  title={{JAX}: composable transformations of {P}ython+{NumPy} programs},
  author={Bradbury, James and Frostig, Roy and Hawkins, Peter and Johnson, Matthew James and Leary, Chris and Maclaurin, Dougal and Necula, George and Paszke, Adam and VanderPlas, Jake and Wanderman-Milne, Skye and others},
  year={2018}
}

@article{hintermuller2026constrained,
  title={Constrained Neural Parameterization for Optimization in Function Spaces},
  author={Hinterm{\"u}ller, Michael and Ning, Jianfeng},
  journal={arXiv preprint arXiv:2606.00855},
  year={2026}
}
\end{document}